\newif\iffullversion
\fullversiontrue

\iffullversion 
\documentclass[11pt,letterpaper]{article}
\usepackage[margin=1in]{geometry}
\usepackage{amsthm}

\else
\documentclass[envcountsect,envcountsame,orivec,runningheads]{llncs}

\usepackage{amsthm}
\fi

\usepackage{cite}

\usepackage[bookmarks=true,pdfstartview=FitH,colorlinks,linkcolor=darkred,filecolor=darkred,citecolor=darkred,urlcolor=darkred]{hyperref}

\usepackage{multicol}

\usepackage{colortbl}
\usepackage{mdframed}
\usepackage{color,soul}

\usepackage{paralist}
\usepackage{pdfpages}
\usepackage{enumitem}
\usepackage{float}
\usepackage{booktabs}
\usepackage[override]{cmtt}
\usepackage{color,xcolor}
\usepackage{authblk}
\usepackage{graphicx,color,eso-pic}
\usepackage{pgfplots, pgfplotstable}

\usepackage[
	lambda,
	advantage,
	operators,
	sets,
	adversary,
	landau,
	probability,
	notions,
	logic,
	ff,
	mm,
	primitives,
	events,
	complexity,
	asymptotics,
	keys]
{cryptocode}

\usepackage{tikz}
\usepackage{pgfplots}
\pgfplotsset{compat=newest}
\usepgfplotslibrary{fillbetween}
\usetikzlibrary{fit, shapes, positioning, patterns, arrows,shadows}
\usetikzlibrary{decorations.pathreplacing,calligraphy}

\tikzset{
    master/.style={
        execute at end picture={
            \coordinate (lower right) at (current bounding box.south east);
            \coordinate (upper left) at (current bounding box.north west);
        }
    },
    slave/.style={
        execute at end picture={
            \pgfresetboundingbox
            \path (upper left) rectangle (lower right);
        }
    }
}

\sethlcolor{lightgray}

\newcommand{\mcal}[1]{\ensuremath{\mathcal {#1}}}

\newcommand{\algA}{{\ensuremath{\mcal{A}}}\xspace}

\newcommand{\A}{{\cal A}}
\renewcommand{\S}{{\cal S}}

\newcommand{\F}{{\mathbb F}}

\renewcommand{\ppt}{\ensuremath{{\text{p.p.t.}}}\xspace}

\newcommand{\zkb}{\ensuremath{{\sf IdealZK}}}

\newcommand{\R}{\ensuremath{\mathbb{R}}}

\newcommand{\E}{\mathbf{E}}

\newcommand{\eps}{\epsilon}

\newcommand{\ek}{\ensuremath{{\sf ek}}\xspace}

\definecolor{darkgreen}{rgb}{0,0.5,0}
\definecolor{lightblue}{RGB}{0,176,240}
\definecolor{darkblue}{RGB}{0,112,192}
\definecolor{lightpurple}{RGB}{124, 66, 168}
\definecolor{grey}{RGB}{139, 137, 137}
\definecolor{maroon}{RGB}{178, 34, 34}
\definecolor{green}{RGB}{34, 139, 34}
\definecolor{types}{RGB}{72, 61, 139}
\definecolor{gold}{rgb}{0.8, 0.33, 0.0}
\definecolor{mygreen}{HTML}{588D6A}
\definecolor{myred}{HTML}{C86733}
\definecolor{myblue}{HTML}{5B68FF}
\definecolor{myshadow}{HTML}{E6C5B4}

\definecolor{darkgray}{gray}{0.3}

\newcommand{\skiptext}[1]{}

\newcommand{\getr}{\ensuremath{{\overset{\$}{\leftarrow}}}\xspace}

\newcommand{\secu}{\lambda}

\usepackage{boxedminipage}
\usepackage{url}
\usepackage{graphicx}
\usepackage[bf,justification=centering]{caption}
\usepackage[justification=centering]{subcaption}
\usepackage{color}
\usepackage{xspace}
\usepackage{multirow}

\usepackage{amsthm,amstext}
\usepackage{amsmath,amstext,amsfonts,amssymb,latexsym}
\usepackage{amsbsy}
\usepackage{stmaryrd}
\usepackage{pifont}

\usepackage[capitalize,noabbrev]{cleveref} 

\newlength{\Oldarrayrulewidth}

\definecolor{darkred}{rgb}{0.5, 0, 0}
\definecolor{darkgreen}{rgb}{0, 0.5, 0}
\definecolor{darkblue}{rgb}{0,0,0.5}

\newcommand\markx[2]{}

\renewcommand{\path}{\ensuremath{\mathsf{path}}\xspace}

\renewcommand{\negl}{{\sf negl}}

\newcommand{\N}{\ensuremath{\mathbb{N}}\xspace}

\newcommand{\ignore}[1]{}

\renewcommand{\part}{\ensuremath{\mcal{S}}\xspace}

\newcounter{task}

\newtheorem{thm}{Theorem}[section]      

\newtheorem{theorem}[thm]{Theorem}

\newtheorem{conjecture}[thm]{Conjecture}
\newtheorem{lemma}[thm]{Lemma}
\newtheorem{claim}[thm]{Claim}
\newtheorem{corollary}[thm]{Corollary}
\newtheorem{fact}[thm]{Fact}
\newtheorem{proposition}[thm]{Proposition}
\newtheorem{example}[thm]{Example}

\theoremstyle{definition}
\newtheorem{definition}[thm]{Definition}
\newtheorem{remark}[thm]{Remark}

\newtheoremstyle{boxes}
{2pt}
{0pt}
{}
{}
{\bfseries}
{}
{\newline}
{\thmname{#1}\thmnumber{ #2}:  
\thmnote{#3}}

\theoremstyle{boxes}

\newcommand{\bids}{{\bf b}}

\newcommand{\mecha}{\ensuremath{({\bf x}, {\bf p}, \mu)}}
\newcommand{\fmec}{\ensuremath{{\mcal{F}_{\rm MPC}}}}
\newcommand{\pmec}{\ensuremath{{\Pi_{\rm MPC}}}}

\newcounter{cnt:challenge}

\newcommand{\bfb}{{\bf b}}
\newcommand{\bfc}{{\bf c}}
\newcommand{\bfd}{{\bf d}}

\newcommand{\bfL}{{\bf L}}
\newcommand{\bfm}{{\bf m}}

\newcommand{\otherdist}{\mcal{D}_{-i}}

\DeclareMathOperator{\EXPT}{{\bf E}}

\DeclareMathOperator*{\myargmax}{arg\,max}

\newcommand{\tildemu}{\widetilde{\mu}}
\newcommand{\tildeutil}{\widehat{\sf util}}

\newcommand{\crs}{\ensuremath{{\sf crs}}}
\newcommand{\CRS}{\ensuremath{{\sf CRS}}}
\newcommand{\NIZK}{\ensuremath{{\sf NIZK}}}
\newcommand{\stmt}{\ensuremath{{\sf stmt}}}

\ifdefined\authorNote
\newcommand{\elaine}[1]{{\footnotesize\color{magenta}[Elaine: #1]}}
\newcommand{\hao}[1]{{\footnotesize\color{blue}[Hao: #1 N.K.]}}
\newcommand{\ke}[1]{{\footnotesize\color{red}[Ke: #1 QOD]}}
\else
\newcommand{\elaine}[1]{}
\newcommand{\hao}[1]{}
\newcommand{\ke}[1]{}
\fi

\title{What Can Cryptography Do For Decentralized Mechanism Design?}

\iffullversion
\author{Elaine Shi, \ \ Hao Chung, \ \ and Ke Wu\footnote{Author order is randomized.}}
\date{Carnegie Mellon University \\
{\tt \{runting@cs, haochung@andrew, kew2@andrew\}.cmu.edu}
}

\else
\author{}
\date{}
\institute{}
\fi

\begin{document}

\iffullversion
\begin{titlepage}
\maketitle
\begin{abstract}
Recent works of Roughgarden (EC'21) and 
Chung and Shi (SODA'23)
initiate the study of a new decentralized mechanism design problem
called transaction fee mechanism design (TFM).
Unlike the classical mechanism design
literature, in the decentralized environment, even the auctioneer (i.e., the miner)
can be a strategic player, and it can even 
collude with a subset of the users facilitated by binding side contracts.
Chung and Shi showed two main impossibility results  
that rule out the existence of a {\it dream} TFM.
First, any TFM 
that provides incentive compatibility for individual users
and miner-user coalitions
must always have zero miner revenue, no matter whether the block size 
is finite or infinite.
Second, assuming finite block size, no non-trivial TFM can simultaneously 
provide incentive compatibility 
for any individual user and for any miner-user coalition.

In this work, we explore what new models and meaningful relaxations  
can allow us to circumvent the impossibility results of Chung and Shi.
Besides today's model that does not employ cryptography,
we introduce a new MPC-assisted model 
where the TFM is implemented by a joint multi-party computation
(MPC) protocol among the miners. 
We prove several feasibility and infeasibility results for achieving 
{\it strict} and {\it approximate} incentive compatibility, respectively, 
in the plain model as well as the MPC-assisted model.
We show that while cryptography is not a panacea, it indeed
allows us to overcome some impossibility results pertaining to the plain model, 
leading to non-trivial mechanisms with useful guarantees that are otherwise
impossible in the plain model.
Our work is also the first to characterize the mathematical landscape of 
transaction fee mechanism design under approximate incentive compatibility,
as well as in a cryptography-assisted model.

\end{abstract}
\thispagestyle{empty}
\end{titlepage}

\fi

\iffullversion
\tableofcontents
\clearpage
\fi

\section{Introduction}

The widespread adoption of blockchains and cryptocurrencies
spurred a new class of {\it decentralized} mechanism design problems.
The recent works of Roughgarden~\cite{roughgardeneip1559,roughgardeneip1559-ec} 
as well as Chung and Shi~\cite{foundation-tfm}
considered a particularly important
decentralized mechanism design problem, that is, 
{\it transaction fee mechanism} (TFM) design.
In a transaction fee mechanism (TFM),  
we are auctioning space in the block to users
who want their transactions included and confirmed in the block.
If the block can contain up to $k$ transactions, one can equivalently think
of selling $k$ identical products to the bidders. 

Prior works~\cite{zoharfeemech,yaofeemech,functional-fee-market,eip1559,roughgardeneip1559,roughgardeneip1559-ec,dynamicpostedprice}
observed that transaction fee mechanism design departs
significantly from classical mechanism design~\cite{agt}. 
The vast majority of classical auctions assume that 
the auctioneer honestly implements the prescribed mechanism.
In comparison, in a blockchain environment, 
the auctioneer (i.e., the miner of the block), can be a 
strategic player in itself: it can deviate from the prescribed mechanism  
if it increases its expected gain; or it can collude
with a subset of the users, and play strategically to improve the coalition's joint utility.
As earlier works pointed out~\cite{zoharfeemech,yaofeemech,functional-fee-market,eip1559,roughgardeneip1559,roughgardeneip1559-ec}, the existence
of decentralized smart contracts in blockchain environments
make it easy for the miner 
and users to 
rendezvous and engage in {\it binding} side contracts.
Such side contracts allow the 
coalition to split their gains off the table in a binding fashion.

Observing the new challenges that arise in a decentralized environment,
earlier works~\cite{zoharfeemech,yaofeemech,functional-fee-market,eip1559,roughgardeneip1559,roughgardeneip1559-ec}
formulated 
a set of desiderata for a ``dream'' TFM:
\begin{itemize}[leftmargin=5mm,itemsep=1pt]
\item {\it User incentive compatibility (UIC)}:
a user's best strategy is to bid truthfully, even when the user has observed others' bids. 
\item {\it Miner incentive compatibility (MIC)}:
the miner's best strategy is to implement the honest mechanism, even when
the miner has observed all users' bids.
\item {\it $c$-side-contract-proofness ($c$-SCP)}:
playing honestly maximizes the joint utility of a 
coalition consisting of the 
miner and at most $c$ users, even after having observed all others' bids.
\end{itemize}

A line of works explored how to get a dream TFM.  
However, assuming that the block size is {\it finite}, i.e., there can
be more bids than the block size, 
all known works fall short of achieving all three 
properties at the same time.
The closest we have come to in terms of achieving a dream TFM
is in fact Etherem's EIP-1559.  
At a very high-level, when there is congestion, EIP-1559 behaves
like a first-price auction which is not UIC. 
When the block size is infinite (i.e., no congestion), 
EIP-1559 approximates the following ``burning posted price'' auction:
there is a fixed reserve price $r$, 
every bid that is at least $r$ gets included and confirmed,
and pays the price of $r$. All users' payment is burnt
and the miner gets nothing\footnote{In practice, the miner gets
a fixed block reward that is irrelevant to our game-theoretic analysis,
so we ignore the fixed block reward in our modeling.}. 
Roughgarden~\cite{roughgardeneip1559,roughgardeneip1559-ec}
proved that when the block size is {\it infinite}, indeed, 
the burning posted price auction 
achieves all three properties at the same time!

Subsequently, Chung and Shi~\cite{foundation-tfm} further explored
the landscape of TFM. They proved two interesting impossibility results:
\begin{enumerate}[leftmargin=6mm]
\item 
{\it Zero miner revenue.}
Any (possibly randomized) TFM 
that satisfies both UIC and SCP must always have 0 miner revenue,  
even when the miner colludes with at most one user, and 
no matter whether the block size is finite or infinite.
This shows that the total burning in 
EIP-1559 is no accident: it is necessary 
to achieve all three properties under infinite block size.
\item 
{\it Finite-block impossibility.}
Suppose that block size is finite, then no non-trivial (possibly randomized) TFM
can achieve UIC and SCP at the same time, even when the miner colludes with at most
one user.
This shows that it is no accident that all prior works fail
to achieve the dream TFM for finite block sizes --- indeed, there is 
a mathematical impossibility!
\end{enumerate}

Given the status quo of our understanding, we ask the following natural question:
\begin{itemize}[leftmargin=5mm]
\ignore{
\item
{\it Is EIP-1559 (or burning posted price auction) the only possible
mechanism that achieves all three properties at the same time
assuming infinite block size?}
In particular, the question is relevant because real-world
blockchains such as Ethereum may adopt
a dynamic adjustment mechanism that chooses the base fee of each block
based on the recent history, such that only those that bid at least
the base fee are considered. The goal of such a dynamic adjustment mechanism
is to avoid congestion, such that we can be in the ``infinite block size''
regime most of the time.
}
\item[] 
{\it Are there meaningful new models or relaxations 
that allow us to circumvent the impossibility results 
of Chung and Shi?}
\end{itemize}

Chung and Shi~\cite{foundation-tfm} 
made an initial exploration along this line.
They show a relaxation that 
allows us to circumvent the impossibilities
and achieve positive miner revenue under finite block size.
In particular, their relaxation requires the additional assumption that 
offending bids (e.g., overbid or fake transactions) that have been posted 
to the public cannot be retracted in the future, and thus the offender
may have to pay a cost when the offending transaction is confirmed in the future.
While this assumption holds for some cryptocurrencies such as Bitcoin, it may
not be universally true for all cryptocurrencies. 
Therefore, an important question is what other models or relaxations allow
us to circumvent the impossibilities. 

In this paper, we explore two new directions, aiming to understand
whether they allow us to circumvent the impossibilities 
of Chung and Shi~\cite{foundation-tfm}:
{\it i)} using an approximate notion of incentive compatibility 
that allows an $\epsilon$ additive slack; 
and {\it ii)} having the miners jointly run a multi-party computation (MPC) protocol
to realize the TFM. 
Throughout the paper, we refer to the today's model, which does 
employ cryptography,
as the {\it plain model}, 
and 
we refer to the case where the TFM is realized with MPC
as the {\it MPC-assisted model}.

\subsection{Our Results and Contributions}

Our paper makes novel contributions
at both {\it conceptual} and {\it technical} levels. 
From a technical perspective, 
prior to our work, we lacked techniques 
for characterizing the solution space 
of approximate incentive compatibility --- in particular, classical
tools like Myerson's Lemma~\cite{myerson} 
breaks down when we allow $\epsilon$ slack
in the incentive compatibility, and thus 
our classical insights often fail.
One of our main technical contributions is to develop new techniques 
for mathematically reasoning about approximate incentive compatibility.
On the conceptual front, 
while an elegant line of work has shown ways in which 
cryptography
and game theory can help each other~\cite{gtcrypto00,gtcrypto01,gtcrypto02,gtcrypto03,giladutilityindjournal,giladgtcrypto,rdp00,rdp01,rdp02,katzgametheory,gtcrypto06,seqrationalcrypto,gt-fair-cointoss,gt-fair-coin-complete,gt-leader-shi,fruitchain,logstar-gt-leader,credibleauction-comm00,credibleauction-comm01} (see \Cref{sec:related}
for more discussions),  
our work is of a different nature. Our results reveal
exciting new connections between cryptography and mechanism design, 
motivated by a practical problem. 
The popularity of blockchains and decentralized applications
poses many exciting new challenges 
for decentralized mechanism design, 
and cryptography-meets-game-theory
is a natural and promising paradigm.
We thus hope that our new conceptual contributions
can provide fodder and inspire new works in
this exciting and much explored space.

We give a summary of our main results below.

\subsubsection{Characterizing Miner Revenue 
under Approximate Incentive Compatibility}
We first focus on the plain model that was studied
in earlier works~\cite{zoharfeemech,yaofeemech,functional-fee-market,roughgardeneip1559,roughgardeneip1559-ec,dynamicpostedprice,foundation-tfm}.
Recall that assuming infinite block size, it is possible
to achieve a dream TFM (e.g., the burning posted price auction), 
but the miner revenue has to be zero.
We ask the following question:   
{\it suppose we are willing to 
relax the incentive compatibility notion and allow 
an $\epsilon$ additive slack, 
can we circumvent the 
zero miner revenue lower bound? 
If so, exactly how much miner revenue can we hope for? }

More specifically, $\epsilon$-incentive-compatibility (including
$\epsilon$-UIC, $\epsilon$-MIC, and $\epsilon$-SCP) requires
that any deviation  
cannot increase the strategic individual or coalition's utility by more than $\epsilon$.
We show that under $\epsilon$-incentive-compatibility,
we can achieve linear (in the number of users) miner revenue
assuming infinite block size.
Moreover, we give matching upper- and 
lower-bounds that tightly characterize exactly how much miner revenue
can be attained.


\paragraph{Infinite block size.}
Consider the simple 
posted price auction with reserve price 
$r \leq \frac{\epsilon}{c}$
where $c$ is the maximum number of users controlled by the strategic coalition:
all bids that bid at least $r$ are confirmed. Each confirmed bid
pays $r$. All payment goes to the miner. 
It is not hard to show that the above auction satisfies 
strict UIC, strict MIC (for an arbitrarily sized miner-coalition), 
and $\epsilon$-SCP against $c$-sized coalitions. 
Further, the expected total miner revenue 
is $\Theta(n \cdot \frac{\epsilon}{c})$ when the users' true
values are not too small.
\ignore{
In particular, 
the term $\min(\frac{\epsilon}{c}, m)$ 
says that 
1) when $c$ 
is large, 
the miner can get only $\frac{\epsilon}{c}$ from each confirmed bid;
and 2) 
when $c$ is small, 
the miner can get as much as $\Theta(m)$ on average from each confirmed bid, i.e.,
often as high as the bids themselves.
}

\hao{TODO: where to put the proof of posted-price auction?}

Although the above posted price achieves linear in $n$ revenue,
the drawback is that the miner revenue is unscalable: 
even as the users' bids
scale up (e.g., by some multiplicative factor), the miner revenue does
not grow proportionally. 
We therefore ask if randomization can help achieve scalability
in miner revenue.  We show that indeed
the following randomized TFM 
achieves scalability in miner revenue:
\ignore{
that 
allows us to achieve 
higher revenue, that is, $\Theta(n \cdot 
\min(\sqrt{\frac{m \cdot \epsilon}{c}}, m))$ where $m$ is the median
of the distribution $\mcal{D}$.
Consider the following MPC-assisted {\it proportional auction}
where the reserve price $r$ is set to the median $m$ of the 
distribution $\mcal{D}$:
}

\begin{mdframed}
\underline{\it Proportional auction}
\hfill // Let $r$ be a fixed reserve price. 
\begin{itemize}[leftmargin=5mm]
\item 
Every bid $b\geq r$ is confirmed with probability $1$ and every 
candidate bid $b<r$ is confirmed with probability $b/r$. 
Each confirmed 
bid $b$ pays $p=\min\{\frac{b}{2}, \frac{r}{2}\}$. 
\item For each confirmed bid, miner gets a  pre-determined threshold $r' = \sqrt{\frac{2 r \epsilon}{9c}}$ if $p\geq r'$.
\end{itemize}
\end{mdframed}

For example, suppose all users' bids are sampled
independently from some distribution $\mcal{D}$, and let
$m$ be the median of the distribution 
such that $\Pr_{x \sim \mcal{D}}[x \geq m] \geq 1/2$ (or
any other constant).
Then, if we set $r = m$, the expected miner revenue 
(taken over the randomness of users' bids 
as well as of the TFM itself) 
is $\Omega(n \cdot \min(m, \sqrt{\frac{m \epsilon}{c}}))$.

\ignore{
Suppose $r$ is set to the median $m$ of the distribution $\mcal{D}$. 
Informally speaking,
when ${c}$ is  large, the miner can earn
$\sqrt{\frac{m \cdot \epsilon}{c}}$
from each confirmed bid.
When ${c}$ is small,
the miner can earn
up to $\Theta(m)$ on average from each confirmed bid, i.e., often as high
the bid itself.
}

Combining the posted price auction and the proportional auction,
we have the following theorem:
\begin{theorem}
Consider the hybrid auction
which, given some bid distribution $\mcal{D}$ with median $m$, 
runs either the posted 
posted price auction with reserve price $r = \min(\frac{\epsilon}{c}, m)$
or the proportional auction with the reserve price $r = m$, depending
on which one has higher expected revenue.
The hybrid auction 
is strict UIC, strict MIC (for an arbitrarily sized miner coalition), 
and $\epsilon$-SCP against
any miner-user coalition with at most $c$ users. 
Further, it 
achieves $\Omega\left(n \cdot(\min(\frac{\epsilon}{c} + 
\sqrt{\frac{m\epsilon}{c}}, m) ) \right)$  
expected total miner revenue. 
\end{theorem}

\ignore{
Consider an MPC-assisted hybrid mechanism 
that either runs the 
posted price auction with reserve 
price $r = \min(\frac{\epsilon}{c}, m)$
or the proportional auction 
with reserve price $r = m$, depending on 
whether $\frac{\epsilon}{c} < 1$ --- we then have the following theorem:
}

Next, we prove a matching bound that shows the limitation on how 
much miner revenue can be attained under approximate 
incentive compatibility, 
as stated in the following theorem --- this bound holds
no matter whether the block size is finite or infinite. 

\begin{theorem}[Limit on miner revenue for infinite block size]
For any possibly randomized TFM (in the plain model) that satisfies $\epsilon$-UIC,
$\epsilon$-MIC, 
and $\epsilon$-SCP for miner-user coalitions with $1$ user, 
the expected total miner revenue over a random bid vector sampled from $\mcal{D}^n$
must be upper bounded by
\[
\E_{\bids\sim\mcal{D}^n}\left[\mu({\bf b})\right]
\leq  6 n \cdot ( \epsilon  + 
\sqrt{\epsilon} \cdot \E_{x \sim \mcal{D}}[\sqrt{x}]),
\]
where $\mu({\bf b})$ denotes the total miner revenue under
the bid vector ${\bf b}$,
$n$ is the number of users,
$\mcal{D}_i$ denotes the true value distribution of user $i \in [n]$.
\elaine{TODO: add the min of the block size}

\hao{I removed the result about deterministic mechanism.}
\label{thm:intro-minerrev-eps}
\end{theorem}


\ignore{
Another natural question is: {\it can 
we circumvent 
the finite-block impossibility under 
approximate incentive compatibility?}
Unfortunately, we show that the answer is negative 
if there is no upper bound on the bids. 
We state this result in the following theorem,
which can be regarded as a further generalization of the finite-block
impossibility result of Chung and Shi~\cite{tfm-foundation}.

\begin{theorem}[Finite-block impossibility even under approximate incentive compatibility]
\elaine{FILL}
\label{thm:intro-finite-imp-approx}
\end{theorem}
}

\paragraph{Finite block size.}
Another natural question is: {\it can we circumvent 
the finite-block impossibility under 
approximate incentive compatibility?}
Unfortunately, 
although it is indeed 
possible to overcome the finite-block impossibility
with approximate incentive compatibility,
we prove a new impossibility 
result that rules out the existence of 
``useful'' mechanisms 
whose social welfare (i.e., the sum of everyone's utilities)
scales up proportionally w.r.t. the bid distribution:

\begin{theorem}[Scalability 
barrier 
for approximate incentive compatibility in the plain model] 
Fix any $\epsilon > 0$, and suppose that the block size is $k$.
Any (possibly random) TFM in the plain model 
that simultaneously satisfies $\epsilon$-UIC, $\epsilon$-MIC, 
and $\epsilon$-SCP
(even when the miner colludes with at most one user)
has at most $\widetilde{O}(k^3 \epsilon)$ social welfare where $k$
is the block size and $\widetilde{O}(\cdot)$
hides logarithmic factors. \elaine{double check}
\ignore{
Further, if we assume that  
bids can be unbounded, then, no (possibly random)
TFM can simultaneously satisfy $\epsilon$-UIC and $\epsilon$-SCP,
even when the miner colludes with at most one user.
}
\label{thm:intro-finite-imp-approx}
\end{theorem}


\subsubsection{Can We Circumvent the Finite-Block Impossibility with Cryptography?}

Due to the negative result of \Cref{thm:intro-finite-imp-approx},
we want to seek other avenues that allow us to circumvent
the finite-block impossibility. 
Since cryptography is widely deployed 
in today's blockchains, it is natural to ask whether we can bring cryptography
to the design of transaction fee mechanisms, to help 
us achieve what is otherwise impossible.

\paragraph{New model: MPC-assisted TFM.}
Consider a scenario
henceforth called the MPC-assisted model, where 
a set of miners jointly run a multi-party computation (MPC)
protocol to implement the TFM. 
One may think of the MPC protocol as providing the following
ideal functionality $\mcal{F}_{\rm TFM}$:
\begin{itemize}[leftmargin=5mm,itemsep=1pt]
\item 
Each player (either user or miner) may act as any number of identities (including 0),
and on behalf of each identity, submit a bid to $\mcal{F}_{\rm TFM}$.
\item 
The ideal functionality 
$\mcal{F}_{\rm TFM}$ executes the prescribed {\it allocation rule} of the TFM 
to decide which transactions to include and confirm in the block;
it executes the {\it payment rule} and {\it miner revenue rule} of the TFM
to decide how much each confirmed bid pays and the total miner revenue.
$\mcal{F}_{\rm TFM}$ 
then sends to all players  
the set of bids that are confirmed, what price each confirmed
bid pays, and the total miner revenue.
\end{itemize}
We require that the total miner revenue does not exceed the total payment,
and that the total miner revenue is split among the miners.

We assume that 
there is a separate process to decide the 
set of miners whose job is to jointly run 
the MPC protocol. For example, this decision
can be made through either proof-of-work or proof-of-stake.
In the former case, the total miner revenue
is effectively split among the miners proportional to their mining power.
In the latter case, 
the total miner revenue
is effectively split among the miners proportional to their stake.

We assume that  
the majority of the miners are honest   
and that the MPC provides guaranteed output (i.e., the strategic
miners cannot 
cause the MPC protocol to abort without producing outcome).
Note that if we can indeed design an incentive compatible protocol
in the MPC-assisted model, then, no miner would be incentivized
to deviate from the honest protocol, and this reinforces
the honest majority assumption. 
We discuss how to extend our results to the  
setting of {\it majority-miner coalitions} in \Cref{rem:majoritycorrupt}.

\ignore{
Henceforth, we call this new model ``MPC-assisted TFM'', and
we use ``the plain model'' to refer to  
the setting in earlier works~\cite{} where bids are broadcast in the clear.
}
Intuitively, an MPC-assisted TFM
restricts the strategy space for players in comparison with the plain model:
\begin{itemize}[leftmargin=8mm]
\item[\underline{R1}] A strategic individual or coalition must decide its strategy 
without having seen honest users' bids
({\it c.f.} in the plain model, a strategic individual or coalition can decide
their strategy after seeing other players' bids).
\item[\underline{R2}]
Once the set of bids are committed to, the allocation rule
must be implemented honestly ({\it c.f.} in the plain model, 
the winning miner or block proposer 
can strategically choose which transactions to include in the block).
\end{itemize}

Exactly because of the MPC-assisted model imposes
the above restrictions on the strategy space, 
we are hopeful that 
it may allow us to circumvent impossibilities.
Before we explain our results, we first discuss
how to define 
incentive compatibility in the MPC-assisted model.

\begin{remark}[On the practicality of MPC]
We start by assuming generic MPC, since this is a good starting  
point as an initial feasibility exploration.
All the impossibility results in our paper hold even with generic MPC.
However, for all the MPC-assisted mechanisms
we propose, 
although we initially describe the feasibility results 
using generic MPC for conceptual simplicity, 
it turns out that we actually do not need generic MPC to actually instantiate
these mechanisms. 
We discuss how to efficiently instantiate
our MPC-assisted mechanisms in \Cref{sec:efficientMPC}.
\label{rem:efficientmpc}
\end{remark}

\begin{remark}[Extending our results to majority-miner coalitions] 
All the results in the paper actually 
hold even when a coalition may control the majority of miners. 
When the majority of the miners may be malicious, the MPC protocol
cannot provide guaranteed output, it can only provide
``security with abort''. In other words, the 
ideal functionality that is realized by the MPC now  
provides the following backdoor: an adversary 
controlling the majority of miners can send $\bot$ to the ideal
functionality, which causes the protocol to abort
and not produce any output.

Threfore, 
if we assume that the coalition can control the majority of miners,
essentially the strategy space includes one more move: the strategic
coalition can cause the protocol to abort in which case no block is mined,
and no on obtains any utility.
Obviously, a rational coalition should never make such a move.
\label{rem:majoritycorrupt}
\end{remark}

\paragraph{Ex post vs. Bayesian notions of incentive compatibility.}
In the plain model, because a strategic individual or coalition can decide their bids
after seeing others' bids, prior works~\cite{roughgardeneip1559-ec,foundation-tfm} 
considered an {\it ex post}
notion of incentive compatibility. 
In the new MPC-assisted model, since players must submit their bids
to $\mcal{F}_{\rm TFM}$ without seeing others' bids, it also makes
sense to consider 
a {\it Bayesian} notion of equilibrium. 

Informally, we say that an MPC-assisted TFM satisfies 
{\it Bayesian Nash Equilibrium (BNE)} for a
strategic coalition (or individual) $\mcal{C}$, 
following the honest strategy allows $\mcal{C}$ to maximize
its expected gain, 
assuming that the bids
of users not in $\mcal{C}$ are drawn independently from some known distribution.
If the coalition $\mcal{C}$ consists
of an individual user, we say that the scheme satisfies {\it Bayesian UIC}.
When $\mcal{C}$ consists of at most $\rho$ fraction of the miners, we say
that the scheme satisfies {\it Bayesian MIC} against a $\rho$-sized miner-coalition,
Finally, when the coalition $\mcal{C}$ consists of at most $\rho$ fraction of 
miners as well as at least $1$ and at most $c$ users, we say
that the scheme satisfies 
{\it Bayesian SCP} against a $(\rho, c)$-sized coalition.

Jumping ahead, 
for the MPC-assisted model, all our mechanism designs  
achieve incentive compatibility even in the {\it ex post}
setting --- in other words,  
the incentive compatibility
guarantees hold even if $\mcal{F}_{\rm TFM}$
leaks other players' bids to the strategic players before
they decide their own strategy.
On the other hand, all of our impossibilities 
hold even for the Bayesian setting.
This makes both our upper- and lower-bounds stronger.

\elaine{i'm using strict rather than exact}

\paragraph{MPC-assisted TFM under strict incentive compatibility.}
Unfortunately, as shown in~\cref{sec:mpc}, 
the MPC-assisted model does not help us
circumvent the zero miner revenue lower bound, even for Bayesian notions of equilibrium.
Instead, the main question we care 
about here is {\it whether the MPC-assisted model
allows us to circumvent the finite-block impossibility.}
It turns out that the answer 
is not a simple binary one.

First, we show that absent user-user collusion, 
we can indeed circumvent the strong finite-block impossibility of 
Chung and Shi~\cite{foundation-tfm}.
Specifically, we can indeed construct
a TFM that simultaneously achieves UIC, MIC, and $(\rho, c = 1)$-SCP 
for any $\rho$.
In particular, consider the following 
{\it posted price auction with random selection} --- recall that to specify
an MPC-assisted TFM, we only need to specify the allocation rule, 
the payment and miner revenue rules.

\begin{mdframed}
\noindent \underline{\it MPC-assisted, posted price auction with random selection}

\vspace{5pt} \noindent
Let $r$ be a fixed reserve price. 
Any bid that is at least $r$
is considered as a candidate. 
Randomly choose up to block size $k$ candidates to confirm. 
Any confirmed bid pays $r$. All payment is burnt and the miner revenue is $0$.
\end{mdframed}

\Cref{sec:efficientMPC}
describes how to instantiate the above 
MPC-assisted mechanism efficiently without using generic MPC.

\begin{theorem}[MPC-assisted, posted price auction with random selection]
The above MPC-assisted, posted price auction with random selection 
satisfies UIC, MIC, and $(\rho, 1)$-SCP in the ex post setting for an arbitrary
$\rho \in [0, 1]$.
\label{thm:intro-finite-posted}
\end{theorem}

Since \Cref{thm:intro-finite-posted} holds
even in the ex post setting, another interpretation
is that the enforcement of the allocation rule
(i.e., restriction R2, and not R1)
is what allows us to circumvent the finite-block impossibility  
when $c = 1$.

The above posted price auction with random selection works for $c = 1$, i.e. no user-user collusion; however,
it fails when the coalition may contain $c \geq 2$ users.
Imagine that the number of users $n = k + 1$, 
and the coalition consists of two users and any fraction of miners.
Now, suppose one of the colluding users 
has true value $v \gg r$, 
and the other has true value $v' = r$.
In this case, the user with true value $v' = r$ 
should simply drop out and not submit a bid. This guarantees 
that the friend with large true value 
will be confirmed, and thus the coalition's joint utility increases.   

It turns out that this is no accident.   
We prove that for $c \geq 2$, 
no MPC-assisted TFM
can achieve UIC, MIC, and SCP for $(\rho, c)$-sized coalitions at the 
same time for any choice of $\rho$. 
Further, the impossibility holds even assuming Bayesian notions
of incentive compatibility.

\begin{theorem}[Finite-block impossibility in the MPC-assisted model for $c \geq 2$]
Let $c \geq 2$ and let $\rho \in [0, 1]$.
No (possibly randomized) MPC-assisted TFM
with non-trivial utility can simultaneously 
achieve Bayesian UIC, Bayesian MIC, and Bayesian SCP for $(\rho, c)$-sized coalitions, 
assuming finite block size.
\label{thm:intro-finite-imp-mpc}
\end{theorem}

\begin{table}[H]
\caption{{\bf Mathematical landscape of TFM}. Results in \colorbox{blue!22}{blue background}
are shown in this paper. \ding{55} means impossible
and \ding{51} means possible.
$\Theta(\cdot)$ means that we show matching upper and lower bounds ---
here $m$ is a term that depends on the scale of the bid distribution,
and we ignore terms related to $c$ for simplicity. 
Unless otherwise noted,
the impossibilities hold even for $c = 1$.}
\label{tab:results}
\vspace{-12pt}
\begin{center}
\begin{tabular}{cccc}
\toprule
                                            &  & {\bf plain model} & {\bf MPC-assisted model}\\
\midrule
\multirow{2}{*}{\bf Infinite block} & strict & 0 miner rev~\cite{foundation-tfm}& \cellcolor{blue!22} {0 miner rev}\\[2pt]
\arrayrulecolor{gray}\cmidrule{2-4}
                     & approximate & \cellcolor{blue!22} $\Theta(n\cdot(\epsilon + \sqrt{m \epsilon}))$ miner rev & \cellcolor{blue!22} $\Theta(n\cdot(\epsilon + \sqrt{m \epsilon}))$ miner rev\\
\arrayrulecolor{black}
\midrule
\multirow{3}{*}{\bf Finite block} & strict  & \ding{55}~\cite{foundation-tfm} &\cellcolor{blue!22}
\ding{51}: $c = 1$, \quad  \quad
\ding{55}: $c \geq 2$\\[3pt]
\arrayrulecolor{gray}\cmidrule{2-4}
& \multirow{1}{*}{approximate} & \cellcolor{blue!22} scalability \ding{55}
(ignoring log terms)
& 
\cellcolor{blue!22}
{scalability \ding{51}}
\\
\arrayrulecolor{black}
\bottomrule
\end{tabular}
\end{center}
\end{table}
\elaine{note that the infinite block size miner rev lower 
bound holds in mpc assisted model}

\paragraph{MPC-assisted TFM under approximate incentive compatibility.}
\ignore{
Given the finite-block impossibility in the MPC-assisted model
for $c \geq 2$ (\Cref{thm:intro-finite-imp-mpc}), 
we further ask: {\it suppose
we are willing to both adopt the MPC-assisted model
and 
allow $\epsilon$ additive slack, can we circumvent this 
finite-block impossibility
for $c\geq 2$? }
}
Recall that in the plain model, even with approximate incentive 
compatibility, we cannot have scalable TFMs whose social welfare
scales w.r.t. the bid distribution (\Cref{thm:intro-finite-imp-approx}).
We show that if we consider approximate incentive compatibility
in the MPC-assisted model, we can overcome this scalability barrier.
Specifically, we construct an MPC-assisted TFM called the 
``diluted posted price auction''
that can achieve up to $\Theta(M \cdot k)$ 
social welfare when many people's bids are large enough,
where $M$ is an upper bound on users' bid.

\begin{mdframed}
\noindent \underline{\it MPC-assisted, diluted posted price auction}
\begin{itemize}[leftmargin=5mm,itemsep=1pt]
\item 
Let $r$ be a fixed reserve price, let $M$ be the 
maximum possible value of the bid, and let $k$ be the block size.
\item 
Remove all bids that are less than $r$, and suppose that there are $\ell$ bids left ---
these bids form the candidate pool.
\item 
Let $N = \max\{c \cdot \sqrt{\frac{k M}{2 \epsilon}}, k\}$. 
If $\ell < N$, pad
the candidate pool with fake $0$ bids 
such that its size is $N$. 
\item 
Choose $k$ bids at random from the candidate pool. All real bids chosen
are confirmed and pay the reserve price $r$.
\item 
The miner gets 
$\frac{2\epsilon}{c}$ for each confirmed bid.
\end{itemize}
\end{mdframed}
\Cref{sec:efficientMPC}
describes how to instantiate the above 
MPC-assisted mechanism efficiently without using generic MPC.

In the above mechanism, suppose we set the reserve
price $r \leq M/2$, and further,  
imagine that 
everyone's true value is $M$, and they all bid their true value.
Further, assume that there are many more users than the block size $k$.
In this case, the block will be filled with $k$ confirmed bids, 
and for each confirmed 
bid obtains utility $M/2$. Thus, we can achieve 
$\Theta(M \cdot k)$ social welfare.

\ignore{
In the above mechanism, suppose that the maximum possible bid $M$
is sufficiently large, and suppose that $r \leq M/2$. 
Then, a user who bids at least $M/2$ has probability of 
$\Theta(\frac1c \cdot \sqrt{\frac{k \epsilon}{M}})$
of being confirmed, and thus its expected utility 
is $\Theta(\frac1c \cdot \sqrt{k \epsilon M})$.
The theorem below states that 
the above MPC-assisted, diluted posted price auction
satisfies strict UIC and strict MIC, and $\epsilon$-SCP.
}

\begin{theorem}[MPC-assisted, diluted posted price auction]
The above MPC-assisted, diluted posted price auction 
satisfies strict UIC, strict MIC, and $\epsilon$-SCP for $(\rho, c)$-sized
coalitions in the ex post setting, for any choice of $\rho$ and $c$. 
Further, the mechanism is scalable, i.e., it 
can achieve $\Theta(M \cdot k)$ 
expected social welfare under some bid configurations. 
\end{theorem}

\elaine{can we achieve positive miner revenue}

\paragraph{Summary of landscape.} 
Summarizing our understanding so far, we present 
the mathematical landscape
of TFM in Table~\ref{tab:results}.
Our results show that cryptography 
can help us  
circumvent fundamental impossibilities
of the plain model under finite block size. 
First, for strict incentive compatibility, 
cryptography allows us to 
overcome the finite-block impossibility
for $c=1$ (\Cref{thm:intro-finite-posted}). 
Second,   
with approximate incentive compatibility,
cryptography allows us to overcome the 
scalability barrier for finite block size in the plain model. 

On the other hand, 
cryptography is also not a panacea.
For example, it does not 
fundamentally help us improve miner revenue 
in the infinite block size setting. 

\ignore{
it indeed
allows us to circumvent several impossibilities
pertaining to the plain model,  
leading to non-trivial mechanisms
with meaningful guarantees that are otherwise impossible
in the plain model.
}

\elaine{TODO: discuss open questions in the end}

\subsection{Additional Related Work}
\label{sec:related}
We review some additional related works 
besides the most closely related works  
on transaction mechanism design~\cite{zoharfeemech,yaofeemech,functional-fee-market,eip1559,roughgardeneip1559,roughgardeneip1559-ec,dynamicpostedprice,foundation-tfm} mentioned earlier.

Earlier, an elegant line of work~\cite{gtcrypto00,gtcrypto01,gtcrypto02,gtcrypto03,giladutilityindjournal,giladgtcrypto,rdp00,rdp01,rdp02,katzgametheory,gtcrypto06,seqrationalcrypto,gt-fair-cointoss,gt-fair-coin-complete,gt-leader-shi,fruitchain,logstar-gt-leader,credibleauction-comm00,credibleauction-comm01}
revealed
ways in which cryptography and game theory can help each other.
Among them, some works~\cite{gtcrypto06} showed how to rely on cryptography
to remove the trusted mediator 
assumption in certain game theoretic notions such as correlated
equilibrium.  
Some~\cite{gtcrypto00,gtcrypto02,gtcrypto05,gtcrypto03,gt-fair-cointoss,gt-fair-coin-complete} 
showed that adopting game theoretic notions
of fairness rather than the more 
stringent cryptographic notions of fairness 
can allow 
us to circumvent well-known lower bounds.
Recently, 
Ferreira et al.~\cite{credibleauction-comm00}
and Essaidi et al.~\cite{credibleauction-comm01}
showed that using cryptographic commitments   
can help us circumvent lower bounds pertaining to 
credible auctions.
As Chung and Shi~\cite{foundation-tfm}
explained, credible auction
is of a different nature from transaction
fee mechanism design. 
Transaction fee mechanism is a new type of decentralized
mechanism design problem,  and the new connections
between cryptography and mechanism design 
revealed in our paper differ 
in nature from the settings in prior works.

\ignore{
The popularity of blockchains and decentralized applications  
raise exciting 
new challenges for decentralized mechanism design.
Classical insights typically fail  
in the decentralized setting partly because 
there is no trusted entity that promises to implement 
the honest mechanism.
Further, due to the wide-scale deployment of cryptography
in modern blockchains, the interaction
of cryptography and decentralized mechanism design
is a very important direction, and thus far poorly understood.
{\it We thus hope that our new conceptual contributions
can provide fodder for thought and inspire new works in 
this exciting and much explored space. 
}
}

\section{Model and Definitions}

\paragraph{Notation.} 
We use bold letters to denote vectors.
For a vector $\bids = (b_1,\dots, b_N)$,
we use $b_i$ to represent the $i$-th entry of vector $\bids$. 
The notation $\bids_{-i} = (b_1,b_2,\dots,b_{i-1}, b_{i+1},\dots,b_N)$ represents all except the $i$-th entry.
We often use $(\bids_{-i}, b_i)$ and $\bids$ interchangeably. 
Throughout the paper, we use $n$ to denote the number
of users, and $N$ to denote the number of bids. 
$N$ is equal to $n$ if everyone behaves truthfully.
However, 
strategic users may post zero or multiple bids ---
in this case $N$ may not be equal to $n$.
Given a distribution $\mcal{D}$, we use the notation $\supp{D}$
to denote its support.
We use $\R^{\geq 0}$
to denote non-negative real numbers.

\elaine{TODO: n number of users, N number of bids}

\subsection{Transaction Fee Mechanism in the Plain Model}
\label{sec:TFM-plain-model}
We first define transaction fee mechanism (TFM) in the plain model.
Henceforth, we use $\mcal{C}$ to denote a coalition
of strategic players (or a strategic individual).
In particular, $\mcal{C}$ can be a user, the miner of the present block, 
or a coalition of the miner and one or more users.

\paragraph{Plain model.}
In the plain model, a transaction fee mechanism (TFM) 
describes the following game:
\begin{enumerate}[leftmargin=5mm,itemsep=1pt]
\item 
Users not in $\mcal{C}$ submit their bids where each bid is represented by 
a single real value --- 
let ${\bf b}_{-\mcal{C}}$ denote the resulting bid vector.
\item 
The coalition $\mcal{C}$ sees ${\bf b}_{-\mcal{C}}$, and 
then users in $\mcal{C}$ submit their bids.
\item 
The miner of the present block, possibly a member of $\mcal{C}$, 
chooses up to $k$ bids to include in the block, where $k$ denotes the maximum block size.
\item 
Among the at most $k$ bids included in the block, 
the trusted blockchain decides 
1) which of them are confirmed, 2) how much each confirmed bid pays,
and 3) how much revenue is paid to the miner.
\end{enumerate}

Therefore, 
to specify a transaction fee mechanism (TFM) in the plain model, 
it suffices to specify the following rules which 
are {\it possibly randomized} functions:
\begin{itemize}[leftmargin=5mm,itemsep=1pt]
\item {\it Inclusion rule}:
given a bid vector ${\bf b}$, the inclusion rule chooses up to $k$ bids to include
in the block; 
\item {\it Confirmation and payment rules}:
Given the at most $k$ bids included in the block, the confirmation rule
decides which ones to confirm, and the payment rule
decides how much each confirmed user pays.
\item {\it Miner revenue rule}:
Given the at most $k$ bids included in the block, the miner revenue
rule decides how much the miner earns.
\end{itemize}

In particular, the inclusion rule is implemented by the miner, and if the
miner is strategic, it may not follow the prescribed inclusion rule but instead
choose an arbitrary set of bids to include.
By contrast, the confirmation, payment, and miner revenue rules are implemented
by the blockchain, and 
honest implementation is guaranteed.

We assume that the (honest) TFM is {\it symmetric} 
in the following sense:  
if we apply any permutation $\pi$ to 
an input bid vector ${\bf b} = (b_1, \ldots, b_N)$, 
it does not change
the distribution of the {\it random variable} represented by the {\it set}
$\{(b_i, x_i, p_i)\}_{i \in [N]}$ where $x_i$ 
and $p_i$ are random variables denoting 
the probability that bid $i$ is confirmed, and its payment, respectively.
An equivalent, more operational view of the above condition 
is the following.
We may assume that 
the honest mechanism can always be equivalently described in the following manner:
given a bid vector $\bfb$ where each bid may carry some 
extra information such as identity or timestamp, 
the honest mechanism always sorts the vector $\bfb$  
by the bid amount first. During this step, if multiple bids have the same amount,
then arbitrary tie-breaking rules may be applied, 
and the tie-breaking can depend on the extra 
information such as timestamp or identity. 
At this point, the inclusion rule and the confirmation
rules should depend {\it only} on the amount of the bids
and their relative position in the sorted bid vector.
Note that our symmetry requirement 
is natural and quite general --- 
it captures all the mechanisms we know so far~\cite{zoharfeemech,yaofeemech,functional-fee-market,eip1559,roughgardeneip1559,roughgardeneip1559-ec,dynamicpostedprice}. 
In particular, due to possible tie-breaking
in the sorting step, 
our symmetry condition 
does {\it not} require two bids of the same amount to receive the same treatment, 
i.e., the distribution of their outcomes can be different.

\ignore{
More formally, the joint distribution 
of the following {\it random variables} is 
the same for all permutations of the same input 
bid vector ${\bf b}$:
1) the {\it set} of bids included (in terms of their values), 
2) the {\it set} of bids confirmed
and their respective payments, and 3) the miner revenue.
Note that this condition does {\it not} impose the restriction that two bids
of the value receive the same treatment. 
In particular, if multiple bids have the same value 
and tie-breaking is necessary, the TFM can specify the tie-breaking rules.
\elaine{think about this assumption again.}
}

\paragraph{Strategy space.}
A user's truthful behavior is 
submit a single bid representing its true value.
However, strategic users may choose
to submit zero to multiple bids, and the bids need not
reflect their true value.

An honest miner does not submit any bids and honestly
implements the prescribed inclusion rule. 
A strategic miner, on the other hand, may
not honestly implement the prescribed inclusion rule --- it can pick
an arbitrary set of up to $k$ bids of its choice to include.
A strategic miner can also post fake bids.
A coalition $\mcal{C}$'s strategy space is defined in the most natural 
manner, i.e., it includes
any strategic behavior of its members.  

Notably, any strategic player in $\mcal{C}$ can decide its actions
{\it after} having observed the bids of the remaining users
not in $\mcal{C}$.

\begin{figure*}
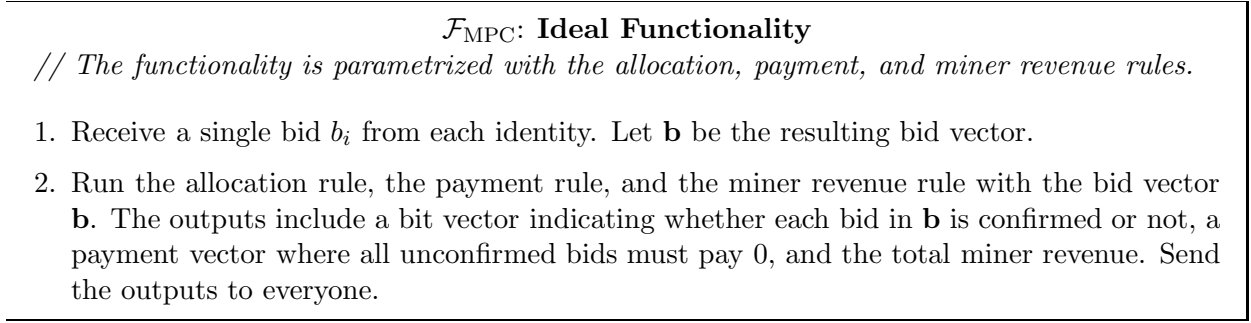

\begin{mdframed}
\begin{center}{${\fmec}$: \bf Ideal Functionality}\end{center}

{\it // The functionality is parametrized with 
the allocation, payment, and miner revenue rules.}

\begin{enumerate}[leftmargin=5mm,itemsep=1pt]
\item 
Receive a single bid $b_i$ from each identity. 
Let $\bids$ be the resulting bid vector.
    
\item 
Run the allocation rule, the payment rule, and the miner revenue rule
with the bid vector $\bids$. 
The outputs include
a bit vector  
indicating 
whether each bid in $\bids$ is confirmed or not,
a payment vector
where all unconfirmed bids must pay $0$,
and the total miner revenue.
Send the outputs to everyone. 
\end{enumerate}
\end{mdframed}
\vspace{-10pt}
\caption{Ideal functionality realized by the MPC protocol.}
\label{fig:Ftfm}
\end{figure*}

\subsection{Transaction Fee Mechanism in the MPC-Assisted Model}
Imagine that all miners jointly run an multi-party computation (MPC)
protocol that implements the TFM. 
\Cref{fig:Ftfm} depicts the natural ideal functionality 
(denoted $\fmec$)
realized by the MPC protocol.
Further, the MPC protocol can achieve full security 
with guaranteed output as long as the majority of the miners
are honest.
Therefore, 
following the modular composition~\cite{Canetti2000} paradigm 
in the standard cryptography literature,  
we can simply assume that a trusted party
$\fmec$ exists --- this is often referred to as the $\fmec$-hybrid model.
We defer how to securely realize $\fmec$ 
to~\cref{sec:mpc}.

\paragraph{MPC-assisted model.}
A transaction fee mechanism (TFM) in the MPC-assisted model
describes the following game:
\begin{enumerate}[leftmargin=5mm,itemsep=1pt]
\item 
Every player (i.e., user or user) can take on {\it zero to multiple} identities,
and every identity submits a bid represented by a single real value
to $\fmec$ defined in \Cref{fig:Ftfm}.
\item 
$\fmec$ 
decides which bids to confirm, 
how much each confirmed bid pays, and the total
miner revenue. 
The total miner revenue is split among the miners.
\end{enumerate}

Therefore, to specify a TFM in the MPC-assisted model,
we need to specify the allocation rule, 
the payment rule, 
and
the miner revenue rule --- 
we assume that these rules are {\it possibly randomized}, 
polynomial-time
algorithms, and the syntax of the rules 
are evident from $\fmec$ in \Cref{fig:Ftfm}.
In comparison with the plain model,
here the {\it inclusion} rule and the {\it confirmation} rule
are combined into  
a single {\it allocation} rule,  
since both inclusion and confirmation decisions are made
by $\fmec$. Just like in the plain model, 
we assume that the (honest) TFM is symmetric.

\paragraph{Strategy space.}
A user's honest behavior is to take on a {\it single} identity,
submit a single bid which reflects its true value.
However, as mentioned above, any strategic user
can take on zero or multiple 
identities, 
submit zero or multiple bids that need not be its true value.

An honest miner
does not take on any identities or submit any bids.
However, a strategic miner can take on   
one or more identities and submit fake bids.
Unlike the plain model, 
here, a strategic miner can no longer choose which 
bids to include in the block --- the allocation rule 
(i.e., the counterpart of the inclusion + confirmation rules
of the plain model)
is enforced by $\fmec$.

One technicality is whether the distribution of users' identities matter, 
and whether
choosing identities strategically should be part of the strategy space.
Jumping ahead, all of our mechanisms are proven to be incentive compatible
even when the strategic individual or coalition can arbitrarily choose
their identities as long as they cannot impersonate honest users' identities.  
On the other hand, all of our impossibility results hold 
even when the strategic individual or coalition  
is forced to choose their 
identities from some a-priori known distribution.
This makes both our feasibility and infeasbility results stronger.

\subsection{Defining Incentive Compatibility}

\paragraph{Utility.}
Every user $i \in [n]$
has a true value $v_i \in \R^{\geq 0}$ if its 
transaction is confirmed. 
If user $i$'s transaction is confirmed and the user pays $p_i$, 
then its utility 
is defined as $v_i - p_i$.
A miner's utility is simply its revenue.

The utility of any strategic coalition $\mcal{C}$
is the sum of the utilities of all members of $\mcal{C}$.
Considering the joint utility of the coalition 
is appropriate since we assume that 
the coalition has a {\it binding} mechanism 
(e.g., decentralized smart contracts)
to split off their gains off the table. 

\paragraph{Ex post incentive compatibility.}
We first define ex post incentive compatibility
for both the plain model and the MPC-assisted model.
Roughly speaking, ex post incentive compatibility 
requires that a strategic player or coalition's best response
is always to behave honestly, even after observing the remaining 
users' bids.
Similary, ex post $\epsilon$-incentive compatibility 
requires that  
no strategy can increase 
a strategic player or coalition's expected
utility by more than $\epsilon$ in comparison with the honest strategy,
and this should hold even if the coalition can decide
its strategy {\it after} having observed the remaining users' bids.

Below in our formal definitions, we define
the {\it approximate} case that allows $\epsilon$ slack.
When $\epsilon = 0$, we get 
{\it strict} incentive compatibiity --- in this case, we 
can omit writing the $\epsilon$.

\begin{definition}[Ex post incentive compatibility]
We say that a mechanism satisfies 
{\it ex post} $\epsilon$-incentive compatibility 
for a set of players $\mcal{C}$ 
(possibly an individual), 
iff for any bid vector ${\bf b}_{-\mcal{C}}$
posted by users not in $\mcal{C}$, 
for any vector of true values ${\bf v}_{\mcal{C}}$ of users
in $\mcal{C}$, 
no strategy can increase $\mcal{C}$'s
expected utility  
by more than $\epsilon$ in comparison with honest behavior.
Specifically, 
\begin{itemize}[leftmargin=5mm]
\item 
{\it UIC}. 
We say that a TFM (in either the plain or MPC-assisted model) satisfies
{\it ex post $\epsilon$-user incentive compatibility (UIC)}, 
iff for any $n$, for any $i \in [n]$, 
for any bid vector ${\bf b}_{-i}$
of all users other than $i$, 
for any true value $v_i$ 
of user $i$, 
no strategy 
can increase $i$'s 
expected utility by more than $\epsilon$ in comparison 
with truthful bidding.

\item 
{\it MIC.}
In the plain model, we focus on the miner of the present
block when defining miner incentive compatibility. 
We say a TFM in the plain model satisfies 
{\it ex post $\epsilon$-miner incentive compatibility MIC}, 
iff for any bid vector ${\bf b}$, 
no strategy can increase the miner's expected utility
by more than $\epsilon$ in comparison with 
honest behavior. 
Recall that  
that here, the miner's 
honest behavior is to honestly implement the inclusion
rule and not inject any fake bids.

In the MPC-assisted model, 
we want MIC to hold for any coalition controlling at most $\rho$ fraction of the miners.
Therefore, we say that an MPC-assisted TFM 
satisfies {\it ex post $\epsilon$-MIC} 
against $\rho$-sized coalitions, 
iff 
for any coalition controlling at most $\rho$ fraction of the miners,
for any bid vector ${\bf b}$, 
no strategy can increase the 
miner's expected utility by more than $\epsilon$
in comparison with  
honest behavior.
In the $\fmec$-hybrid world, 
the miner's honest behavior is simply  
not to take on any identities and inject any fake bids.

\item 
{\it SCP.}
In the plain model, we want side-contract-proofness
to hold for any miner-user coalition that involves 
the miner of the present block, and up to $c$ users. 
We say that a TFM in the plain model 
satisfies {ex post $\epsilon$-side-contract-proofness (SCP)} 
for $c$-sized coalitions, iff for
any miner-user coalition 
consisting of the miner and up to $c$ users, 
for any bid vector ${\bf b}_{-\mcal{C}}$
posted by users not in $\mcal{C}$, 
no strategy can increase $\mcal{C}$'s expected
utility by more than $\epsilon$ in comparison with honest behavior.

In the MPC-assisted model, we want SCP 
to hold for any miner-user coalition that involves 
up to $\rho$ fraction of the miners and up to $c$ users. 
We say that an MPC-assisted TFM 
satisfies {\it ex post $\epsilon$-SCP} for $(\rho, c)$-sized coalitions, iff for
any miner-user coalition\footnote{We require the miner-user coalition
to consist of a non-zero fraction the miners and at least one user --- otherwise
the definition would degenerate to UIC or MIC.} 
consisting of at most $\rho$ fraction of the miners and 
up to $c$ users, 
for any bid vector ${\bf b}_{-\mcal{C}}$
posted by users not in $\mcal{C}$, 
no strategy can increase the coalition's utility
by more than $\epsilon$ in comparison with honest behavior.
\end{itemize}
\end{definition}

\paragraph{Bayesian incentive compatibility.}
For the MPC-assisted model, 
it also makes sense to consider a Bayesian notion 
of incentive compatibility. 
In particular, the MPC-assisted model requires that the strategic player
or coalition 
decides its strategy without having seen the remaining users' bids.
We may assume that the 
strategic player or coalition 
has some a-prior belief of each honest  
user's true value distribution. 
We assume that all honest users'
true values are independently and identically  
distributed (i.i.d.) and sampled from 
some distribution $\mcal{D}$.
In Bayesian incentive compatibility, we imagine
that a strategic individual or coalition 
cares about maximizing its expected utility
where the expectation is taken over not just
the random coins of the mechanism, but also
the remaining honest users' bids.

Henceforth, we denote the bid vector as $\bids$.
Since the strategic players can choose to inject fake bids or drop out, the length of $\bids$ is not necessarily equal to the number of users.
Given a set $\mcal{C}$
of users, we use $\bids_{-\mcal{C}}$ to denote the bids from users outside the coalition and
$\mcal{D}_{-\mcal{C}}$
to denote the joint distribution $\bids_{-\mcal{C}}$.
That is, $\mcal{D}_{-\mcal{C}} = \mcal{D}^{h}$, where $h$ is the number of honest users outside the coalition.
Similarly, for any fixed individual $i$, we use $\bids_{-i}$ to denote the bids from the remaining users and $\mcal{D}_{-i}$ to denote the joint distribution of $\bids_{-i}$.
Again, we define $\epsilon$-incentive compatibility
for the Bayesian setting below, where 
the corresponding strict incentive compatibility
notions 
can be 
obtained by setting $\epsilon = 0$.

\begin{definition}[Bayesian incentive compatibility]
We say that an MPC-assisted TFM satisfies Bayesian 
$\epsilon$-incentive compatibility 
for a coalition or individual $\mcal{C}$, 
iff for any ${\bf v}_{\mcal{C}}$
denoting the true values of users in $\mcal{C}$, 
sample ${\bf b}_{-\mcal{C}} \sim \mcal{D}_{-\mcal{C}}$, 
then, no strategy  
can increase 
$\mcal{C}$'s expected
utility by more than $\epsilon$ 
in comparison with honest bevavior, 
where the expectation is taken over randomness
of the 
honest users bids 
${\bf b}_{-\mcal{C}}$, as well as random coins consumed by the TFM.
Specifically, 
\begin{itemize}[leftmargin=5mm]
\item 
{\it UIC.}
We say that an MPC-assisted TFM satisfies
Bayesian $\epsilon$-UIC, iff 
for any $n$, for any user $i \in [n]$, 
for any true value $v_i \in \R^{\geq 0}$
of user $i$, 
for any strategic bid vector ${\bf b}_i$
from user $i$ which could be empty or consist of multiple bids, 
\[
\underset{{\bf b}_{-i}\sim \mcal{D}_{-i}}{\E}
\left[{\sf util}^i({\bf b}_{-i}, v_i)\right] \geq
\underset{{\bf b}_{-i}\sim \mcal{D}_{-i}}{\E} 
\left[{\sf util}^i({\bf b}_{-i}, {\bf b}_i)\right]
- \epsilon
\]
where ${\sf util}^i({\bf b})$ denotes the expected utility (taken over
the random coins of the TFM)
of user $i$ when the bid vector is ${\bf b}$.
\item 
{\it MIC.}
We say that an MPC-assisted TFM satisfies
Bayesian $\epsilon$-MIC for $\rho$-sized coalitions, iff 
for any miner coalition $\mcal{C}$ controlling at most $\rho$ 
fraction of the miners, 
for any strategic bid vector ${\bf b}'$ injected
by the miner, 
\[
\underset{{\bf b_{-\mcal{C}}\sim \mcal{D}_{-\mcal{C}}}}{\E}
\left[{\sf util}^{\mcal{C}}({\bf b}_{-\mcal{C}})\right]\geq
\underset{{\bf b_{-\mcal{C}}\sim \mcal{D}_{-\mcal{C}}}}{\E}
\left[{\sf util}^{\mcal{C}}({\bf b}_{-\mcal{C}}, {\bf b}') \right]
- \epsilon
\]
where ${\sf util}^{\mcal{C}}({\bf b})$
denotes the expected utility (taken over the random coins
of the TFM) of the coalition $\mcal{C}$ when the input bid
vector is ${\bf b}$.
\item 
{\it SCP.}
We say that an MPC-assisted TFM satisfies
Bayesian $\epsilon$-SCP for $(\rho, c)$-sized coalitions, iff 
for any miner-user coalition consisting of at most $\rho$
fraction of the miners and at most $c$ users, 
for any true value vector ${\bf v}_{\mcal{C}}$ 
of users in $\mcal{C}$, 
for any strategic bid vector ${\bf b}_{\mcal{C}}$ 
of the coalition 
(whose length may not be equal to the number of users in $\mcal{C}$),
\[
\underset{{\bf b_{-\mcal{C}}\sim \mcal{D}_{-\mcal{C}}}}{\E}
\left[{\sf util}^{\mcal{C}}({\bf b}_{-\mcal{C}}, {\bf v}_{\mcal{C}})\right]\geq
\underset{{\bf b_{-\mcal{C}}\sim \mcal{D}^{-\mcal{C}}}}{\E}
\left[{\sf util}^{\mcal{C}}({\bf b}_{-\mcal{C}}, {\bf b}_{\mcal{C}}) \right]
- \epsilon
\]
\end{itemize}
\end{definition}

Note that the Bayesian notions of incentive compatibility
do not make sense in the plain model, since in the plain model,
the strategic individual or coalition can decide
its move {\it after} having observed the remaining honest users' bids.
This is why we adopt only the ex post notion 
in the plain model.
Formally, it is easy to show that any mechanism that satisfies
Bayesian incentive compatibility in the plain model
also satisfies ex post incentive compatibility.

In the MPC-assisted model, both notions make sense, and the ex post
notions are strictly stronger than the Bayesian counterparts. 
Jumping ahead, all of our impossibility results for the MPC-assisted
model work even for the Bayesian notions, 
and all of our mechanism designs  
in the MPC-assisted model 
work even for the ex post notions. This makes both our lower-
and upper-bounds
stronger.

\section{Approximate Incentive Compatibility for Infinite Block Size}
\label{sec:approx-ic-infinite}
In the plain model, 
no UIC and SCP mechanism (even for $c =1$ and infinite block size) 
can achieve  
positive miner revenue~\cite{foundation-tfm}.
In~\cref{sec:strict-ic-0-miner-rev}, we show that the same zero miner
revenue lower bound holds even in the MPC-assisted model. 
Therefore, we consider   
how to get meaningful miner revenue  
using the relaxed notion of approximate incentive compatibility.
In this section, we give a tight characterization of 
approximate incentive compatibility for infinite block size.
This tight characterization applies to both the 
MPC-assisted model and the plain model.  


\subsection{Bounds on Miner Revenue}
We first prove
 a limit on miner revenue in the MPC-assisted model, 
which holds even for 
in the Bayesian setting.
The same limit applies to the plain model 
for the ex post setting --- to see this, observe
that the strategy space is strictly larger in the plain model,
and moreover, for the plain model, we only care about $\rho = 1$.

We now show an MPC-assisted mechanism simultaneously 
satisfies $\eps$-UIC, $\eps$-MIC and $\eps$-SCP 
even for the Bayesian setting and even for $c = 1$ and 
an arbitray choice $\rho \in (0, 1]$,  
then the miner can gain at most $O(n\cdot(\epsilon+ 
\cdot \sqrt{m^* \cdot \eps}))$-miner 
revenue, where $n$ is the number of users,
and $m^*$ is a term that depends on the ``scale'' of the bid distribution.
\elaine{TODO: explain scale in intro better}

To prove the limit on the miner revenue, we care only about the probability of each bid being confirmed, the expected payment of each bid,
and the miner revenue.
Therefore, we introduce the following
notations to denote the outputs of the allocation, payment,
and miner revenue rules --- 
we assume that each user's true value is drawn i.i.d. from 
some distribution $\mcal{D}$
since we are considering the Bayesian setting:
\begin{itemize}[leftmargin=5mm,itemsep=1pt]
    \item {\bf Allocation rule}: given a bid vector $\bids = (b_1,\dots,b_N)$, the allocation rule outputs a vector 
${\bf x}(\bids) := (x_1,\dots,x_N) \in [0, 1]^N$, where each $x_i$ denotes the probability of $b_i$ being confirmed. 
        \item {\bf Payment rule}: given a bid vector $\bids = (b_1,\dots,b_N)$, the payment rule outputs a vector 
${\bf p}(\bids) := 
(p_1,\dots,p_N) \in \R^N$, where each $p_i$ denotes 
the expected payment of $b_i$.
    \item {\bf Miner revenue rule}: 
given a bid vector $\bids = (b_1,\dots,b_N)$, the miner revenue rule 
outputs $\mu(\bids) \in \R$, 
denoting the amount paid to the miner.
\end{itemize}
We also define $\mcal{D}_{-i} := \mcal{D}^{N-1}$, and for the $i$-th user, we define
    \[\overline{x_i}(\cdot) = \underset{\bids_{-i}\sim\otherdist}{\E} [{\bf x}_i(\bids_{-i}, \cdot)], \quad 
    \overline{p_i}(\cdot) = \underset{\bids_{-i}\sim\otherdist}{\E} [{\bf p}_i(\bids_{-i}, \cdot)], \quad 
    \overline{\mu_i}(\cdot) =\underset{\bids_{-i\sim\otherdist}}{\E}[\mu(\bids_{-i}, \cdot)].
    \]

Henceforth, we often use $\mecha$ to denote a TFM in the MPC-assisted model.
The crux of our proof is to characterize how miner revenue changes 
when we lower one user's bid to $0$ (Lemma~\ref{lem:miner-eps-UIC}).
We then apply this argument $n$ times, and lower each user's bid
one by one to $0$ to get the desired bound. To make the second step 
work, we need to use approximate MIC 
to remove a user's bid from consideration once we 
have lowered it to zero --- this ensures that in 
any step of our inductive argument, the non-strategic users'
bids are always i.i.d. sampled from $\mcal{D}$.

\paragraph{Warmup.}
To understand how much the miner revenue changes
when one user lowers its bid to $0$, 
we start from a simplified case where a TFM $\mecha$ is 
Bayesian {\it strict}-UIC and Bayesian $\eps$-SCP 
for $c = 1$ and some $\rho \in (0, 1]$.
By Myerson's Lemma~\cite{myerson}, 
strict-UIC implies that, for any user $i$, 
the allocation rule $x_i(\cdot)$ must be non-decreasing. 
Moreover, the expected payment when bidding $b$ is specified as 
\[
		\overline{p_i}(b) 
		= b\cdot \overline{x_i}(b) -\int_{0}^b \overline{x_i}(t) dt.
\]

We care about how much the miner revenue
can increase when user $i$ bids $r$ instead of $0$. 
One trivial upper bound can be obtained as follows. 
Imagine that user $i$'s true value is $0$, but it bids
$r$ instead. In this case, the user's loss
in utility (in comparison 
with truthful bidding) 
is represented
by the area of the gray triangle $S$
in \Cref{fig:single-step}.
Due to $\eps$-SCP, the miner revenue increase 
when user $i$ bids $r$ instead of $0$ must be upper bounded
by $S + \epsilon$.
This bound, however, is not tight. 
To make it tighter, 
we consider bounding it in two steps by introducing a mid-point $r' \in (0, r)$.
If user $i$'s true value is $0$, but it bids $r'$ instead,
its utility loss is the area $S_1$ of \Cref{fig:two-step}. 
By $\epsilon$-SCP, we conclude that 
$\overline{\mu_i}(r') - \overline{\mu_i}(0) \geq S_1 + \epsilon$.
Now, imagine user $i$'s true value is $r'$ but it bids $r$ instead.
Using a similar argument, 
we conclude that 
$\overline{\mu_i}(r) - \overline{\mu_i}(r') \geq S_2 + \epsilon$
(see \Cref{fig:two-step}). 
Summarizing the above, we have that  
$\overline{\mu_i}(r) - \overline{\mu_i}(0) \geq S_1 + S_2 + 2\epsilon$. 

\ignore{
Therefore, by $\eps$-SCP, the miner revenue should increase by no more than $S+\epsilon$, i.e., $\mu_i(r) - \mu_i(0)\leq S+\epsilon$.
However, this bound is not tight enough.
If the user $i$ first increases its bid from $0$ to $r'$, it loses $S_1$ utility, where $S_1$ is marked as in Figure~\ref{fig:two-step}.
Again by $\eps$-SCP, $\mu_i(r') - \mu_i(0)\leq S_1+\epsilon$.
By the same argument, $\mu_i(r) - \mu_i(r')\leq S_2+\epsilon$, where $S_2$ is marked as in Figure~\ref{fig:two-step}.
Putting together, $\mu_i(r) - \mu_i(0)$ should be upper bounded by $S_1 + S_2 +2\epsilon$, which is a tighter bound than $S+\epsilon$ for large $r$.
}
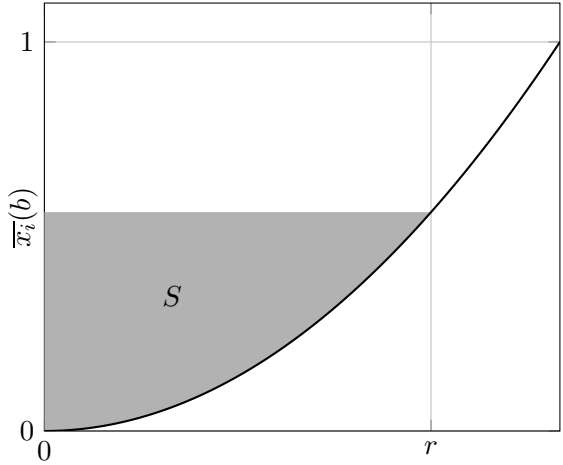
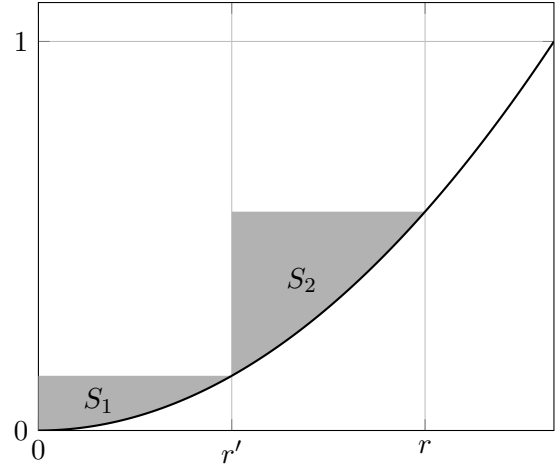
\begin{figure}[H]
    \centering
    \begin{subfigure}[b]{0.45\textwidth}
    \centering
    \begin{tikzpicture}[trim axis left]
        \begin{axis}[grid=major, xmin=0, xmax=2, ymin=0, ymax=1.1,
             ylabel={$\overline{x_i}(b)$}, ylabel style={yshift=-0.5cm},
             xtick = {0,1.5}, xticklabels={$0$, $r$},
             ytick = {0,1}, yticklabels={$0$, $1$},]
        \addplot[black, thick, samples=300, smooth,domain=0:2, name path = A] {0.25*x^2};
        \addplot[draw=none,name path=B] {0.25*1.5*1.5};
        \addplot[black!30] fill between[of=A and B,soft clip={domain=0:1.5}];
        \end{axis}
        \draw(1.7,1.8) node {$S$};
        \end{tikzpicture}
        \caption{When user $i$ changes its bid from $0$ to $r$, it loses utility $S$. Therefore, miner revenue changes by no more than $S+\epsilon$.}
        \label{fig:single-step}
    \end{subfigure}
    \hfill
    \begin{subfigure}[b]{0.45\textwidth}
    \centering
    \begin{tikzpicture}[trim axis left]
        \begin{axis}[grid=major, xmin=0, xmax=2, ymin=0, ymax=1.1,
             xtick = {0,0.75,1.5}, xticklabels={$0$, $r'$, $r$},
             ytick = {0,1}, yticklabels={$0$, $1$},]
        \addplot[black, thick, samples=300, smooth,domain=0:2, name path = A] {0.25*x^2};
        \addplot[draw=none,name path=C] {0.25*1.5*1.5};
        \addplot[draw=none,name path=D] {0.25*0.75*0.75};
        \addplot[black!30] fill between[of=A and C,soft clip={domain=0.75:1.5}];
        \addplot[black!30] fill between[of=A and D,soft clip={domain=0:0.75}];
        \end{axis}
        \draw(0.8,0.4) node {$S_1$};
        \draw(3.5,2) node {$S_2$};
        \end{tikzpicture}
        \caption{When user $i$ changes its bid from $0$ to $r'$, it loses utility $S_1$. Then when it changes its bid from $r'$ to $r$, it loses utility $S_2$.}
        \label{fig:two-step}
    \end{subfigure}
    \caption{User's utility change}
    \label{fig:user-util-change}
\end{figure}
\elaine{the figures waste too much space}

To get a tight bound, the key is how to choose
the optimal number of steps $L$ we use in the above argument.
Taking more steps makes the total area of the gray triangles 
smaller; however, every step incurs an extra $\epsilon$. 
Given the number of steps $L$, the sum of the $L$ triangles
is upper bounded by $r/L$, and since each step incurs
an additive $\epsilon$ term, our goal is to minimize
the expression $r/L + \epsilon L$. 
Picking $L=\sqrt{\frac{r}{\epsilon}}$ 
minimizes the expression and thus we have that  
$\overline{\mu_i}(r) - \overline{\mu_i}(0)\leq 2\sqrt{r\epsilon}$.
\ignore{
Specifically, if we consider the allocation rule specified in~\cref{fig:user-util-change}, we choose $r_1,\dots,r_{L}$ such that $r_l = l\cdot \sqrt{r\epsilon}$, and the size of each triangle $S_l = \frac{1}{2}\sqrt{r\epsilon}\cdot [x_i(r_l) - x_i(r_{l-1})]$.
Then 
\begin{equation}
    \label{eqn:strict-miner-util-ub}
    \sum_{l=1}^L S_l+ l\eps = \frac{1}{2}\sum_{l=1}^{L}\sqrt{r\epsilon}[x_i(r_l) - x_i(r_{l-1})] + L\epsilon\leq 2\sqrt{r\epsilon},
\end{equation}
where the last step relies on the property that the allocation rule $x_i(\cdot)$ is non-decreasing.
}

\paragraph{Full proof.}
The above warmup argument works for strict-UIC and $\epsilon$-SCP.
We want to prove a limitation on miner revenue
for Bayesian $\epsilon$-UIC and $\epsilon$-SCP. 
The challenge is that for $\epsilon$-UIC,
Myerson's lemma no longer holds --- in particular, 
the allocation rule may not even be monotone any more.
The key idea our proof
is to give a generalization of Myerson's lemma
to account for the $\epsilon$ slack
in incentive compatibility. 
We first prove a generalization 
of Myerson's {\it payment difference sandwich}
for $\eps$-UIC. 

\begin{lemma}
\label{lemma:payment-sandwich}
Given any (possibly randomized) MPC-assisted TFM that is Bayesian 
$\eps$-UIC, 
it must be that 
for any user $i$, 
for any $y \leq z$, 
\begin{equation}
    z\cdot[\overline{x_i}(z)-\overline{x_i}(y)]+\eps\geq \overline{p_i}(z) - \overline{p_i}(y) \geq y\cdot [\overline{x_i}(z)-\overline{x_i}(y)]-\epsilon.
    \label{eqn:payment-sandwich}
\end{equation}
\end{lemma}
\begin{proof}
The proof is similar to the proof of Myerson's Lemma. 
Note that user $i$'s expected utility is $v\cdot \overline{x_i}(b) - \overline{p_i}(b)$ if its true value is $v$ and its bid is $b$.
By the definition of Bayesian $\eps$-UIC, it must be that
\[z\cdot \overline{x_i}(z) - \overline{p_i}(z)+\epsilon\geq z\cdot \overline{x_i}(y) - \overline{p_i}(y).\]
Otherwise if user $i$'s true value is $z$, bidding $y$ can bring it strictly more than $\epsilon$ utility compared to bidding truthfully, which contradicts Bayesian $\eps$-UIC. 
By the same reasoning, we have
\[y\cdot \overline{x_i}(y) - \overline{p_i}(y)+\epsilon \geq y\cdot \overline{x_i}(z) - \overline{p_i}(z).\]
The lemma thus follows by combining these two inequalities.

\end{proof}

Based on this payment difference sandwich, we have the following result about the expected miner's revenue for approximate incentive compatibility.
\begin{lemma}
\label{lem:miner-step-eps-UIC}
Fix any $\rho \in (0, 1]$. 
For any (possibly randomized) MPC-assisted TFM that is 
Bayesian $\eps_u$-UIC and 
Bayesian $\eps_s$-SCP against a $(\rho,1)$-sized coalition, 
it must be that 
for any user $i$, 
for any $y \leq z$,
\begin{equation}
    \overline{\mu_i}(z) - \overline{\mu_i}(y) \leq \frac{1}{\rho}(\eps_u + \eps_s + S(y,z)),
    \label{eqn:eps-UIC-miner-revenue}
\end{equation}
where $S(y,z) = (z-y)[\overline{x_i}(z)-\overline{x_i}(y)]$.
\end{lemma}
\begin{proof}
The utility of user $i$ is $v\cdot \overline{x_i}(b)-\overline{p_i}(b)$ if its true value is $v$ and it bids $b$. 
Imagine that the user $i$'s true value is $y$. 
If user $i$ overbids $z > y$ instead of its true value $y$, then its expected utility decreases by
\begin{align*}
    \Delta &= y\cdot \overline{x_i}(y)-\overline{p_i}(y) - [y\cdot \overline{x_i}(z)-\overline{p_i}(z)]\\
    &= -y\cdot[\overline{x_i}(z) - \overline{x_i}(y)] + (\overline{p_i}(z)-\overline{p_i}(y))\\
    &\leq -y\cdot[\overline{x_i}(z) - \overline{x_i}(y)] + z\cdot[\overline{x_i}(z) - \overline{x_i}(y)] +\eps_u &\text{By
Bayesian $\epsilon_u$-UIC and 
~\eqref{eqn:payment-sandwich}}\\
    &= (z-y)\cdot[\overline{x_i}(z)-\overline{x_i}(y)] + \eps_u = S(y,z) + \eps_u.
\end{align*}
A graphical description of $S(y,z)$ is 
shown in Figure~\ref{fig:user-util-change} --- note that 
$S(y,z)$ can be {\it negative} since  
the allocation rule $\overline{x_i}(\cdot)$ may not be monotone
under approximate UIC.
\begin{figure}[H]
    \centering
    \begin{subfigure}[b]{0.45\textwidth}
    \centering
    \begin{tikzpicture}[trim axis left]
        \begin{axis}[grid=major, xmin=0, xmax=2, ymin=0, ymax=1,
             xlabel=$b$, ylabel={$\overline{x_i}(b)$},ylabel style={yshift=-0.5cm},
             xtick = {0,1, 1.5}, xticklabels={$0$, $y{\textcolor{white}{'}}$, $z{\textcolor{white}{'}}$},
             ytick = {-0.5,1},
             scale=0.8, restrict y to domain=-1:1]
        \addplot[black, thick, samples=300, smooth,domain=0:2, name path = A] {0.25*x^2};
        \addplot[draw=none,name path=C] {0.25};
        \addplot[draw=none,name path=B] {0.5625};
        \addplot[black!30] fill between[of=C and B,soft clip={domain=1:1.5}];
        \end{axis}
        \end{tikzpicture}
        \caption{An illustrative example of $S(y,z)$ in increasing function. The size of the gray area in the figure is exactly $S(y,z)$.}
        \label{fig:single-diff-increase}
    \end{subfigure}
    \hfill
    \begin{subfigure}[b]{0.45\textwidth}
        \centering
        \begin{tikzpicture}[trim axis left]
        \begin{axis}[grid=major, xmin=0, xmax=1.75, ymin=0, ymax=1,
             xlabel=$b$, ylabel={$\overline{x_i}(b)$},
             xtick = {0,0.75, 1.5}, xticklabels={$0$, $y{\textcolor{white}{'}}$, $z{\textcolor{white}{'}}$},
             ytick = {-0.5,1},
             scale=0.8, restrict y to domain=-1:1]
        \addplot[black, thick, samples=300, smooth,domain=0:2, name path = A] {1-(x-1)^2};
        \end{axis}
        \draw[dashed] (2.34,4.25) rectangle (4.7,3.4);
        \end{tikzpicture}
        \caption{When the function decreases, $S(y,z)$ can be negative. $S(y,z)$ is the negative of the dashed rectangle area.}
        \label{fig:single-diff-decrease}
    \end{subfigure}
    \caption{User's utility change}
    \label{fig:user-util-change}
\end{figure}
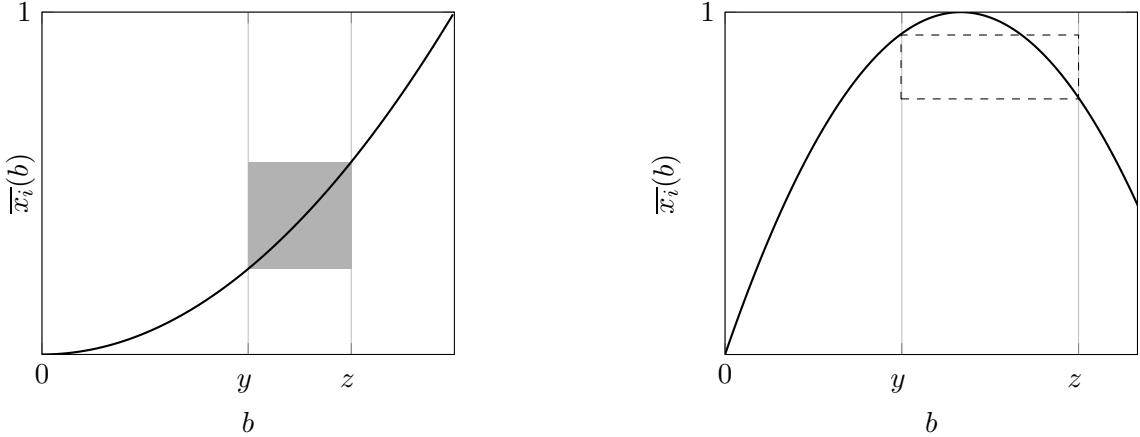

\elaine{figure 3 not general, does not show the decreasing case}

By Bayesian $\epsilon_s$-SCP, 
it must be that $\rho\overline{\mu_i}(z) - \rho\overline{\mu_i}(y)\leq \Delta+\eps_s$; otherwise, a strategic player controlling $\rho$ fraction of the miners can collude with user $i$, and ask user $i$ to bid $z$ instead of its true value $y$. 
This increases the coalition's utility by 
strictly more than $\eps_s$ compared to the honest strategy, 
which contradicts Bayesian $\eps_s$-SCP.
\end{proof}
\begin{lemma}
\label{lem:miner-eps-UIC}
Fix any $\rho\in(0,1]$.
For any (possibly randomized) MPC-assisted TFM that is 
Bayesian $\eps_u$-UIC and 
Bayesian $\eps_s$-SCP against a $(\rho,1)$-sized coalition,
for any user $i$, for any value $r$, it must be that
\begin{equation}
    \overline{\mu_i}(r) - \overline{\mu_i}(0) \leq
    \begin{cases}
    \frac{2}{\rho}(\epsilon_s+\eps_u), &\text{if } r\leq \epsilon_s+\eps_u\\
    \frac{2}{\rho}(\sqrt{r(\eps_s+\eps_u)}), &\text{if } r>\epsilon_s+\eps_u.
    \end{cases}
\end{equation}
\end{lemma}
\ke{change the statement in intro}
\begin{proof}
Let $\eps' = \eps_s+\eps_u$. To prove this Lemma, we consider the following two cases.
\paragraph{Case 1: If $r\leq \epsilon'$.} In this case, by Lemma~\ref{lem:miner-step-eps-UIC}, we have that
\[\overline{\mu_i}(r)-\overline{\mu_i}(0)\leq \frac{1}{\rho}\left(\eps_u+\eps_s+S(0,r)\right)\leq\frac{1}{\rho}\left( \epsilon_u+\eps_s+r\right)\leq \frac{2\eps'}{\rho}.\]

\paragraph{Case 2: If $r > \epsilon'$.} We choose a sequence of points that partitions the interval $[0,r]$ as follows.
Let $L = \lfloor\sqrt{\frac{r}{\epsilon'}}\rfloor$.
Set $r_0 = 0$ and $r_{L+1} = r$.
For $l = 1,\dots,L$, we set $r_l = l\cdot\sqrt{r\epsilon'}$.
Each segment except the last one is of length $\sqrt{r\epsilon'}$, while the last one has length no more than $\sqrt{r\epsilon'}$.

Now we proceed to bound $\overline{\mu_i}(r) - \overline{\mu_i}(0)$. 
Note that 
\begin{align*}
    \overline{\mu_i}(r) - \overline{\mu_i}(0) &= \sum_{l=0}^{L} [\overline{\mu_i}(r_{l+1}) - \overline{\mu_i}(r_{l})]\\
    & \leq \sum_{l=0}^L\frac{1}{\rho}[\eps'+S(r_l,r_{l+1})] &\text{By~\cref{lem:miner-step-eps-UIC}}\\
    & = \frac{L\eps'}{\rho} + \frac{1}{\rho}\sum_{l=0}^L(r_{l+1}-r_{l})\cdot[\overline{x_i}(r_{l+1})-\overline{x_i}(r_l)]\\
    & \leq  \frac{L\eps'}{\rho} + \frac{1}{\rho}\sqrt{r\eps'}\sum_{l=0}^L[\overline{x_i}(r_{l+1})-\overline{x_i}(r_l)] &\text{By the choice of }r_l\\
    & \leq \frac{L\eps'}{\rho} + \frac{1}{\rho}\sqrt{r\epsilon'} &\text{By }\overline{x_i}(r)\leq 1
\end{align*}
Since $L = \lfloor\sqrt{\frac{r}{\epsilon'}}\rfloor\leq \sqrt{\frac{r}{\epsilon'}}$, we have that
\[\overline{\mu_i}(r) - \overline{\mu_i}(0)\leq \frac{2\sqrt{r\epsilon'}}{\rho}.\]

\end{proof}

Now, 
we want to bound the miner revenue by lowering each user's bid to 
$0$ one by one, and apply \Cref{lem:miner-eps-UIC} in each step.
To make this argument work,
one key insight is to rely on  
approximate MIC to remove a user's bid from consideration
after lowering it to zero --- see \Cref{eqn:eps-MIC} in the proof
of \Cref{thm:LB-eps-UIC} below. 
This ensures that in any step of the induction,
any honest user's bid is sampled from $\mcal{D}$. 
\begin{theorem}[Limit on miner revenue for approximate incentive compatibility]
	\label{thm:LB-eps-UIC}
	Suppose that there are $n$ users, whose true values are drawn i.i.d. from some 
distribution $\mcal{D}$.
Given any (possibly randomized) MPC-assisted TFM
that is Bayesian $\eps_u$-UIC, Bayesian $\eps_m$-MIC against a $\rho$-sized miner coalition and Bayesian $\eps_s$-SCP against a $(\rho,1)$-sized coalition, it must be that
\begin{equation}\underset{\bids\sim\mathcal{D}^{n}}{\E}[\mu(\bids)]\leq \frac{2n}{\rho}\left(\epsilon+C_{\mcal{D}}\sqrt{\epsilon}\right),
\label{eqn:minerrevlimit}
\end{equation}
	where $\eps = \eps_s+\eps_u+\eps_m$, and $\mcal{C}_{\mcal{D}} = \E_{X\sim\mcal{D}}[\sqrt{X}]$ is a term 
that depends on the ``scale'' of the distribution $\mcal{D}$.
\end{theorem}
\begin{proof}
Since the TFM is Bayesian $\eps_m$-MIC, it must be that
for any $\ell$, 
\begin{equation}
    \underset{\bids\sim\mathcal{D}^{\ell}}{\E}[\rho\mu(\bids, 0 )]\leq \underset{\bids\sim\mathcal{D}^{\ell}}{\E}[\rho\mu(\bids)]+\eps_m.
    \label{eqn:eps-MIC}
\end{equation}
Otherwise, 
the strategic miner 
can inject a bid $0$ and increase its miner revenue by 
strictly more than $\eps_m$, while it does not need pay anything for injecting this $0$-bid. This violates Bayesian $\eps_m$-MIC. 

Let $f(\cdot)$ be the p.d.f.~of distribution $\mcal{D}$. 
By the law of total expectation,
\[\quad\underset{\bids\sim\mcal{D}^{n}}{\E} [\mu(\bids)]= \int_{0}^{\infty} \underset{\bids'\sim\mcal{D}^{n-1}}{\E} [\mu(\bids', r)] f(r) dr.\]
Let $\eps' = \eps_s+\eps_u$. Since the mechanism is Bayesian $\eps_u$-UIC and Bayesian $\eps_s$-SCP against $(\rho,1)$-sized coalition, by~\cref{lem:miner-eps-UIC}, it must be that
\begin{align*}
    \int_{0}^{\eps'} \underset{\bids'\sim\mcal{D}^{n-1}}{\E} [\mu(\bids', r)] f(r) dr &\leq \int_{0}^{\eps'} \left[\underset{\bids'\sim\mcal{D}^{n-1}}{\E} [\mu(\bids', 0)]+\frac{2\eps'}{\rho}\right]f(r)dr;\\
    \int_{\eps'}^{\infty} \underset{\bids'\sim\mcal{D}^{n-1}}{\E} [\mu(\bids', r)] f(r) dr &\leq \int_{\eps'}^{\infty} \left[\underset{\bids'\sim\mcal{D}^{n-1}}{\E} [\mu(\bids', 0)]+\frac{2\sqrt{r\eps'}}{\rho}\right]f(r)dr.
\end{align*}
Summing up the two inequalities above, we can bound the expected miner revenue with
\begin{align*}
    &\quad\underset{\bids\sim\mcal{D}^{n}}{\E} [\mu(\bids)]\\ 
    &= \int_{0}^{\eps'} \underset{\bids'\sim\mcal{D}^{n-1}}{\E} [\mu(\bids', r)] f(r) dr +  \int_{\eps'}^{\infty} \underset{\bids'\sim\mcal{D}^{n-1}}{\E} [\mu(\bids', r)] f(r) dr \\
    &\leq \int_{0}^{\eps'} \left[\underset{\bids'\sim\mcal{D}^{n-1}}{\E} [\mu(\bids', 0)]+\frac{2\eps'}{\rho}\right]f(r)dr+\int_{\eps'}^{\infty} \left[\underset{\bids'\sim\mcal{D}^{n-1}}{\E} [\mu(\bids', 0)]+\frac{2\sqrt{r\eps'}}{\rho}\right]f(r)dr\\
    &\leq \underset{\bids'\sim\mcal{D}^{n-1}}{\E} [\mu(\bids', 0)]+\frac{2\eps'}{\rho}\int_{0}^{\eps'} f(r) dr + \frac{2\sqrt{\eps'}}{\rho} \int_{\eps'}^{\infty}\sqrt{r}f(r)dr
\end{align*}
By~\eqref{eqn:eps-MIC}, we have that $\underset{\bids'\sim\mcal{D}^{n-1}}{\E} [\mu(\bids',0)]\leq \underset{\bids'\sim\mcal{D}^{n-1}}{\E} [\mu(\bids')] + \frac{\eps_m}{\rho}$.
Therefore, 
\begin{align*}
     &\quad\underset{\bids\sim\mcal{D}^{n}}{\E} [\mu(\bids)]\\ 
     & \leq \underset{\bids'\sim\mcal{D}^{n-1}}{\E} [\mu(\bids', 0)]+\frac{2\eps'}{\rho}\int_{0}^{\eps'} f(r) dr + \frac{2\sqrt{\eps'}}{\rho} \int_{\eps'}^{\infty}\sqrt{r}f(r)dr\\
     & \leq \underset{\bids'\sim\mcal{D}^{n-1}}{\E} [\mu(\bids')]+\frac{\eps_m}{\rho}+\frac{2\eps'}{\rho}+\frac{2\sqrt{\eps'}}{\rho}\E_{X\sim\mcal{D}}[\sqrt{X}]\\
     &\leq \underset{\bids'\sim\mcal{D}^{n-1}}{\E} [\mu(\bids')]+\frac{2\epsilon}{\rho}+\frac{2C_{\mcal{D}}\sqrt{\epsilon}}{\rho},
\end{align*}
where the last step comes from the fact that $\eps = \eps_s+\eps_u+\eps_m$.
The theorem follows by induction on $n$, where in each induction step we repeat the argument above.
\end{proof}

It is easy to see that the same miner revenue limit of \Cref{thm:LB-eps-UIC}
also holds in the plain model,
as stated
in the following corollary.
\begin{corollary}
Suppose that there are $n$ users, whose true values are drawn i.i.d. from some
distribution $\mcal{D}$.
Given any (possibly randomized) TFM
in the plain model
that is $\eps_u$-UIC, $\eps_m$-MIC,  and 
$\eps_s$-SCP even for $c = 1$, it must be that
\begin{equation}\underset{\bids\sim\mathcal{D}^{n}}{\E}[\mu(\bids)]\leq 2n\left(\epsilon+C_{\mcal{D}}\sqrt{\epsilon}\right),
\end{equation}
        where $\eps = \eps_s+\eps_u+\eps_m$, and $\mcal{C}_{\mcal{D}} = \E_{X\sim\mcal{D}}[\sqrt{X}]$ is a term
that depends on the ``scale'' of the distribution $\mcal{D}$.
\end{corollary}
\begin{proof}
Follows directly from \Cref{thm:LB-eps-UIC} which holds in particular for $\rho = 1$, 
and the fact that 
the strategy space in the plain model is strictly larger than in the MPC-assisted 
model.
\end{proof}

\subsection{Achieving Optimal Revenue: Proportional Auction}
\label{section:proportional}

We now show that the limit on miner revenue 
in \Cref{thm:LB-eps-UIC}
is asymptotically tight, i.e., we can indeed design
a TFM, even in the plain model, whose miner revenue asymptotically
matches \Cref{eqn:minerrevlimit}
for some natural bid distribution.

\begin{mdframed}
    \begin{center}
    {\bf Proportional Auction} (plain model)
    \end{center}
\vspace{2pt}
\noindent    \textbf{Parameters:} the slack $\eps$, 
the reserved price $r$ where $r \geq 2\eps$.

\vspace{2pt}
\noindent
    \textbf{Input:} a bid vector $\bfb = (b_1,\dots, b_N)$.
    
\vspace{2pt}
\noindent\textbf{Mechanism:}
    \begin{itemize}[leftmargin=5mm,itemsep=1pt,topsep=2pt]
    \item {\it Inclusion rule.}
    Include all bids in $\bids$.
    \item 
    {\it Confirmation rule.}
    For each bid $b$, if $b < r$, it is confirmed with the probability $b/r$; otherwise, if $b \geq r$, it is confirmed with probability $1$.
    \item 
    {\it Payment rule.}
    For each confirmed bid $b$, if $b < r$, it pays $b/2$; otherwise, it pays $r/2$.
    \item 
    {\it Miner revenue rule.}
    For each confirmed bid $b$, if $b\geq \sqrt{2r\epsilon}$\footnote{This guarantees that the miner revenue does not exceed the total payment.}, then miner is paid $\frac{\sqrt{2r\eps}}{2}$.
    \end{itemize}
\end{mdframed}

The above mechanism is called the proportional mechanism
since the user's confirmation probability 
is proportional to the bid in the region $[0, r]$, 
and any bid that is at least $r$
is confirmed with probability $1$.

\begin{theorem}
\label{thm:proportional}
The above proportional auction in the plain model is UIC, MIC and $\frac{5}{4}c\eps$-SCP 
against $c$-sized coalitions for arbitrary $c\geq 1$.
\end{theorem}

\paragraph{Proof intuition.}
We provide the proof intuition 
and defer the full proof to~\cref{section:proportional-MPC}.
First, UIC and MIC are easy to prove.
Observe that the 
allocation rule (i.e., the union of the 
inclusion and confirmation rules)  
is monotone,
and by design, the payment rule is the unique
one that satisfies Myerson's Lemma.
Therefore, the mechanism satisfies UIC.
It is easy to see that injecting a bid does not help the miner, 
since each bid's contribution to the miner revenue 
is independent and limited by the payment amount. 

Proving that the mechanism satisfies 
$\frac{5}{4}c\epsilon$-SCP is more technical.
Here we give an illustrative explanation to show that the joint utility of each user and the miner can increase by at most $\frac{5}{4}\epsilon$.
Since underbidding does not increase the user's utility or the miner's revenue, we focus on overbidding. 
Note that overbidding does not increase the joint utility for a user whose true value is $v\geq r$.
Therefore, we focus in the case where the colluding user has true value $v<r$ and overbids. 

If $v\geq \sqrt{2r\eps}$, the user's utility loss when overbidding to $v'$ is represented by the gray triangle in~\cref{fig:large-true-val}.
Meanwhile, the miner's expected revenue increases by $\frac{\sqrt{2r\epsilon}}{2}(\frac{v'}{r}-\frac{v}{r})$, which is the area of the dashed rectangle in ~\cref{fig:large-true-val}.
Therefore, when the user overbids by $v' -v = \frac{\sqrt{2r\epsilon}}{2}$, the coalition's utility increase is maximized and equals to $\frac{\eps}{4}$.

If $v < \sqrt{2r\eps}$ and the colluding user overbids to $v'\geq \sqrt{2r\eps}$, then the user's utility loss when overbidding to $v'$ is represented by the area of the gray triangle in~\cref{fig:small-true-val}.
The miner's revenue now increases by $\frac{v'}{r}\cdot\frac{\sqrt{2r\eps}}{2}$, because the user's utility would be $0$ if the user behaves honestly.
The increase in the miner's revenue is represented by the dashed rectangle in~\cref{fig:small-true-val}.
The increase in the joint utility of the coalition is maximized when $v$ is arbitrarily close to $\sqrt{2r\epsilon}$ and the user overbids by $v'-v = \frac{\sqrt{2r\epsilon}}{2}$.
In this case, the joint utility of the coalition increases by $\frac{5}{4}\epsilon$.
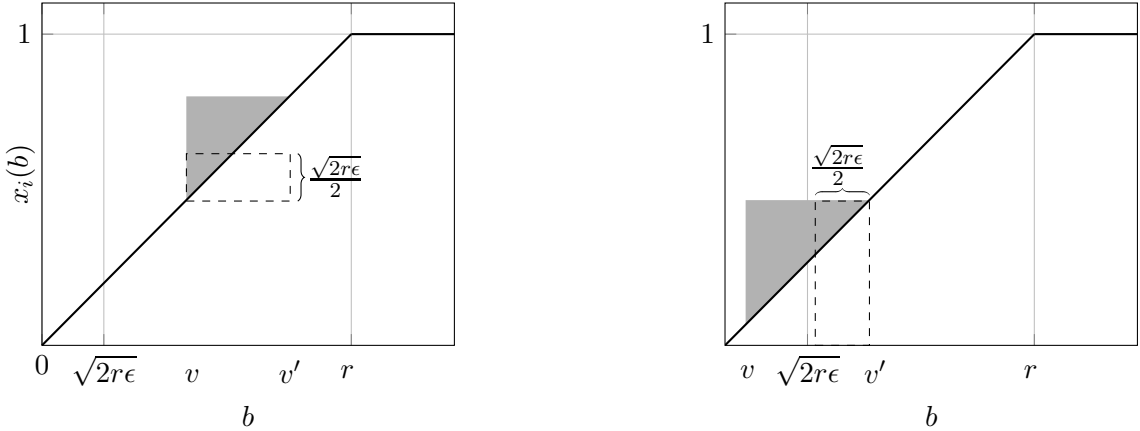
\begin{figure}[H]
    \centering
    \begin{subfigure}[b]{0.45\textwidth}
    \centering
    \begin{tikzpicture}[trim axis left]
        \begin{axis}[grid=major, xmin=0, xmax=2, ymin=0, ymax=1.1,
             xlabel=$b$, ylabel={$x_i(b)$},ylabel style={yshift=-0.5cm},
             xtick = {0,0.3,1.5}, xticklabels={$0$, $\sqrt{2r\epsilon}$, $r{\textcolor{white}{'}}$},
             ytick = {-0.5,1},
             scale=0.8, restrict y to domain=-1:1]
        \addplot[black, thick, samples=300, smooth,domain=0:1.5, name path = A] {2/3*x};
        \addplot[black, thick, samples=300, smooth,domain=1.5:2, name path = A1] {1};
        \addplot[draw=none,name path=C] {1.2*2/3};
        \addplot[black!30] fill between[of=C and A,soft clip={domain=0.7:1.2}];
        \end{axis}
        \draw[dashed] (1.92,1.92) rectangle (3.3,2.55);
        \draw (2,-0.4) node {$v$};
        \draw (3.3,-0.35) node {$v'$};
        \draw [decorate,
    decoration = {brace,mirror}] (3.4,1.92) --  (3.4,2.55);
        \draw (3.9,2.25) node {$\frac{\sqrt{2r\eps}}{2}$};
        \end{tikzpicture}
        \caption{An illustrative example of the coalition's joint utility change when the user's true value $v\geq \sqrt{2r\epsilon}$. }
        \label{fig:large-true-val}
    \end{subfigure}
    \hfill
    \begin{subfigure}[b]{0.45\textwidth}
        \centering
        \begin{tikzpicture}[trim axis left]
        \begin{axis}[grid=major, xmin=0, xmax=2, ymin=0, ymax=1.1,
             xlabel=$b$, 
             xtick = {0.4,1.5}, xticklabels={$\sqrt{2r\epsilon}$, $r{\textcolor{white}{'}}$},
             ytick = {-0.5,1},
             scale=0.8, restrict y to domain=-1:1]
        \addplot[black, thick, samples=300, smooth,domain=0:1.5, name path = A] {2/3*x};
        \addplot[black, thick, samples=300, smooth,domain=1.5:2, name path = A1] {1};
        \addplot[draw=none,name path=C] {0.7*2/3};
        \addplot[black!30] fill between[of=C and A,soft clip={domain=0.1:0.7}];
        \end{axis}
        \draw[dashed] (1.2,0) rectangle (1.92,1.92);
        \draw (2,-0.35) node {$v'$};
        \draw (0.3,-0.35) node {$v$};
        \draw [decorate,
    decoration = {calligraphic brace}] (1.2,2) --  (1.92,2);
        \draw (1.5,2.4) node {$\frac{\sqrt{2r\eps}}{2}$};
        \end{tikzpicture}
        \caption{An illustrative example of the coalition's joint utility change when the user's true value $v< \sqrt{2r\epsilon}$. }
        \label{fig:small-true-val}
    \end{subfigure}
    \caption{Coalition's joint utility change when the miner colluding with one user}
    \label{fig:proportional-sketch}
\end{figure}

\section{Characterization of Finite Block Size in the Plain Model}
In real-world blockchains, we do not have an infinite block size.
Chung and Shi~\cite{foundation-tfm} showed that no non-trivial plain-model TFM 
can achieve strict UIC and 
strict SCP (even when $c = 1$)
for finite block size.
In this section, we show that although  
approximate incentive compatibility 
can help us overcome this impossibility, 
nonetheless we cannot get useful mechanisms
whose social welfare scales with 
the bid distribution (ignoring logarithmic terms).

\ifdefined\unbounded

\subsection{Impossibility for Approximate Incentive Compatibility: Unbounded Bids}
\label{section:impossibility-approx}
Unfortunately, 
relaxing to approximate incentive compatibility
does not help us overcome the finite-block size
impossibility if bids can be unbounded. 
Formally, we have the following theorem.

\begin{theorem}
\label{thm:impossible-approx-plain-finite}
Suppose the block size is finite.
For any $\eps>0$,
no (possibly randomized) non-trivial TFM 
in the plain model
can simultaneously satisfy $\eps$-UIC and $\eps$-SCP, 
even if the miner colludes with only one user.
\end{theorem}

The rest of \cref{section:impossibility-approx} is dedicated to proving \cref{thm:impossible-approx-plain-finite}.
Before proving \cref{thm:impossible-approx-plain-finite}, we prove two useful lemmas.
The following lemma stating that for any non-trivial TFM and for any $\Delta > 0$, there always exists a user whose utility is at least $\Delta$ as long as the miner is honest.

\begin{lemma}
\label{lemma:impossible-approx-plain-finite}
Suppose a plain-model TFM $\mecha$ satisfies $\eps$-UIC for some $\eps \geq 0$.
For any $\Delta > 0$, if there exists a bid vector $\bfb = (b_1,\ldots,b_n)$ such that $x_i(\bfb) > 0$ for some $i \in [n]$, then there exists a true value $v_i$ of user $i$ such that $v_i \cdot x_i(\bfb_{-i}, v_i) - p_i(\bfb_{-i}, v_i) > \Delta + 2\cdot \max\left(\eps,\sqrt{v_i \cdot \eps}\right)$.
\end{lemma}
\begin{proof}
If user $i$'s bid is unconfirmed, the payment must be zero.
If user $i$'s bid is unconfirmed, the payment never exceeds the bid.
Thus, if user $i$ bids truthfully, its utility must be non-negative, so we have $r \cdot x_i(\bfb_{-i}, r) - p_i(\bfb_{-i}, r) \geq 0$ for all $r$ and $\bfb$.

By \cref{lemma:payment-sandwich} \hao{does it implies ex post case?}, for all $\bfb$ and for all $z > y$, we have \[
    z[x_i(\bfb_{-i}, z)-x_i(\bfb_{-i}, y)] \geq p_i(\bfb_{-i}, z) - p_i(\bfb_{-i},y) - \eps.    
\]
Moreover, we have \[
    z \cdot x_i(\bfb_{-i}, z) - y \cdot x_i(\bfb_{-i}, y) = z[x_i(\bfb_{-i}, z)-x_i(\bfb_{-i}, y)] + (z-y)\cdot x_i(\bfb_{-i}, y).  
\]
Combine two equations above, we have \[
    z \cdot x_i(\bfb_{-i}, z) - y \cdot x_i(\bfb_{-i}, y) \geq p_i(\bfb_{-i}, z) - p_i(\bfb_{-i},y) + (z-y)\cdot x_i(\bfb_{-i},y) - \eps.
\]
Re-arrange the equation, we have \[
    z \cdot x_i(\bfb_{-i}, z) - p_i(\bfb_{-i}, z) \geq y \cdot x_i(\bfb_{-i}, y) - p_i(\bfb_{-i},y) + (z-y)\cdot x_i(\bfb_{-i},y) - \eps.
\]
As we have shown, $y \cdot x_i(\bfb_{-i}, y) - p_i(\bfb_{-i},y) \geq 0$. 
Because $x_i(\bfb_{-i},y) > 0$, 
there must exist an real number $z$ such that
\begin{equation}\label{eq:bigutil}
  (z-y) \cdot x_i(\bfb_{-i}, y) - \eps > 2\cdot \max\left(\eps,\sqrt{z\eps}\right) + \Delta.
\end{equation}
Note that in the above step, we are using the fact
that there is no a-prori upper bound on the bids.
If Eq.(\ref{eq:bigutil}) holds, we have \[
    z \cdot x_i(\bfb_{-i}, z) - p_i(\bfb_{-i}, z) \geq 2\cdot \max\left(\eps,\sqrt{z\eps}\right) + \Delta.
\]
Thus, if user $i$'s true value is $z$, its utility is at least $2\cdot \max\left(\eps,\sqrt{z\eps}\right) + \Delta$.
\end{proof}

The following lemma states that any user's utility is bounded by the miner's revenue.

\begin{lemma}
    \label{lemma:user-util-bound-by-miner}
    Suppose the block size is upper bounded by $k$.
    Fix any $\eps>0$, suppose a (possibly randomized) plain-model TFM 
satisfies $\eps$-UIC and $\eps$-SCP (even when the miner colludes with only one user).
    Then, for any possible bid vector, if the miner revenue is upper bounded by $\tau$,
    every user's expected utility must be upper bounded by $\tau + \eps$.
\end{lemma}
\begin{proof}

    For the sake of contradiction, suppose there exists a bid $\bids$ and a user $i$ with value $v$ such that ${\sf util}_i(\bfb,v) > \tau + \eps$.
    In other words, when the world is $(\bids, v)$ and the miner is honest, 
    user $i$'s expected utility is $\tau + \eps + \delta$ for some $\delta > 0$.
    
    Let $N = \lceil\frac{vk}{\delta}\rceil + 1$.
    Imagine that the world is $\widetilde{\bids}$ of length $N+|\bids|$ where
    \[\widetilde{\bids} = (\bids, \underbrace{v,v,\dots,v}_{N}).\]
    Because the block size is upper bounded by $k$, there must exist a user $j$ whose bid is $v$ while its confirmation is $x \leq \frac{k}{N}$. 
    Therefore, if user $j$ bids truthfully, its true value is $v$.
    Thus, user $j$'s utility is at most $v\frac{k}{N}< \delta$.
    By assumption, the miner revenue is at most $\tau$.
    Thus, when the miner is honest, the joint utility of the miner and user $j$ is strictly less than $\tau + \delta$.
    However, the miner can collude with user $j$.
    The coalition pretends that the world is $(\bids, v)$, and user $j$ plays the role of user $i$.
    By assumption, if the world is $(\bids, v)$, user $i$'s utility is $\tau + \eps + \delta$.
    Because the coalition does not inject any fake bid, the miner's utility is at least zero.
    Thus, by deviating from the mechanism, the joint utility of the coalition becomes at least $\tau + \eps + \delta$, which increases strictly more than $\eps$.
    It contradicts to $\eps$-SCP.
\end{proof}


Now we are ready to prove \cref{thm:impossible-approx-plain-finite}.

\paragraph{Proof of \cref{thm:impossible-approx-plain-finite}.}
For the sake of reaching a contradiction, suppose there exists a non-trivial TFM $\mecha$ that satisfies $\eps$-UIC and $\eps$-SCP for some $\eps$.
Because $\mecha$ is non-trivial, there must exist a bid vector $\bfb = (b_1,\ldots,b_n)$ such that $x_i(\bfb) > 0$ for some $i \in [n]$.

By \cref{lemma:impossible-approx-plain-finite}, there exists a value $v_i$ such that \[
    v_i \cdot x_i(\bfb_{-i}, v_i) - p_i(\bfb_{-i}, v_i) > \mu(\bfb_{-i},0) + \eps + 2\cdot \max\left(\eps,\sqrt{v_i \cdot \eps}\right).
\]
Imagine that user $i$'s true value is $v_i$.
When user $i$ bids truthfully, the miner revenue is $\mu(\bfb_{-i},v_i)$.
However, by \cref{lem:miner-eps-UIC}, we have \[
\mu(\bfb_{-i},v_i) - \mu(\bfb_{-i},0) \leq 2 \cdot \max\left(\eps,\sqrt{v_i \cdot \eps}\right).
\]
Consequently, user $i$'s utility is larger than the miner's utility by more than $\eps$.
By \cref{lemma:user-util-bound-by-miner}, $\mecha$ cannot satisfy $\eps$-UIC and $\eps$-SCP.

\else
\fi

\label{section:bounded-dist}

\ignore{
The impossibility in \Cref{thm:impossible-approx-plain-finite}
relies on the bid distribution being unbounded.
One potential way to get round this impossibility is to assume
that there is an a-priori known bound on the maximum bid.
Indeed, this could allow us to get around the impossibility
of \Cref{thm:impossible-approx-plain-finite} 
in the plain model, 
as shown in \Cref{section:staircase}.  \elaine{add ref}
However, even for bounded bid distribution, we show that 
relaxing to approximate incentive compatibility
cannot help us in designing a class of useful 
mechanisms 
whose social welfare (i.e., the sum of everyone's utilities)
scales up w.r.t. the bid distribution (barring
polylogarithmic factors).
}

\begin{theorem}
    \label{thm:eps-impossible-approx-plain-finite}
    Suppose the block size is upper bounded by $k$.
    Fix any $\eps>0$. Given any TFM in the plain model 
that satisfies $\eps$-UIC, $\eps$-MIC and $\eps$-SCP when the miner can collude with at most $c=1$ user, 
and given any bid vector ${\bf b}$, let $M = \max({\bf b})$ be the maximum
bid of any user, 
it must be that 
\begin{itemize}[leftmargin=5mm,itemsep=1pt]
\item 
the miner's expected revenue is upper bounded by $12 k^2 \eps\log\left(\frac{M}{\epsilon}+1\right) + 2k\eps$;
\item 
every user's expected utility is upper bounded by $12 k^2 \eps\log\left(\frac{M}{\epsilon}+1\right) + (2k+1)\eps$
conditioned on the bid being included in the block, and  
assuming 
the bid reflects its true value;
\item 
the expected social welfare
is upper bounded by $O\left(k^3 \eps \log \left(\frac{M}{\epsilon}+1\right) + k^2\eps\right)$. 
\end{itemize}
\end{theorem}

A direct corollary of \cref{thm:eps-impossible-approx-plain-finite} is that there is no non-trivial mechanism that satisfies approximate incentive compatibility if the user's true value is unbounded. 
This implies that there is no universal mechanism that works for all bid distributions.
Formally,
\begin{corollary}
Suppose the block size is upper bounded by $k$.
Fix any $\eps>0$.
If users' true values are unbounded,
then no (possibly randomized) non-trivial TFM in the plain model can simultaneously satisfy $\eps$-UIC and $\eps$-SCP, even if the miner colludes with only one user.
\end{corollary}
\begin{proof}
    For the sake of contradiction, assume that there exists an $\eps > 0$, such that there exists a non-trivial TFM satisfying $\eps$-UIC and $\eps$-SCP. 
    Recall that $x_i(\bids)$ denotes the probability of user $i$'s bid being confirmed given that the world consists of the bid vector $\bids$ (assuming the mechanism is honestly implemented).
    We define $\widetilde{x}_i(\bids')$ to be the probability of user $i$'s bid being confirmed conditioned on its bid being included in the block configuration $\bids'$. 
    According to the assumption that the mechanism is non-trivial, there must exist an $i\in[k]$ and a block configuration $\bids' = (b^*,\bids_{-i})$ such that $b^*$ has a positive probability $\widetilde{x}_i(\bids')$ of being confirmed.
    
    Now imagine the world consists of the bid vector $\bids$ where 
    \[{\bids} = (b_1,b_2,\dots,b_{k-1}, \underbrace{M,M,\dots,M}_{T}),\]
    where $T\geq \frac{2k}{\widetilde{x}_i(\bids')}$ and $M$ is some large number (larger than $\max\{b_1,\dots,b_k\}$) that we will specify later.
    
    Since the block size is bounded by $k$, there must exist a user $j$ whose true value is $M$ yet its probability of being confirmed is no more than $\frac{k}{T} \leq \frac{1}{2}\widetilde{x}_i(\bids')$ by our choice of $T$.
    Therefore, user $j$'s utility (assuming the mechanism is honestly implemented) is at most $M\cdot \frac{1}{2}\widetilde{x}_i(\bids')$.
    Now consider the coalition of the miner and user $j$.
    By \cref{thm:eps-impossible-approx-plain-finite}, their joint utility when behaving honestly is at most 
    \[M\cdot \frac{1}{2}\widetilde{x}_i(\bids') + 12 k^2 \eps\log\left(\frac{M}{\epsilon}+1\right) + 2k\eps.\]
    
    However, the miner can ask user $j$ to bid $b^*$ instead of its true value $M$ and include $(b_1,\dots,b_{k-1},b^*)$ into the block, where the bid $b^*$ comes from user $j$.
    Since the payment cannot exceed the bid, now the utility of user $j$ is at least
    \[M\cdot\widetilde{x}_i(\bids') - b^*.\]
    As long as $M$ is large enough such that 
    \[M\cdot\widetilde{x}_i(\bids') - b^* \geq M\cdot \frac{1}{2}\widetilde{x}_i(\bids') + 12 k^2 \eps\log\left(\frac{M}{\epsilon}+1\right) + 2k\eps + \eps,\]
    the coalition gains $\eps$ more joint utility comparing to honest strategy.
    This contradicts $\eps$-SCP. 
    Note that since user's true value can be unbounded, such $M$ must exist.
    Therefore, there does not exist a non-trivial mechanism that satisfies $\eps$-UIC and $\eps$-SCP simultaneously.
    \ke{one bug in this proof: when the $b^*$ in the block is replaces with user $j$'s bid, user $j$ may not get the same confirmation probability.}
\end{proof}

The rest of \cref{section:bounded-dist} is dedicated to proving \cref{thm:eps-impossible-approx-plain-finite}.

\subsection{Proof Roadmap}
We first explain the blueprint. 
To prove that the total social welfare 
is small, we first show that the miner revenue 
must be $\widetilde{O}(k^2 \epsilon)$ for any bid configuration.
If we can show this, then given that 
the block size is finite, we can show that every user $i$'s utility
conditioned on being included is small, which then
allows us to bound the total social welfare. Suppose
this is not the case, i.e., suppose 
that under some bid configuration ${\bf b} := (b_1, \ldots, b_N)$,   
there is a user $i$ with expected utility (conditioned on being  
included) significantly larger 
than the maximum possible expected miner revenue
(which is upper bounded by $\widetilde{O}(k^2 \epsilon)$).
Then, imagine a world consisting of ${\bf b}$ and additionally  
(infinitely) 
many users whose true value is the same as $b_i$. 
In this case, there must be one such user $j$ whose expected utility is almost $0$. 
Thus, if $j$ is the miner's colluding friend, 
the miner would be willing to sacrifice all of its
revenue, pretend that the world consists of ${\bf b}$ 
where the $i$-th coordinate is replaced with $j$'s bid, and
run the honest mechanism subject 
to $j$ being included. 
In this case, the coalition can increase its expected joint utility
since user $j$ would be doing much better than the honest case.

The crux of our proof, therefore, is to show 
that the 
expected miner revenue must be bounded
for any bid vector. 
To show this, we take two main steps.
First, we show that if the world consists of only
bids of value $M$, the expected miner revenue must be small 
(see \cref{lemma:base-case}).
Using the above as base case, we then go through
an inductive argument to show that
in fact, for any bid vector
where users do not necessarily bid $M$,  
the miner revenue must be small too (see \cref{lemma:eps-impossible-approx-plain-finite}).
Note that showing the first step itself
relies on another inductive argument that inducts
on the length of the bid vector.

\subsection{Detailed Proof}

\subsubsection{Individual User's Influence on Miner Revenue is Bounded}

Before proving \cref{thm:eps-impossible-approx-plain-finite}, 
we introduce some useful lemmas. \ke{The miner revenue bound is below}
The following lemma states that 
if, given some bid configuration, a user's expected utility  
is not too large, then,  
the miner's expected revenue should not drop  
too much when we lower that user's bid to $0$. 

\begin{lemma}
\label{lemma:miner-rev-bound-by-user}
Given 
any (possibly randomized) TFM in the plain model 
that satisfies $\eps$-UIC, $\eps$-MIC and $\eps$-SCP against $1$-sized coalition,
for any $\bids_{-i}$ and $v$, we have the following
where  
${\sf util}^i(\bids)$
denotes user $i$'s expected utility
and $\mu(\bids)$
is the expected miner revenue 
when the bid vector is $\bids$:
\begin{equation*}
    \mu(\bids_{-i}, v) - \mu(\bids_{-i}, 0) \leq
    \begin{cases}
    4\epsilon, & v\leq 2\epsilon\\
{\sf util}^i(\bids_{-i},v)
+3\eps\log\frac{v}{\epsilon}+4\epsilon, & v>2\epsilon.
    \end{cases}
\end{equation*}
\end{lemma}
\begin{proof}
Henceforth, we use ${\bf x}({\bf b})$ to denote the 
vector of probabilities that each bid in ${\bf b}$
is included and confirmed, and let 
${\bf p}({\bf b})$ denote the vector of expected payments 
for every user when  
the bid vector is ${\bf b}$.

First, observe that \Cref{lemma:payment-sandwich} 
and \Cref{eqn:payment-sandwich-plain} still hold in the plain model where the 
terms $\overline{x}_i(\cdot)$, $\overline{p}_i(\cdot)$, and $\overline{\mu}(\cdot)$
are now replaced with $x_i({\bf b}_{-i}, \cdot)$ 
$p_i({\bf b}_{-i}, \cdot)$, and $\mu({\bf b}_i, \cdot)$
respectively, i.e., we now fix an arbitrary fixed ${\bf b}_{-i}$
rather than taking expectation 
over the random choice ${\bf b}_{-i}$.

Specifically, \Cref{lemma:payment-sandwich} implies that 
for any ${\bf b}_{-i}$, 
for any $b \leq b'$,
\begin{equation}
    b'\cdot [x_i({\bf b}_{-i}, b')-x_i({\bf b}_{-i}, b)]+\eps\geq p_i({\bf b}_{-i}, b') - p_i({\bf b}_{-i}, b) \geq b\cdot [x_i({\bf b}_{-i}, b')-x_i({\bf b}_{-i}, b)]-\epsilon.
    \label{eqn:payment-sandwich-plain}
\end{equation}
\Cref{lem:miner-step-eps-UIC} implies that 
for any ${\bf b}_{-i}$,
for any $b \leq b'$,
\begin{equation}
    \mu(\bids_{-i}, b') - \mu(\bids_{-i}, b) \leq 2\epsilon + (b' - b)\cdot [x_i({\bf b}_{-i}, b') - x_i({\bf b}_{-i}, b)].
    \label{eqn:miner-rev-step-plain}
\end{equation}

Henceforth in this proof, we always fix an arbitrary $\bids_{-i}$.
For simplicity, in this proof, we omit $\bids_{-i}$ 
and use the short-hand notations 
$x_i(v) := x_i(\bids_{-i}, v)$, $p_i(v) := p_i(\bids_{-i}, v)$,
and $\mu(v) := \mu(\bids_{-i}, v)$.

For $v\leq 2\epsilon$, the lemma directly follows 
from~\eqref{eqn:miner-rev-step-plain}.
In the rest of the proof, we focus on the case where $v > 2\epsilon$.
Define a function $u_i(b)$ such that $\int_{0}^b u_i(t)dt = b\cdot x_i(b) - p_i(b)$.
For any $b\leq b'$, the payment when bidding $b$ is 
\[p_i(b) = b\cdot x_i(b) - \int_{0}^b u_i(t)dt.\]
Since we do not have the guarantee that the utility increases with the bids, it can be that $u_i(b) \leq 0$ for some $b$.
However, we have that guarantee that at any point, $\int_{0}^b u_i(t)dt$ is non-negative.

By \Cref{eqn:payment-sandwich-plain}, we know that for any $b\leq b'$, we have $p_i(b') - p_i(b)\leq b'[x_i(b')-x_i(b)]+\eps$, i.e.,
\begin{equation*}
    \left[b'\cdot x_i(b') - \int_{0}^{b'} u_i(t)dt\right] - \left[b\cdot x_i(b) - \int_{0}^b u_i(t)dt\right]\leq b'\cdot[x_i(b')-x_i(b)]+\eps,
\end{equation*}
which is equivalent to 
\begin{equation}
    \xi(b,b'):=(b'-b)\cdot x_i(b) - \int_{b}^{b'} u_i(t) dt \leq \eps.
    \label{eqn:util-area}
\end{equation}
Intuitively, the meaning 
of $\xi(b, b')$ is 
how much we are over-estimating if we use a rectangle
of width $b'-b$ and height $x_i(b)$  
to approximate the area-under-curve\footnote{We may assume that any area under $0$ contributes
negatively to the area-under-curve.} for $u_i$, between $b$ and $b'$.
For example, the blue area in~\cref{fig:area-between-allo-and-util} 
represents $\xi(b,b')$, whereas the red area 
\emph{minus} 
the gray area is $\xi(b'',v)$.

Now consider the following sequence: $b_l = v - \frac{v}{2^l}$ for $l=0,\dots,L$ where $L = \lceil\log\frac{v}{2\epsilon}\rceil$.
By~\eqref{eqn:miner-rev-step-plain}, the miner revenue 
\[\mu(b_l) - \mu(b_{l-1})\leq 2\epsilon + S(b_{l-1}, b_l),\]
where $S(b_{l-1}, b_l) := (b_{l} - b_{l-1})\cdot [x_i(b_{l}) - x_i(b_{l-1})]$.
Summing up the miner revenue difference together, we have
\begin{align*}
    &\mu(v) - \mu(0) = \mu(v) - \mu(b_L) + \sum_{l=1}^L \mu(b_l) - \mu(b_{l-1})\\
    \leq &2\epsilon + (v - b_{L})\cdot [x_i(v) - x_i(b_{L})] + \sum_{l=1}^L \left(S(b_{l-1}, b_l)+2\eps\right) &\text{By~\eqref{eqn:miner-rev-step-plain}}\\
    \leq & 4\epsilon+2L\epsilon+\sum_{l=1}^L S(b_{l-1}, b_l). &\text{By }v-b_L\leq 2\epsilon
    \label{eqn:miner-rev-sum}
\end{align*}
Now we proceed to bound the sum $\sum_{l=1}^L S(b_{l-1}, b_l)$. 
For each $l = 1,\dots, L$, by the choice of the sequence, 
we have 
\[b_{l} - b_{l-1} = \frac{v}{2^l} = v - b_l, \quad \text{and} \quad
S(b_{l-1}, b_l) = (v-b_l)\cdot [x_i(b_{l}) - x_i(b_{l-1})]
\]
For simplicity, let $b_{L+1} := v$. We have the following: 
\begin{align*}
    &\sum_{l=1}^L S(b_{l-1}, b_l)
    =\sum_{l=1}^L (v-b_l)\cdot [x_i(b_{l}) - x_i(b_{l-1})]\\
    =& (v-b_L)\cdot x_i(b_L) + \sum_{l=1}^{L-1} (b_{l+1} - b_l)\cdot x_i(b_{l})\\
    =&\sum_{l=1}^L (b_{l+1} - b_l)\cdot x_i(b_{l}). &\text{By } v = b_{L+1} 
\end{align*}
In other words, the sum $\sum_{l=1}^L S(b_{l-1}, b_l)$ is equal
to the total area of the dashed rectangles \elaine{TODO: changed to dashed} in 
\Cref{fig:miner-rev-step}.
We want to show that the sum $\sum_{l=1}^L S(b_{l-1}, b_l)$
is not significantly greater than 
${\sf util}^i(v)$, i.e., the area under the $u_i$-curve. 
The follow calculation says that this difference is upper bounded by 
$\sum_{l = 1}^L \xi(b_l, b_{l+1})$. 
Formally, 
\begin{align*}
    &\sum_{l=1}^L S(b_{l-1}, b_l) - \int_{0}^v u_i(t)dt  = \sum_{l=1}^L (b_{l+1} - b_l) x_i(b_{l}) - \int_{0}^v u_i(t)dt\\
    \leq & \sum_{l=1}^L \left\{(b_{l+1} - b_l) \cdot x_i(b_{l})-\int_{b_l}^{b_{l+1}} u_i(t)dt\right\} & \text{By } \int_{0}^{b_1} u_i(t)\geq 0\\
=  & \sum_{l=1}^L \xi(b_l, b_l+1) \leq \sum_{l=1}^{L} \epsilon = L\eps. &\text{By~\eqref{eqn:util-area}}
\end{align*}
\ke{Add a label in figure}

Putting it together, the change in miner revenue $\mu(v) - \mu(0)$ is upper bounded by
\begin{align*}
    &\mu(v) - \mu(0) \leq 4\epsilon+2L\epsilon+\sum_{l=1}^L S(b_{l-1}, b_l)\\
    \leq &4\epsilon + 2L\epsilon + L\epsilon+\int_{0}^v u_i(t)dt \leq {\sf util}^i(\bids_{-i},v) +3\eps\log\frac{v}{\epsilon}+4\epsilon,
\end{align*}
where the last step comes from the fact that $L\leq \log\frac{v}{\eps}$ by our choice of $L$.
\ignore{
In~\cref{fig:miner-rev-step}, we give an illustrative example of the proof. 
The sum of the area of rectangles represents the sum $\sum_{l=1}^L S(b_{l-1}, b_l)$.
While the red area is at most ${\sf util}^i(\bids_{-i},v)$ by the definition of the utility curve, the blue area (or blue area minus gray area when the utility curve goes above the rectangle) in each rectangle is at most $\eps$, according to~\eqref{eqn:util-area}.
}
\end{proof}

\begin{figure}[H]
    \centering
    \begin{subfigure}[t]{0.45\textwidth}
    \begin{tikzpicture}[trim axis left]
        \begin{axis}[grid=major, xmin=0, xmax=2, ymin=-0.1, ymax=1.1, 
             xtick = {1,1.5,1.7,2}, xticklabels={$b$, $b'$, $b''$, $v$},
             ytick = {0,1}, yticklabels={$0$, $1$}, legend pos=north west]
        \addplot[black, thick, samples=300, smooth,domain=0:2, name path = A] {0.25*x^2};\addlegendentry{$x_i(\cdot)$};
        \addplot[blue, thick, samples=300, smooth,domain=0:2, name path = A1] {x*(x-1)*(x-1.5)};\addlegendentry{$u_i(\cdot)$};
        \addplot[draw=none,name path=B] {0.25};
        \addplot[draw=none,name path=C] {0.25*1.7*1.7};
        \addplot[blue!30] fill between[of=A1 and B,soft clip={domain=1:1.5}];
        \addplot[red!30] fill between[of=A1 and C,soft clip={domain=1.7:1.92}];
        \addplot[black!30] fill between[of=A1 and C,soft clip={domain=1.92:2}];
        \end{axis}
        \draw (3,2) node {$\xi(b,b')$};
        \draw[->] (3.5,1.8) -- (4.3,1.3);
        \draw (5.5,4.5) node {$\xi(b'',v)$};
        \draw[->] (5.5,4.2) -- (6.3,3.6);
        \end{tikzpicture}
        \caption{
The blue area denotes $\xi(b,b')$, and the red area minus the gray area
denotes $\xi(b'', v)$.
}
        \label{fig:area-between-allo-and-util}
    \end{subfigure}
    \hfill
    \begin{subfigure}[t]{0.45\textwidth}
    \begin{tikzpicture}[trim axis left]
        \begin{axis}[grid=major, xmin=0, xmax=2, ymin=-0.1, ymax=1.1,
             xtick = {0,1,1.5,1.75,2}, xticklabels={$b_0$, $b_1$, $b_2$, $b_3$, $v$},
             ytick = {0,1}, yticklabels={$0$, $1$}, legend pos=north west]
        \addplot[black, thick, samples=300, smooth,domain=0:2, name path = A1] {0.25*x^2};\addlegendentry{$x_i(\cdot)$};
        \addplot[blue, thick, samples=300, smooth,domain=0:2, name path = A] {x*(x-1)*(x-1.5)};\addlegendentry{$u_i(\cdot)$};
        \addplot[draw=none,name path=B1] {0.25};
        \addplot[draw=none,name path=B2] {0.25*1.5*1.5};
        \addplot[draw=none,name path=B0] {0};
        \addplot[draw=none,name path=B3] {0.25*1.75*1.75};
        \addplot[blue!30] fill between[of=B1 and A,soft clip={domain=1:1.5}];
        \addplot[blue!30] fill between[of=B2 and A,soft clip={domain=1.5:1.75}];
        \addplot[blue!30] fill between[of=B3 and A,soft clip={domain=1.75:1.93}];
        \addplot[black!30] fill between[of=A and B3,soft clip={domain=1.92:2}];
        \draw[thick,dashed] (1,0) rectangle (1.5,0.25);
        \draw[thick,dashed] (1.5,0) rectangle (1.75,0.25*1.5*1.5);
        \draw[thick,dashed] (1.75,0) rectangle (2,0.25*1.75*1.75);
        \end{axis}
        \draw (4.3,3.6) node {$\xi(b_2,b_3)$};
        \draw[->] (4.3,3.4) -- (5.5,2.7);
        \draw (5.5,4.5) node {$\xi(b_3,v)$};
        \draw[->] (5.5,4.3) -- (6.4,3.8);
        \draw (2.5,2.3) node {$\xi(b_1,b_2)$};
        \draw[->] (2.5,2.1) -- (4.3,1);
        \end{tikzpicture}
        \caption{
The sum of the dashed rectangles is equal to $\sum_{l = 1}^L S(b_l, b_{l+1})$. 
The difference between $\sum_{l = 1}^L S(b_l, b_{l+1})$
and the area under the $u_i(\cdot)$ curve is upper bounded by 
$\sum_{l = 1}^L \xi(b_l, b_{l+1})$, represented
by the sum of the blue
areas minus the gray area.
}
        \label{fig:miner-rev-step}
    \end{subfigure}
    \caption{Graphical explanation of the proof to~\cref{lemma:miner-rev-bound-by-user}}
    \label{fig:miner-rev}
\end{figure}
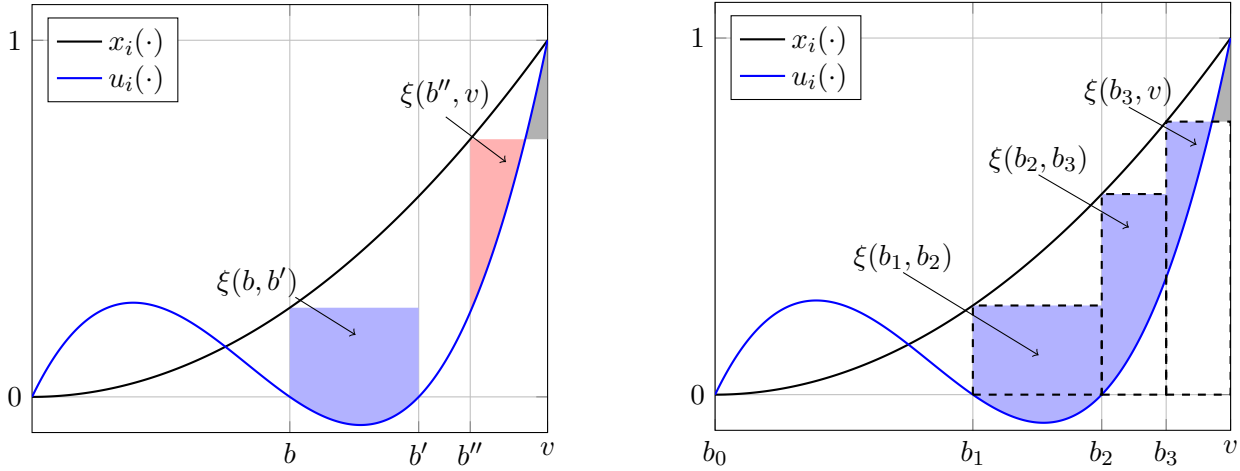

Because the miner can inject a bid $0$ for free, 
\cref{lemma:miner-rev-bound-by-user} implies the following corollary, 
which says that if we remove a bid, the miner revenue should not be affected by too much.

\begin{corollary}
    \label{corollary:miner-rev-bound-by-user}
    Let $\mecha$ denote any (possibly randomized) TFM
in the plain model that satisfies $\eps$-UIC, $\eps$-MIC and $\eps$-SCP against $1$-sized coalition.
    For any $\bids_{-i}$ and $v$, 
    \begin{equation*}
        \mu(\bids_{-i}, v) - \mu(\bids_{-i}) =
        \begin{cases}
        5\epsilon, & v\leq 2\epsilon\\
{\sf util}^i(\bids_{-i},v)
+3\eps\log\frac{v}{\epsilon}+5\epsilon, & v>2\epsilon.
        \end{cases}
    \end{equation*}
\end{corollary}
\begin{proof}
Because the miner can inject a bid $0$ for free, by $\eps$-MIC, it must be 
\begin{equation}
    \mu(\bids_{-i}, 0) - \mu(\bids_{-i}) \leq \eps.
\label{eqn:inject0} 
\end{equation}
The corollary is now directly implied by 
\Cref{eqn:inject0}
and \cref{lemma:miner-rev-bound-by-user}.
\end{proof}

\subsubsection{Bounds on Miner Revenue}
We now prove bounds for the miner's revenue. 
To do this, we first 
prove a bound on miner revenue when everyone bids the same value $M$
(see \Cref{lemma:base-case}).
Then, we generalize to the case when everyone's bids 
need not be the same (see \Cref{lemma:eps-impossible-approx-plain-finite}).

\paragraph{Notation.}
Henceforth, 
for $t \in \N \cup \{0\}$, we define $\bfm_t := (M,\ldots,M)$ where $|\bfm_t| = t$;
that is, $\bfm_t$ 
consists of $t$ copies of $M$.
Recall that $\mu(\bfb)$ denote the expected miner revenue given that
the world consists of the bid vector $\bfb$ (assuming the mechanism is honestly implemented).
We define $\tildemu(\bfb')$ to be the expected miner revenue 
given that \emph{the block configuration} is $\bfb'$.

\elaine{I changed the lemma statement}
\begin{lemma}
\label{lemma:base-case}
Suppose that the block size is upper bounded by $k$. Fix an arbitrary 
any $\eps>0$ and 
$M > 2\epsilon$ and let $\bfm_t := (M, M, \ldots, M)$ 
be a vector containing $t$ repetitions of $M$.  Then, 
for any (possibly randomized) TFM in the plain model 
that satisfies $\eps$-UIC, $\eps$-MIC and $\eps$-SCP even when
the miner colludes with at most $c = 1$ user, it holds that
$\tildemu(\bfm_t) \leq 12 k^2 \eps\log\frac{M}{\epsilon}$
for all $t \leq k$.
\ke{block size is $k$ or block size is upper bounded by $k$?}
\end{lemma}
\begin{proof}
Imagine 
the world consists of the bid vector $\bfm_K$ where $K > \frac{Mk}{\eps}$
is sufficiently large.
Let $\bfm_{t^*}$ be the block configuration that gives the miner optimal revenue; that is $t^* = \myargmax_{t\leq k} \tildemu(\bfm_t)$.
Clearly, it must be $\tildemu(\bfm_{t^*}) \geq \mu(\bfm_K)$.
Because of $\eps$-MIC, we have $\mu(\bfm_{t^*}) \geq \tildemu(\bfm_{t^*}) - \eps$.
Otherwise, if $\mu(\bfm_{t^*}) < \tildemu(\bfm_{t^*}) - \eps$, when the world is $\bfm_{t^*}$, the miner could simply choose $\bfm_{t^*}$ as the block configuration so that the revenue becomes $\tildemu(\bfm_{t^*})$, which is more than $\epsilon$ higher 
than its honest utility 
$\mu(\bfm_{t^*})$.
Combining the two inequalities, we have $\mu(\bfm_{t^*}) \geq \mu(\bfm_K) - \eps$.


Recall that ${\sf util}^i (\bids)$ denotes user $i$'s expected utility when the bid vector is ${\bids}$. 
Next, we will show that for any $t \leq K$ and any user $i \in [t]$, it must be 
\begin{equation}
\mu(\bfm_t) + {\sf util}^i (\bfm_t) \leq \mu(\bfm_K) + 2\eps.
\label{eq:base-case}
\end{equation}
For the sake of reaching a contradiction, suppose there is an integer $t$ and user $i$ such that $\mu(\bfm_t) + {\sf util}^i (\bfm_t) > \mu(\bfm_K) + 2\eps$.
Imagine that the world is $\bfm_K$, where $K > \frac{Mk}{\eps}$. 
There must exist a user $j$ whose confirmation probability is at most $x_j(\bfm_K) \leq \frac{k}{K} < \frac{\eps}{M}$, as at most $k$ bids can be included in a block.
Therefore, user $j$'s utility is at most ${\sf util}^j(\bfm_K) \leq x_j(\bfm_K) \cdot M < \eps$.
Imagine that the miner now colludes with user $j$.
The miner implements the inclusion rule as if the world consists of 
the bid vector $\bfm_t$ 
where the $i$-th position is occupied by user $j$'s bid.
Since the TFM is symmetric, and both users bid $M$, 
user $j$'s expected utility 
is now ${\sf util}^i(\bfm_t)$.
The joint utility of the coalition now is $\mu(\bfm_t) + {\sf util}^i (\bfm_t) > \mu(\bfm_K) + 2\eps > \mu(\bfm_K) + {\sf util}^j(\bfm_K) + \eps$, which contradicts $\eps$-SCP.
Consequently, 
\Cref{eq:base-case} must hold for any $t \leq K$ and any user $i \in [t]$.

According to \Cref{eq:base-case}, we have $\mu(\bfm_{t^*}) + {\sf util}^i (\bfm_{t^*}) \leq \mu(\bfm_K) + 2\eps$ for any user $i$.
As we have shown, it must be $\mu(\bfm_{t^*}) \geq \mu(\bfm_K) - \eps$.
Combining these two inequalities, we have 
\[{\sf util}^i (\bfm_{t^*}) \leq \mu(\bfm_K) + 2\eps -\mu(\bfm_{t^*})\leq \mu(\bfm_K) + 2\eps - \mu(\bfm_K) + \eps = 3\eps.\]
Since the utility of user $i$ is bounded, by applying \cref{corollary:miner-rev-bound-by-user}, it must be
\begin{equation}\label{eq:base-case2}
    \mu(\bfm_{t^*}) - \mu(\bfm_{t^* - 1}) \leq {\sf util}^i (\bfm_{t^*}) + 3\eps\log\frac{M}{\epsilon} + 5\eps \leq 8\eps + 3\eps\log\frac{M}{\epsilon}.
\end{equation}

Consequently, we have
\begin{align*}
    &{\sf util}^i (\bfm_{t^*-1}) \leq \mu(\bfm_K) + 2\eps - \mu(\bfm_{t^*-1}) & \text{By~\eqref{eq:base-case}}\\
    \leq &\mu(\bfm_K) + 2\eps - \mu(\bfm_{t^*})  + 8\eps + 3\eps\log\frac{M}{\epsilon}& \text{By~\eqref{eq:base-case2}}\\
    \leq &\mu(\bfm_K) + 2\eps - \mu(\bfm_{K}) +\eps  + 8\eps + 3\eps\log\frac{M}{\epsilon} &\text{By }\mu(\bfm_{t^*}) \geq \mu(\bfm_K) - \eps\\
    = &11\eps + 3\eps\log\frac{M}{\epsilon}.
\end{align*}
Then, we can apply \cref{corollary:miner-rev-bound-by-user} again, and we have \[
    \mu(\bfm_{t^* - 1}) - \mu(\bfm_{t^* - 2})\leq  {\sf util}_i (\bfm_{t^*-1}) + 3\eps\log\frac{M}{\epsilon} + 5\eps
    \leq 16\eps + 6\eps\log\frac{M}{\epsilon}.
\]
By the same reason, for any $r \leq t^*$, we have
\begin{equation}\label{eq:base-case3}
    \mu(\bfm_{t^* - r}) - \mu(\bfm_{t^* - r - 1}) \leq (8r + 8) \eps + (3r+3) \cdot \eps\log\frac{M}{\epsilon}.
\end{equation}
Since $M \geq 2\eps$, we have $ \eps\log\frac{M}{\epsilon}\geq \eps$.
By Eq.(\ref{eq:base-case3}), we have 
\begin{align*}
    &\mu(\bfm_{t^*}) - \mu(\bfm_0) =
    \sum_{r=0}^{t^*-1} \mu(\bfm_{t^* - r}) - \mu(\bfm_{t^* - r - 1}) \\
    \leq &\left(8t^* + 4(t^* -1)t^* \right)\eps + \left(3t^* + \frac{3(t^*-1)t^*}{2}\right)\cdot\eps\log\frac{M}{\epsilon}\\
    \leq &11(t^*)^2 \eps\log\frac{M}{\epsilon}. &\text{By } \eps\log\frac{M}{\epsilon}\geq \eps \text{ and } t^*\geq 1
\end{align*}
Notice that $\mu(\bfm_0) = 0$, so we have \[
    \mu(\bfm_{t^*}) \leq  11(t^*)^2 \eps\log\frac{M}{\epsilon}.
\]

Recall that we define $t^* = \myargmax_{t\leq k} \tildemu(\bfm_t)$.
By definition, $\tildemu(\bfm_t) \leq \tildemu(\bfm_{t^*})$ for all $t \leq k$.
As we have shown at the beginning, it must be $\mu(\bfm_{t^*}) \geq \tildemu(\bfm_{t^*}) - \eps$.
Thus, we have $\tildemu(\bfm_t) \leq \tildemu(\bfm_{t^*}) \leq \mu(\bfm_{t^*}) + \eps$ for all $t \leq k$.
Combine the arguments above, we have $\tildemu(\bfm_t) \leq 11k^2 \eps\log\frac{M}{\epsilon}+\eps\leq 12k^2 \eps\log\frac{M}{\epsilon}$ for all $t \leq k$.

\end{proof}

\begin{lemma}
    \label{lemma:eps-impossible-approx-plain-finite}
    Suppose the block size is upper bounded by $k$.
    Fix any $\eps>0$. 
    For any (possibly randomized) TFM in the plain model 
that satisfies $\eps$-UIC, $\eps$-MIC and $\eps$-SCP (even when the miner only colludes with one user),
    for any block configuration $\bfb$, 
the following must hold where $M$ is the maximum bid amount in the bid vector $\bfb$: 
\[
    \tildemu(\bfb) \leq 
    \begin{cases}
        2k \eps, &\text{if } M < 2\eps,\\
        12 k^2 \eps\log\frac{M}{\epsilon} + 2k\eps, &\text{if } M \geq 2\eps.
    \end{cases}    
    \]
\end{lemma}
\begin{proof}


Given any block configuration $\bfb$,
the miner revenue must be upper bounded by the sum of the bids in $\bfb$.
Thus, if $M < 2\eps$, the miner revenue is upper bounded by $2k\eps$.

Henceforth, we focus on the case $M \geq 2\eps$.
Throughout the proof, we say that a bid $b$ is a \emph{low bid} if $b < M$.
Then, any block configuration, up to reordering, can be 
represented by $(\bfm_t, \bfL)$ for some $t \geq 1$, where 
$\bfm_t$ consists of $t$ repetitions of $M$, 
$\bfL$ which is possibly of length $0$, contains only low bids.
We prove the following claim by induction on the length of $\bfL$: 
\begin{itemize}
    \item[] {\it For any $\bfL$ consisting of only low bids, 
for any $t$ such that $t + |\bfL| \leq k$, 
the miner revenue $\widetilde{\mu}(\bfm_t, \bfL) \leq \tau+2|\bfL|\eps$, where  we set $\tau:= 12 k^2 \eps\log\frac{M}{\epsilon}$.}
\end{itemize}
For the base case where $|\bfL| = 0$, i.e.~the block does not contain any low bid, it is proven by \cref{lemma:base-case}.

Now, suppose we have proven that for any $\bfL'$ of length $R$, for any $t$, the miner revenue $\widetilde{\mu}(\bfm_t, \bfL')\leq \tau+2R\eps$. 
We are going to show that for any $\bfL$ of length $R+1$, for any $t$, the miner revenue $\widetilde{\mu}(\bfm_t, \bfL)\leq \tau+2(R+1)\eps$. 

For the sake of contradiction, suppose there exists a bid $\bfL$ of length $R+1$ and there exists a $t$, such that for the block configuration $(\bfm_t, {\bf L}) = (\bfm_t, d_1,\ldots,d_R,d_{R+1})$, the miner's revenue is $\tau + 2(R+1)\eps + \delta$ for some $\delta > 0$.
Now, imagine that the world consists of $(\bfm_K, d_1,\ldots,d_R)$, where $K > \frac{kM}{\eps}$.
In this case, the block configuration output by the honest 
inclusion rule must be of the form $(\bfm_{t^*}, \bfd)$ for some $t^* \leq k - |\bfd|$ 
and $\bfd \subseteq \{d_1,\ldots,d_R\}$ consists
of only low bids.
Since $(\bfm_{t^*}, \bfd)$ only contains at most $R$ low bids, 
the miner revenue $\tildemu(\bfm_{t^*}, \bfd) \leq \tau + 2R\eps$ by induction hypothesis.

By our choice of $K$, there must exist a user $i$ with true value $M$, whose confirmation probability $x_i(\bfm_K, d_1,\ldots,d_R) \leq \frac{k}{K} < \frac{\eps}{M}$ when the miner is honest.
\hao{We choose $K$ like this.}
Thus, user $i$'s utility is at most $M \cdot x_i(\bfm_K, d_1,\ldots,d_R) < \eps$.
Now the miner can collude with user $i$, ask user $i$ to bid $d_{R+1}$ instead of its true value $M$ and include $(\bfm_t, d_1,\ldots,d_R,d_{R+1})$ in the block.
Since $d_{R+1} < M$ and the payment never exceeds the bid, user $i$'s utility is at least zero. 
This implies that the decrease of the utility of user $i$ is strictly less than $\eps$.
Now the miner revenue is $\tau + 2(R+1)\eps + \delta$ by our assumption,
whereas the miner revenue in the honest case is at most $\tau + 2R\eps$.
Thus, the miner revenue increases by more than $2\eps$ compared to the honest case. 
Thus, the joint utility of the coalition increases by more than $\eps$, which contradicts $\eps$-SCP.
Therefore, by induction, we have that $\mu(\bfm_t, \bfL) \leq \tau+2|\bfL|\eps$ for any $\bfL$ and any $t$ where $|\bfL| + t \leq k$.
Finally, 
since $|\bfL| \leq k$,  
we conclude that $\tildemu(\bfb) \leq 12 k^2 \eps\log\frac{M}{\epsilon} + 2k\eps$.
\end{proof}

\subsubsection{Completing the Proof of \cref{thm:eps-impossible-approx-plain-finite}}
We now complete the proof of \Cref{thm:eps-impossible-approx-plain-finite}.
To do so, we prove that each user's utility conditioned on being included 
must be bounded
given that the miner revenue is bounded (see \Cref{lemma:user-util-bound-by-miner}), 
which then leads
to our conclusion that the total social welfare must be small.

\begin{lemma}
    \label{lemma:user-util-bound-by-miner}
    Suppose that the block size is upper bounded by $k$.
    Fix any $\eps>0$. 
    For any (possibly randomized) TFM in the plain model satisfies $\eps$-UIC and $\eps$-SCP (even when the miner colludes with only one user), 
for any bid vector $\bfb$ where $M := \max(\bfb)$, for and any user $i$, 
conditioned on user $i$ being included in the block, user $i$'s utility 
must be upper bounded by $U + \eps$
where $U = \underset{|\bfb'| \leq k, \max(\bfb') \leq M}{\max}\tildemu(\bfb')$,
i.e., $U$ is the maximum possible revenue 
the miner can 
get among all possible block configurations where 
all bids are at most $M$.
\end{lemma}
\begin{proof}

    For the sake of contradiction, 
suppose 
that under some bid vector $\bids'$ where all bids are at most $M$, some user $j$'s    
expected utility conditioned on being included in the block
is 
strictly more than $U + \epsilon$.
This implies that there must exist a block configuration 
$\bids = (b_1,\ldots,b_{|\bfb|})$ where all bids are at most $M$, 
and some $i \leq |\bfb|$, such that 
under conditioned on the block configuration being $\bids$, 
the $i$-th bid $b_i$ in the block 
has expected utility 
at least $U + \epsilon + \delta$ for some positive $\delta$.
    Let $T = \lceil\frac{b_i k}{\delta}\rceil + 1$.
    Imagine that the world consists of the bid vector 
${\bids'}$ of length $T+|\bids|$ where
    \[{\bids'} = (\bids, \underbrace{b_i,b_i,\dots,b_i}_{T}).\]
    Because the block size is upper bounded by $k$, there must exist a user $j$ whose bid is $b_i$ while its confirmation probability is at most $\frac{k}{T}$. 
    Therefore, if user $j$ bids truthfully, 
its utility is at most $b_i\cdot \frac{k}{T}< \delta$.
    By our assumption, the miner revenue is at most $U$ under any 
block configuration where bids are upper bounded by $M$.
    Thus, when behaving honestly, the miner and user $j$ have joint utility strictly less than $U + \delta$.
    However, the miner can collude with user $j$ and prepare the block where the block configuration is $\bfb$ and the $i$-th position is replaced with user $j$'s bid instead. 
    In this case, user $j$'s utility is $U + \eps + \delta$.
    Because the coalition does not inject any fake bid, the miner's utility is at least zero.
    Thus, by deviating from the mechanism, the joint utility of the coalition becomes at least $U + \eps + \delta$, which exceeds the honest case by more than $\eps$.
    This contradicts $\eps$-SCP.
\end{proof}

\paragraph{Proof of \cref{thm:eps-impossible-approx-plain-finite}.}
Suppose the world consists of an arbitrary 
bid vector $\bfb$.
Let $M = \max(\bfb)$.
If $M < 2\eps$, the miner can have at most $2k\eps$-miner revenue by \cref{lemma:eps-impossible-approx-plain-finite}.
For any user $i$ who is bidding truthfully, its true value must be upper bounded by $M$ since $M = \max(\bfb)$.
Moreover, each confirmed user's utility is at most its true value, which is upper bounded by $M < 2\eps$.
Since there are at most $k$ number of confirmed user,
the expected social welfare is 
$\sum_i {\sf util}^i({\bf b})$
plus the miner's expected utility,
which is upper bounded by $4k\eps$.

In the rest of the proof, we assume $M \geq 2\eps$ and we define
$\tildeutil^i({\bf b})$ to be the utility of user $i$
conditioned on being confirmed when the world
consists of the bid vector ${\bf b}$.
By \cref{lemma:eps-impossible-approx-plain-finite}, the miner can have at most $\left(12 k^2 \eps\log\frac{M}{\epsilon} + 2k\eps\right)$-miner revenue.
By \cref{lemma:user-util-bound-by-miner}, 
for any $i$,
$\tildeutil^i({\bf b}) \leq 12 k^2 \eps\log\frac{M}{\epsilon} + (2k+1)\eps$.
Let $\gamma_i$ be the probability that user $i$ is included in the block given the bid vector $\bids$.
Observe that $\sum_i \gamma_i \leq k$ for any $\bids$.
Therefore, the expected total utility of all users is upper bounded by 
\[
\sum_i {\sf util}^i({\bf b}) 
 =   
\sum_i \tildeutil^i({\bf b})  \cdot \gamma_i 
\leq  
\left(12 k^2 \eps\log\frac{M}{\epsilon} + (2k+1)\eps\right) \cdot \sum_{i} \gamma_i
= O\left(k^3 \eps \log \frac{M}{\eps}\right).
\]
The expected social welfare is 
$\sum_i {\sf util}^i({\bf b})$
plus the miner's expected utility. Clearly, it is also upper bounded by 
$O\left(k^3 \eps \log \frac{M}{\eps}\right)$.

Combine the argument above, because $\log\left(\frac{M}{\epsilon} + 1 \right)$ is always non-negative, the theorem follows.

\section{Characterization for Finite Block Size in the MPC-Assisted Model}
\label{sec:MPC-finite}
\subsection{Characterization for Strict Incentive Compatibility}
In this section, we give a characterization of strict incentive compatibility in the MPC-assisted model for finite block size. 
We show that cryptography helps 
us overcome the finite-block impossibility~\cite{foundation-tfm}
for $c = 1$, but 
for $c\geq 2$, the impossibility still holds.

\subsubsection{Feasibility for $c=1$}
In the MPC-assisted model, we indeed can have a mechanism that achieves UIC, MIC, and $(\rho,1)$-SCP against a coalition controlling 
$\rho \in (0, 1]$ fraction of the miners and $c = 1$ user.
\begin{mdframed}
    \begin{center}
    {\bf MPC-assisted, posted price auction with random selection}
    \end{center}
    \noindent {\bf Parameters}: the reserved price $r$, and a block size $k$.\\
    {\bf Input}: a bid vector $\bfb = (b_1,\dots, b_{N})$.\\
    {\bf Mechanism}:
    \begin{itemize}[leftmargin=5mm,itemsep=1pt]
    \item 
    {\it Allocation rule.}
    Any bid that is at least $r$ is considered as a candidate. Randomly select $k$ bids from the candidates to confirm.
    \item 
    {\it Payment rule.}
    Each confirmed bid pays $r$.
    \item 
    {\it Miner revenue rule.}
    Miner gets $0$ revenue.
    \end{itemize}
\end{mdframed}
In the above mechanism, the miner gains zero revenue.
This is inevitable as shown in~\cref{thm:inject-constantRevenue} of \Cref{sec:strict-ic-0-miner-rev}.
Even in the MPC-assisted model, the miner must have zero revenue if we insist on strict incentive compatibility (even under Bayesian notions of equilibrium).

\begin{theorem}
\label{thm:posted-price-burning}
Assuming a finite block size $k$.
The above MPC-assisted, posted price auction with random selection in the MPC-assisted model satisfies  
UIC, MIC, and $(\rho,1)$-SCP (in the ex post setting) 
for arbitrary $\rho\in(0,1]$.
\end{theorem}
\begin{proof}
We will prove the three incentive compatibility properties separately.

\paragraph{UIC.} 
Let $v_i$ denote the true value of user $i$.
First, refusing to bid cannot increase its utility.
Moreover, injecting bids does not help either.
To see this, assume that user $i$ bids its true value $v_i$ and injects a bid $b'$.
If $b'<r$, then it does not influence user $i$'s utility.
If $b'\geq r$, it either decreases the probability of user $i$ being confirmed if $v_i\geq r$, or it brings user $i$ negative expected utility if $v_i < r$.

Thus, we only need to argue that overbidding or underbidding does not increase the user's utility.
If user $i$'s true value $v_i<r$, then its utility when overbidding $b\geq r$ is $q\cdot (v_i - r)<0$, where $q$ is the probability of $b$ being confirmed.
If user $i$'s true value $v_i\geq r$, then underbidding $b<r$ brings it $0$-utility, whereas the honest utility $q(v_i - r)$ is positive. 
Therefore, no matter how user $i$ deviates from the protocol, its utility does not increase.

\paragraph{MIC.} Since the total miner revenue is always $0$, injecting fake bids does not increase the colluding miner's utility.
The miner cannot increase its utility by deviating from the protocol.

\paragraph{SCP.} No matter how the coalition deviates, the colluding miner's revenue is always $0$. Therefore, the joint utility of the coalition is at most the utility of the colluding user. 
By strict UIC, the joint utility does not increase.
\end{proof}

Note that the above mechanism does not work for $c=2$.
Imagine that the miner colludes with two users $i$ and $j$, where user $i$ has true value exactly $r$ and user $j$ has a sufficiently large true value.
User $i$ may choose not to bid to increase the probability of user $j$ being confirmed.
This brings the coalition strictly more utility than behaving honestly.

\subsubsection{Impossibility for $c \geq 2$}
Unfortunately, even in the MPC-assisted model, no mechanism 
with non-trivial utility 
can achieve UIC, MIC, and $(\rho,2)$-SCP, even for Bayesian
notions of incentive compatibility.
To see this, observe that under the strict incentive compatible notion, $(\rho,c)$-SCP implies that any coalition of $\leq c$ users cannot benefit 
from any deviation\footnote{We credit Bahrani, Garimidi, Roughgarden, Shi, and 
Weinberg for making this observation.}, 
since the miner revenue has to be $0$ 
by~\cref{thm:inject-constantRevenue} of \Cref{sec:strict-ic-0-miner-rev}.
Similar to the proof in Goldberg and Hartline~\cite{goldberghartline}, 
we show that any mechanism that is Bayesian UIC and Bayesian SCP against a $(\rho,2)$-sized coalition (for an arbitrary $\rho \in (0, 1]$ 
must satisfy the following condition: no matter how 
a user $j$ changes its bid, user $i$'s utility should not change. 
Formally,
\begin{lemma}
\label{lem:hartline}
Given any (possibly random) mechanism in the MPC-assisted model that is Bayesian UIC and Bayesian SCP against $(\rho,2)$-sized coalition for some $\rho\in(0,1]$, and 
suppose each user's true value is drawn i.i.d. from a distribution $\mathcal{D}$.
Then, for any user $i$ and user $j$, for any bid $b_j$ and $b'_j$, it must be that for any $\ell\geq 1$,
\[\underset{(v,\bids_{-i,j})\sim\mathcal{D}^{\ell}}{\E}[{\sf util}^i(v,b_j,\bids_{-i,j})] = \underset{(v,\bids_{-i,j})\sim\mathcal{D}^{\ell}}{\E}[{\sf util}^i(v,b'_j,\bids_{-i,j})],\]
where $\bids_{-i,j}$ represents all except user $i$ and user $j$'s bids.
\end{lemma}
The proof of this lemma is deferred to \cref{section:missing-proof}.
This lemma implies that no matter how user $j$ changes its bid, the expected utility of user $i$ should not change if user $i$'s true value is sampled randomly from $\mcal{D}$.
Consequently, we have the following result stating that user $i$'s utility should remain the same when bidding its true value, regardless of how many users are there.

\begin{lemma}
\label{lemma:sym-0-bid}
Given any (possibly randomized) mechanism in the MPC-assisted model that achieves Bayesian UIC and Bayesian SCP against $(\rho,2)$-sized coalition for some $\rho\in(0,1]$, it holds that 
for any user $i$ and $j$, for any bid $b_j$, for any $\ell\geq 1$, 
    \[\underset{(v_i,\bids_{-i,j})\sim\mathcal{D}^{\ell}}{\E}[{\sf util}^i(v_i,b_j,\bids_{-i,j})]\leq \underset{(v_i,\bids_{-i,j})\sim\mathcal{D}^{\ell}}{\E}[{\sf util}^i(v_i,\bids_{-i,j})],\]
where $v_{\textsf{id}}$ ($b_{\textsf{id}}$) denotes a bid $v$ ($b$) coming from identity $\textsf{id}$.
\end{lemma}
\paragraph{Proof roadmap for \Cref{lemma:sym-0-bid}.}
By \Cref{lem:hartline}, 
$\underset{v_i,\bids_{-i,j}\sim\mathcal{D}^{\ell}}{\E}[{\sf util}^i(v_i,b_j,\bids_{-i,j})] = \underset{v_i,\bids_{-i,j}\sim\mathcal{D}^{\ell}}{\E}[{\sf util}^i(v_i,0_j,\bids_{-i,j})]$.
Therefore, to prove \Cref{lemma:sym-0-bid}, it suffices
to prove that 
\[\underset{\bids_{-i,j}\sim\mathcal{D}^{\ell-1}}{\E}[{\sf util}^i(v_i,0_j,\bids_{-i,j})] \leq \underset{\bids_{-i,j}\sim\mathcal{D}^{\ell-1}}{\E}[{\sf util}^i(v_i,\bids_{-i,j})].
\]
This claim is relatively easy to prove if we 
are willing to assume a {\it strong} symmetry assumption explained below.
With a technically more involved proof, we 
can eventually get rid of this {\it strong} symmetry assumption 
and prove it under our 
current (much weaker) symmetry assumption defined in~\cref{sec:TFM-plain-model}.   

\begin{itemize}[leftmargin=3mm]
\item[]
{\it Strong symmetry assumption.}
On top of our 
current symmetric assumption defined in~\cref{sec:TFM-plain-model},
we additionally assume that for any 
bid vector ${\bf b} := (b_1, \ldots, b_N)$, 
if for $i \neq j$, $b_i = b_j$, then 
the random variables $(x_i, p_i)$
and $(x_j, p_j)$
are identically distributed, 
where $(x_i, p_i)$ are random variables denoting $i$'s confirmation probability
and $i$'s payment, respectively, and $(x_j, p_j)$ are similarly defined.
\end{itemize}

In other words, the strong symmetry assumption additionally assumes
that two bids of the same amount receive the same treatment, on top of
our existing symmetry assumption --- note that this is a very strong assumption,
and this is why we want to get rid of it eventually.
If the above strong symmetry assumption holds, then we have that 
for any identity $i'$ that injects a $0$ bid, 
\[\underset{\bids_{-i,j}\sim\mathcal{D}^{\ell-1}}{\E}[{\sf util}^i(v_i,0_j,\bids_{-i,j})] = \underset{\bids_{-i,j}\sim\mathcal{D}^{\ell-1}}{\E}[{\sf util}^i(v_i,0_{i'},\bids_{-i,j})],
\]
This is because under the strong symmetry assumption, anyone who bids
the same amount as $i$ has the same expected utility, 
and moreover, this utility
is not affected by whether the $0$ bid is posted by $j$ or $i'$.
Finally, we have for any $v_i$,
\[\underset{\bids_{-i,j}\sim\mathcal{D}^{\ell-1}}{\E}[{\sf util}^i(v_i,0_{i'},\bids_{-i,j})]\leq \underset{\bids_{-i,j}\sim\mathcal{D}^{\ell-1}}{\E}[{\sf util}^i(v_i,\bids_{-i,j})].\]
Otherwise, user $i$ can inject a $0$-bid using an arbitrary identity $i'$, 
which strictly increases its utility. This contradicts Bayesian UIC. 
We refer the reader to~\cref{sec:sym-0-bid} for a full proof of 
\Cref{lemma:sym-0-bid}
without relying on the strong symmetry assumption.

\begin{lemma}
\label{lem:user-i-util-unchange}
Given any (possibly randomized) mechanism in the MPC-assisted model that achieves Bayesian UIC, MIC and Bayesian SCP against $(\rho,2)$-sized coalitions for some $\rho\in(0,1]$, it holds that  
for any user $i$, any value $v_i$, for any $\ell\geq 1$,
\[\underset{v_i,\bids_{-i}\sim\mathcal{D}^{\ell}}{\E}[{\sf util}^i(v_i,\bids_{-i})] = \underset{v_i\sim\mathcal{D}}{\E}{\sf util}^i(v_i).\]
\end{lemma}
\begin{proof}
We first show that for any $j$, user $i$'s expected utility should not change if user $j$ refuses to bid. Formally, for any $b_j$,
\begin{equation}
    \underset{v_i,\bids_{-i,j}\sim\mathcal{D}^{\ell}}{\E}[{\sf util}^i(v_i,b_j,\bids_{-i,j})] = \underset{v_i,\bids_{-i,j}\sim\mathcal{D}^{\ell}}{\E}[{\sf util}^i(v_i,\bids_{-i,j})].
    \label{eqn:util-equiv-if-quit}
\end{equation}

To see this, by~\cref{lemma:sym-0-bid}, we have
\begin{align*}
    \underset{v_i,\bids_{-i,j}\sim\mathcal{D}^{\ell}}{\E}[{\sf util}^i(v_i,b_j,\bids_{-i,j})] \leq
    \underset{v_i,\bids_{-i,j}\sim\mathcal{D}^{\ell}}{\E}[{\sf util}^i(v_i,\bids_{-i,j})].
\end{align*}
Next, we are going to show that 
\[\underset{v_i,\bids_{-i,j}\sim\mathcal{D}^{\ell}}{\E}[{\sf util}^i(v_i,b_j,\bids_{-i,j})] = \underset{v_i,\bids_{-i,j}\sim\mathcal{D}^{\ell}}{\E}[{\sf util}^i(v_i,0_j,\bids_{-i,j})]\geq \underset{v_i,\bids_{-i,j}\sim\mathcal{D}^{\ell}}{\E}[{\sf util}^i(v_i,\bids_{-i,j})].\]
To see why this holds, note that the first equality follows from \cref{lem:hartline}.
The inequality comes from 2-SCP:
Since by MIC, it must be that $\underset{\bids_{-i,j}\sim\mathcal{D}^{\ell}}{\E}[\mu(v_i,0_j,\bids_{-i,j})]\leq \underset{\bids_{-i,j}\sim\mathcal{D}^{\ell}}{\E}[\mu(v_i,\bids_{-i,j})]$,
therefore, it must be that $ \underset{v_i,\bids_{-i,j}\sim\mathcal{D}^{\ell}}{\E}[{\sf util}^i(v_i,0_j,\bids_{-i,j})]\geq \underset{v_i,\bids_{-i,j}\sim\mathcal{D}^{\ell}}{\E}[{\sf util}^i(v_i,\bids_{-i,j})]$.
Otherwise, if there exists a $v_i$ such that this does not hold, the miner can collude with user $i$ and user $j$ with true value $0$, and ask user $j$ not to bid. 
This strategy strictly increases the coalition's joint utility and thus contradicts Bayesian SCP against $(\rho,2)$-sized coalition.
Equation~\eqref{eqn:util-equiv-if-quit} thus follows.

Let $f(\cdot)$ denote the p.d.f. of $\mathcal{D}$.
By definition of expectation,
\begin{align*}
    &\underset{v_i,\bids_{-i}\sim\mathcal{D}^{\ell}}{\E}[{\sf util}^i(v_i,\bids_{-i})] = \int_{0}^{\infty} \underset{v_i,\bids_{-i,1}\sim\mathcal{D}^{\ell-1}}{\E}[{\sf util}^i(v_i,z_1,\bids_{-i,1})]f(z_1) dz_1\\
    =& \int_{0}^{\infty} \underset{v_i,\bids_{-i,1}\sim\mathcal{D}^{\ell-1}}{\E}[{\sf util}^i(v_i,\bids_{-i,1})]f(z_1) dz_1 &\text{By Equation~\eqref{eqn:util-equiv-if-quit}}\\
    =& \underset{v_i,\bids_{-i,1}\sim\mathcal{D}^{\ell-1}}{\E}[{\sf util}^i(v_i,\bids_{-i,1})].
\end{align*}
The lemma follows by repeating the above argument.
\end{proof}

Now we are ready to prove the theorem stating that there is no mechanism that gives non-zero utility to either users or miners 
and yet satisfies Bayesian UIC and SCP against $(\rho,2)$-sized coalitions.

\begin{theorem}
Suppose the block size is $k$.
No MPC-assisted mechanism with non-trivial utility
simultaneously 
achieves Bayesian UIC, MIC and Bayesian SCP against $(\rho,2)$-sized coalitions.
\end{theorem}
\begin{proof}
By~\cref{thm:inject-constantRevenue}, the miner-revenue has to be $0$.
Therefore, it suffices to prove that every user must have $0$-utility.

Consider a crowded world with $K$ number of users 
and all of their bids are sampled independently at 
random from $\mcal{D}$. 
There must exist a user $j^*$ whose
probability of being confirmed is at most $k/K$, and thus its expected utility 
is at most $\max(\mcal{D}) \cdot k/K$ where $k$ is the block size.
Thus, $\underset{\bids\sim\mathcal{D}^{K}}{\E}[{\sf util}^{j^*}(\bids)] = 0$ by taking $K$ to be arbitrarily large.

By \Cref{lem:user-i-util-unchange}, it must be that $\underset{v_{j^*}\sim\mcal{D}}{\E}{\sf util}^i(v_{j^*}) = 0.$
Since $\underset{v_i\sim\mcal{D}}{\E}{\sf util}^i(v_i) = \underset{v_j\sim\mcal{D}}{\E}{\sf util}^j(v_j)$ for all user $i$ and $j$, every user's expected
utility must be the $0$ where the expectation is taken over the randomness of bids as well as the randomness of the mechanism, i.e.,
for any user $i$, any $\ell\geq 1$, we have $\underset{\bids\sim\mathcal{D}^{\ell}}{\E}[{\sf util}^{i}(\bids)] = 0$.
Since user's utility is non-negative, this implies that
$\underset{\bids_{-i}\sim\mathcal{D}^{\ell-1}}{\E}[{\sf util}^{i}(v,\bids_{-i})] = 0$.

\end{proof}


\subsection{Feasibility of Approximate Incentive Compatibility}
Although strict (even Bayesian) 
incentive compatibility is impossible to achieve for $c \geq 2$ in the 
MPC-assisted model, we have meaningful feasibility results 
if we allow $\eps$ additive slack.
Still, we use $k$ to denote the finite block size and $M$ to denote the upper bound of the true values.
Specifically, we can achieve $\Theta(kM)$ social welfare as long
as many people place high enough bids, 
which is asymptotically the best possible social welfare one can hope for.
\begin{mdframed}
    \begin{center}
    {\bf MPC-assisted, Diluted Posted Price Auction}
    \end{center}
    \paragraph{Parameters:} the block size $k$, an upper bound $c$ of the number of users
colluding with the miner, an upper bound $M$ of users' true values, 
a slack $\eps \geq 0$, and a posted-price $r$ such that $r\geq \frac{\eps}{2c}$. 
    \vspace{-1em}
    \paragraph{Input:} a bid vector $\bfb = (b_1,\dots,b_{N})$.
    
    \vspace{-1em}
    \paragraph{Mechanism:}
    \begin{enumerate}[leftmargin=5mm]
    \item 
    {\it Allocation rule.}
    \begin{itemize}[leftmargin=5mm]
        \item 
        Given a bid vector $\bfb = (b_1,\dots,b_{N})$, remove all bids which are smaller than $r$.
        Let $\widetilde{\bids} = (\widetilde{b}_1,\dots, \widetilde{b}_{\ell})$ denote the resulting vector. 
        \item
        Let $T = \max \left(2c \sqrt{\frac{kM}{\eps}}, k\right)$.
        If $\ell \geq T$, let $\bfd = \widetilde{\bids}$.
        Else, let $\bfd = (\widetilde{b}_1,\dots, \widetilde{b}_{\ell}, 0, \ldots, 0)$ such that $|\bfd| = T$.
		In other words, $\bfd$ is $\widetilde{\bids}$ appended with $T - \ell$ zeros.
        \item 
        Randomly choose a set $S$ of size $k$ from $\bfd$, and every non-zero bid in $S$ is confirmed.
    \end{itemize}
    \item 
    {\it Payment rule.}
    For each confirmed bid $b$, it pays $r$.
    \item 
{\it Miner revenue rule.}
    For each confirmed bid $b$, the miner is paid $\frac{\eps}{2c}$.
    \end{enumerate}
\end{mdframed}

\begin{theorem}
Suppose there exists an upper bound $M$ on users' true values.
The above MPC-assisted, diluted posted price auction satisfies UIC, MIC, and $\eps$-SCP (in the ex post setting) 
against $(\rho,c)$-sized coalitions for arbitrary $\rho\in(0,1]$ and $c \geq 1$.
\end{theorem}
\begin{proof}
We will prove the three incentive compatibility properties separately.
Note that in this mechanism, refusing to bid is equivalent to underbidding some value less than $r$.
So we mainly focus on the strategy space of bidding untruthfully and injecting bids.
When we say the expected utility of a user, the randomness is taken over the randomness in the mechanism.
\paragraph{UIC.} 
Fix any user $i$, and let $v$ denote the true value of user $i$.
In the  mechanism, any confirmed bid pays $r$ and any bid less than $r$ must be unconfirmed.
Thus, if $v \leq r$, bidding untruthfully cannot give a postive utility, 
so bidding truthfully and getting 0-utility is optimal.

Below we focus on the case when $v > r$.
In this case, the bid has a non-negative probability of being confirmed and it pays $r$.
So following the honest strategy leads to positive utility.
Bidding less than $r$ will cause the bid  to be unconfirmed and will not help the user.
Therefore, we may assume that the user
bids at least $r$ and may inject some fake bids.
Observe that any bid that is at least $r$ is treated the same by the mechanism.  
Moreover, injecting fake bids 
\ignore{
Bidding any number $b \geq r$ does not change the confirmation probability and the expected payment compared to bidding $v$, and bidding any number $b < r$ will cause it to be unconfirmed.
Thus, bidding truthfully is still optimal.
Consequently, we conclude that bidding truthfully is always optimal no matter what the world is (even if the world consists of some fake bids).
Next, if user $i$ injects some fake bids, either the injecting bids 
}
either 
make no difference (when $\ell \leq T$ after injecting), or it reduces the probability of bid $v$ being elected into the set $S$ (when $\ell > T$ after injecting).
Therefore, bidding untruthfully
and/or injecting fake bids 
does not help the user.




\paragraph{MIC.} 
By injecting fake bids, the strategic miner cannot increase the expected number of 
real bids in the vector ${\bf d}$. Thus, injecting fake bids 
cannot increase other bids' contribution towards the miner's revenue.
Therefore, the expected gain in miner revenue 
must be upper bounded by 
the fake bids' contribution towards miner revenue  
minus the expected payments of the fake bids.
For each confirmed bid, the miner revenue is fixed to $\frac{\eps}{2c}$, 
which is no more than 
the payment of the bid. 
Thus, the expected miner revenue
cannot increase through injecting fake bids, i.e., the mechanism is MIC.

\paragraph{$\eps$-SCP.} 
First, we argue that 
injecting bids does not help the coalition.
Specifically, using a similar proof as UIC, injecting bids does not help improve
the utility of any user in the coalition. 
Using a similar argument as MIC, injecting bids
does not improve the miner's revenue minus the payment of the injected bids. 
Therefore, injecting bids will not increase the coalition's joint utility.

Now it suffices to argue that underbidding or overbidding does not increase the coalition's joint utility by more than $\eps$.
Suppose when bidding honestly, the number of bids in $\widetilde{\bids}$ is $\ell$.
Each bid in $\widetilde{\bids}$ is confirmed with probability $\frac{k}{\max\{T,\ell\}}$.
Assume that by bidding untruthfully, the coalition changes the length of $\widetilde{\bids}$ to $\ell'$.
Now each bid in $\widetilde{\bids}$ is confirmed with probability $\frac{k}{\max\{T,\ell'\}}$.

We partition the players in the coalition into the following groups:
\begin{itemize}[leftmargin=5mm]
    \item Those whose true values are less than $r$ and bid less than $r$. Their expected utility does not change.
    \item Those whose true values are less than $r$ and bid higher than or equal to $r$. Their expected utility does not increase.
    \item Those whose true values are at least $r$ and bid less than $r$. Their expected utility does not increase.
    \item Those whose true values are at least $r$ and bid at least $r$. 
    For each of these users, its expected utility increases by at most
    \begin{equation}
        (v-r)\frac{k}{\max\{T,\ell'\}} - (v-r)\frac{k}{\max\{T,\ell\}}.
        \label{eqn:dilut-util-change}
    \end{equation}
    Note that for $\ell'\geq\ell$, then $\eqref{eqn:dilut-util-change}\leq 0$.
    Therefore, we only need to consider the case where $\ell'<\ell$. 
    If $\ell \leq T$, then~\eqref{eqn:dilut-util-change} is $0$.
    If $\ell > T$, then~\eqref{eqn:dilut-util-change} is upper bounded by
    \begin{align*}
        \eqref{eqn:dilut-util-change}&\leq (v-r)\left[\frac{k}{\ell'} - \frac{k}{\ell}\right]\\
        &\leq (v-r)\left[\frac{k}{\ell-c} - \frac{k}{\ell}\right]\leq (v-r)\frac{ck}{\ell(\ell-c)}\\
        &\leq M\cdot \frac{ck}{T(T-c)}.
    \end{align*}
    By the choice of $T$, we have that $T(T-c)\geq \frac{1}{2}T^2$.
    Thus, 
    \[\eqref{eqn:dilut-util-change}\leq M\cdot \frac{ck}{T(T-c)}\leq \frac{2Mck}{T^2}\leq \frac{\epsilon}{2c}.\]
    This implies that each user's utility can increase by at most $\frac{\eps}{2c}$. 
    Meanwhile, for each user in the coalition, it can increase the miner's revenue by no more than $\frac{\eps}{2c}$ via bidding untruthfully. 
    Since there are at most $c$ users in the coalition, the coalition can gain at most $\eps$ more utility in total, no matter how they deviate.
\end{itemize}
\end{proof}


\section*{Acknowledgments}
This work is in part supported by NSF awards 2212746, 2044679, 1704788, a Packard Fellowship, a generous gift from the late Nikolai Mushegian, a gift from Google, and an ACE center grant from Algorand Foundation. 
The authors would like to thank the anonymous reviewers for their helpful comments.
We also thank Matt Weinberg
for helpful technical discussions regarding how to 
efficiently instantiate our MPC-assisted mechanisms.

\bibliographystyle{alpha}
\bibliography{tfm-refs,refs,crypto,gametheory}

\newcommand{\etalchar}[1]{$^{#1}$}
\begin{thebibliography}{KMSW22}

\bibitem[ACH11]{giladgtcrypto}
Gilad Asharov, Ran Canetti, and Carmit Hazay.
\newblock Towards a game theoretic view of secure computation.
\newblock In {\em Eurocrypt}, 2011.

\bibitem[ADGH06]{gtcrypto02}
Ittai Abraham, Danny Dolev, Rica Gonen, and Joseph Halpern.
\newblock Distributed computing meets game theory: Robust mechanisms for
  rational secret sharing and multiparty computation.
\newblock In {\em PODC}, 2006.

\bibitem[AL11]{giladutilityindjournal}
Gilad Asharov and Yehuda Lindell.
\newblock Utility dependence in correct and fair rational secret sharing.
\newblock {\em Journal of Cryptology}, 24(1), 2011.

\bibitem[BCD{\etalchar{+}}]{eip1559}
Vitalik Buterin, Eric Conner, Rick Dudley, Matthew Slipper, and Ian Norden.
\newblock Ethereum improvement proposal 1559: Fee market change for eth 1.0
  chain.
\newblock \url{https://github.com/ethereum/EIPs/blob/master/EIPS/eip-1559.md}.

\bibitem[BEOS19]{functional-fee-market}
Soumya Basu, David~A. Easley, Maureen O'Hara, and Emin~G{\"{u}}n Sirer.
\newblock Towards a functional fee market for cryptocurrencies.
\newblock {\em CoRR}, abs/1901.06830, 2019.

\bibitem[Can00]{Canetti2000}
Ran Canetti.
\newblock Security and composition of multiparty cryptographic protocols.
\newblock {\em Journal of Cryptology}, 2000.

\bibitem[CCWS21]{gt-leader-shi}
Kai-Min Chung, T-H.~Hubert Chan, Ting Wen, and Elaine Shi.
\newblock Game-theoretic fairness meets multi-party protocols: The case of
  leader election.
\newblock In {\em CRYPTO}. Springer-Verlag, 2021.

\bibitem[CGL{\etalchar{+}}18]{gt-fair-cointoss}
Kai{-}Min Chung, Yue Guo, Wei{-}Kai Lin, Rafael Pass, and Elaine Shi.
\newblock Game theoretic notions of fairness in multi-party coin toss.
\newblock In {\em TCC}, volume 11239, pages 563--596, 2018.

\bibitem[CS21]{foundation-tfm}
Hao Chung and Elaine Shi.
\newblock Foundations of transaction fee mechanism design.
\newblock {\em arXiv preprint arXiv:2111.03151}, 2021.

\bibitem[DR07]{gtcrypto06}
Yevgeniy Dodis and Tal Rabin.
\newblock Cryptography and game theory.
\newblock In {\em AGT}, 2007.

\bibitem[EFW22]{credibleauction-comm01}
Meryem Essaidi, Matheus V.~X. Ferreira, and S.~Matthew Weinberg.
\newblock Credible, strategyproof, optimal, and bounded expected-round
  single-item auctions for all distributions.
\newblock In Mark Braverman, editor, {\em 13th Innovations in Theoretical
  Computer Science Conference, {ITCS} 2022, January 31 - February 3, 2022,
  Berkeley, CA, {USA}}, volume 215 of {\em LIPIcs}, pages 66:1--66:19, 2022.

\bibitem[FMPS21]{dynamicpostedprice}
Matheus V.~X. Ferreira, Daniel~J. Moroz, David~C. Parkes, and Mitchell Stern.
\newblock Dynamic posted-price mechanisms for the blockchain transaction-fee
  market.
\newblock {\em CoRR}, abs/2103.14144, 2021.

\bibitem[FW20]{credibleauction-comm00}
Matheus V.~X. Ferreira and S.~Matthew Weinberg.
\newblock Credible, truthful, and two-round (optimal) auctions via
  cryptographic commitments.
\newblock In P{\'{e}}ter Bir{\'{o}}, Jason~D. Hartline, Michael Ostrovsky, and
  Ariel~D. Procaccia, editors, {\em {EC} '20: The 21st {ACM} Conference on
  Economics and Computation, Virtual Event, Hungary, July 13-17, 2020}, pages
  683--712. {ACM}, 2020.

\bibitem[GH05]{goldberghartline}
Andrew~V. Goldberg and Jason~D. Hartline.
\newblock Collusion-resistant mechanisms for single-parameter agents.
\newblock In {\em {SODA} 2005}, pages 620--629, 2005.

\bibitem[GKM{\etalchar{+}}13]{rdp00}
Juan~A. Garay, Jonathan Katz, Ueli Maurer, Bj{\"{o}}rn Tackmann, and Vassilis
  Zikas.
\newblock Rational protocol design: Cryptography against incentive-driven
  adversaries.
\newblock In {\em FOCS}, 2013.

\bibitem[GKTZ15]{rdp01}
Juan Garay, Jonathan Katz, Bj\"{o}rn Tackmann, and Vassilis Zikas.
\newblock How fair is your protocol? a utility-based approach to protocol
  optimality.
\newblock In {\em PODC}, 2015.

\bibitem[GLR10]{seqrationalcrypto}
Ronen Gradwohl, Noam Livne, and Alon Rosen.
\newblock Sequential rationality in cryptographic protocols.
\newblock In {\em FOCS}, 2010.

\bibitem[GMW87]{gmw87}
O.~Goldreich, S.~Micali, and A.~Wigderson.
\newblock How to play any mental game.
\newblock In {\em ACM symposium on Theory of computing (STOC)}, 1987.

\bibitem[GO14]{groth2014cryptography}
Jens Groth and Rafail Ostrovsky.
\newblock Cryptography in the multi-string model.
\newblock {\em Journal of cryptology}, 27(3):506--543, 2014.

\bibitem[GPS19]{guo2019synchronous}
Yue Guo, Rafael Pass, and Elaine Shi.
\newblock Synchronous, with a chance of partition tolerance.
\newblock In {\em Annual International Cryptology Conference}, pages 499--529.
  Springer, 2019.

\bibitem[GTZ15]{rdp02}
Juan~A. Garay, Bj{\"{o}}rn Tackmann, and Vassilis Zikas.
\newblock Fair distributed computation of reactive functions.
\newblock In {\em DISC}, volume 9363, pages 497--512, 2015.

\bibitem[Har]{myerson-lecture-hartline}
Jason Hartline.
\newblock Lectures on optimal mechanism design.
\newblock \url{http://users.eecs.northwestern.edu/~hartline/omd.pdf}.

\bibitem[HT04]{gtcrypto00}
Joseph Halpern and Vanessa Teague.
\newblock Rational secret sharing and multiparty computation.
\newblock In {\em STOC}, 2004.

\bibitem[IML05]{gtcrypto05}
Sergei Izmalkov, Silvio Micali, and Matt Lepinski.
\newblock Rational secure computation and ideal mechanism design.
\newblock In {\em FOCS}, 2005.

\bibitem[Kat08]{katzgametheory}
Jonathan Katz.
\newblock Bridging game theory and cryptography: Recent results and future
  directions.
\newblock In {\em TCC}, 2008.

\bibitem[KMSW22]{logstar-gt-leader}
Ilan Komargodski, Shin’ichiro Matsuo, Elaine Shi, and Ke~Wu.
\newblock log*-round game-theoretically-fair leader election.
\newblock In {\em CRYPTO}, 2022.

\bibitem[KN08]{gtcrypto01}
Gillat Kol and Moni Naor.
\newblock Cryptography and game theory: Designing protocols for exchanging
  information.
\newblock In {\em TCC}, 2008.

\bibitem[LSZ19]{zoharfeemech}
Ron Lavi, Or~Sattath, and Aviv Zohar.
\newblock Redesigning bitcoin's fee market.
\newblock In {\em The World Wide Web Conference, {WWW} 2019}, pages 2950--2956,
  2019.

\bibitem[Mye81]{myerson}
Roger~B. Myerson.
\newblock Optimal auction design.
\newblock {\em Math. Oper. Res.}, 6(1), 1981.

\bibitem[NRTV07]{agt}
Noam Nisan, Tim Roughgarden, Eva Tardos, and Vijay~V. Vazirani.
\newblock {\em Algorithmic Game Theory}.
\newblock Cambridge University Press, USA, 2007.

\bibitem[OPRV09]{gtcrypto03}
Shien~Jin Ong, David~C. Parkes, Alon Rosen, and Salil~P. Vadhan.
\newblock Fairness with an honest minority and a rational majority.
\newblock In {\em {TCC}}, 2009.

\bibitem[PS17]{fruitchain}
Rafael Pass and Elaine Shi.
\newblock Fruitchains: A fair blockchain.
\newblock In {\em PODC}, 2017.

\bibitem[RBO89]{rabin-benor}
T.~Rabin and M.~Ben-Or.
\newblock Verifiable secret sharing and multiparty protocols with honest
  majority.
\newblock In {\em STOC}, 1989.

\bibitem[Rou20]{roughgardeneip1559}
Tim Roughgarden.
\newblock Transaction fee mechanism design for the {Ethereum} blockchain: An
  economic analysis of {EIP}-1559.
\newblock Manuscript, \url{https://timroughgarden.org/papers/eip1559.pdf},
  2020.

\bibitem[Rou21]{roughgardeneip1559-ec}
Tim Roughgarden.
\newblock Transaction fee mechanism design.
\newblock In {\em EC}, 2021.

\bibitem[WAS22]{gt-fair-coin-complete}
Ke~Wu, Gilad Asharov, and Elaine Shi.
\newblock A complete characterization of game-theoretically fair, multi-party
  coin toss.
\newblock In {\em Eurocrypt}, 2022.

\bibitem[Yao]{yaofeemech}
Andrew Chi-Chih Yao.
\newblock {An Incentive Analysis of Some Bitcoin Fee Designs (Invited Talk)}.
\newblock In {\em ICALP 2020}.

\end{thebibliography}
\clearpage

\appendix
\section{Full Proof of~\cref{thm:proportional}}
\label{section:proportional-MPC}


We now prove~\cref{thm:proportional} 
of~\cref{section:proportional}, i.e., the propotional auction
in the plain model 
satisfies  
UIC, MIC, and $\frac54 c\eps$-SCP against 
any miner-user coalition 
with an arbitrary $c \geq $ number of users.

\paragraph{Proof of~\cref{thm:proportional}.}
We prove the three properties individually.

\paragraph{UIC.}
Because the confirmation and the payment of each bid are independent of other bids, injecting fake bids does not help to increase any user's utility.
Next, suppose user $i$'s true value is $v_i$.
If user $i$ bids $b_i$, its expected utility is \[
\begin{cases}
    \left(v_i - \frac{b_i}{2}\right)\frac{b_i}{r}, & \text{ if } b_i < r,\\
    v_i - \frac{r}{2}, & \text{ if } b_i \geq r.
\end{cases}
\]
By direct calculation, the expected utility is maximized when $b_i = v_i$.
Thus, proportional auction is strict UIC.

\paragraph{MIC.}
Since the block size is infinite, the miner's best strategy is to include all bids to maximize its revenue.
Notice that the confirmation of each bid and the miner revenue of each bid are independent of other bids.
Thus, injecting fake bids does not change the miner revenue from ``other bids.''
Moreover, for each confirmed bid, the miner revenue is upper bounded by the payment of that bid.
Thus, the increment of the miner revenue never exceeds the cost of the injected fake bids.
Thus, the miner revenue cannot increase by injecting fake bids, so the mechanism is strict MIC.

\paragraph{$\frac{5}{4}c\eps$-SCP.}
As we have shown in the argument for strict UIC and strict MIC, injecting fake bids does not change the colluding miner's revenue.
Because the confirmation and the payment of each bid are independent of other bids, injecting fake bids does not help to increase any user's utility.
Thus, in the rest of the proof, we assume the only deviation of the coalition is to change the bids from colluding users' true values to other values.
Let user $i$ be a colluding user.
We will show that the joint utility increases at most by $\frac{5}{4}\eps$ if user $i$ changes its bid from its true value to other values, no matter what other bids are.
Because there are at most $c$ colluding users, the mechanism is $\frac{5}{4}c\eps$-SCP for all $c$.

Let user $i$ be a colluding user with true value $v_i$, and let $b_i$ be user $i$'s bid.
We now proceed to analyze the utility of coalition based on how users in the coalition bid untruthfully.
\paragraph{1. Underbidding.} Suppose $b_i < v_i$.
Notice that the miner can get the payment from $b_i$ only when $b_i$ is confirmed, and the miner is paid $ \frac{\sqrt{2r\eps}}{2}$ if $b_i\geq \sqrt{2r\epsilon}$.
When user $i$ underbids, the miner's revenue can not increase.
Because the mechanism is strict UIC, underbidding does not increase user $i$'s utility either.
Thus, the joint utility does not increase if $b_i < v_i$.

\paragraph{2. Overbidding.} Suppose $b_i > v_i$.
We first consider the following cases based on whether the true value $v_i$ is less than $r$.

\begin{itemize}
\item 
{\bf If $v_i \geq \sqrt{2r\epsilon}$.} 
If $v_i \geq r$, bidding truthfully already guarantees user $i$'s bid to be confirmed, and the miner is paid $\sqrt{\frac{r\eps}{2}}$.
Thus, when $v_i \geq r$, overbidding does not increase the joint utility.
In the following, we assume $v_i < r$.
Let $\Delta = \min(b_i - v_i, r - v_i) > 0$.
If user $i$ bids truthfully, its bid is confirmed with the probability $\frac{v_i}{r}$, so its expected utility is \[
\left(v_i - \frac{v_i}{2}\right)\frac{v_i}{r}.
\]
Next, suppose user $i$ bids $b_i > v_i$.
Then, $b_i$ is confirmed with the probability $\frac{v_i + \Delta}{r}$, and the payment is $\frac{v_i + \Delta}{2}$ if $b_i$ is confirmed.
Thus, user $i$'s expected utility is \[
\left(v_i - \frac{v_i + \Delta}{2}\right) \frac{v_i + \Delta}{r}.
\]
Hence, compared to bidding truthfully, user $i$'s expected utility decreases by \[
    \left(v_i - \frac{v_i}{2}\right)\frac{v_i}{r} -  \left(v_i - \frac{v_i + \Delta}{2}\right) \frac{v_i + \Delta}{r} = \frac{\Delta^2}{2r} > 0.
\]
On the other hand, 
if user $i$ bids truthfully, the miner's expected revenue is $\frac{v_i}{r}\sqrt{\frac{r\eps}{2}}$.
If user $i$ bids $b_i > v_i$, the miner's expected revenue is $\frac{v_i + \Delta}{r}\sqrt{\frac{r\eps}{2}}$.
Thus, compared to bidding truthfully, the miner's expected utility increases by \[
    \frac{v_i + \Delta}{r}\sqrt{\frac{r\eps}{2}} - \frac{v_i}{r}\sqrt{\frac{r\eps}{2}} = \frac{\Delta}{r}\sqrt{\frac{r\eps}{2}}.
\]
Combine the argument above, the joint utility increases by 
\begin{equation}\label{eq:proportional}
    \frac{\Delta}{r}\sqrt{\frac{r\eps}{2}} - \frac{\Delta^2}{2r}.
\end{equation}
The maximum of Eq.(\ref{eq:proportional}) is $\frac{\eps}{4}$, so overbidding $b_i$ can only increase the joint utility by $\frac{\eps}{4}$.

\item 
{\bf If $v_i < \sqrt{2r\epsilon}$.}
Because the mechanism is strict-UIC, overbidding does not increase user $i$'s utility.
If $b_i < \sqrt{2r\eps}$, the miner revenue is still zero.
Thus, we assume $b_i \geq \sqrt{2r\eps}$.
From the argument in the previous case, we know that compared to bidding truthfully, user $i$'s expected utility decreases by $\frac{\Delta^2}{2r}$.
However, if user $i$ bids truthfully, the miner's revenue is zero.
If user $i$ bids $b_i > v_i$, the miner's expected revenue is $\frac{v_i + \Delta}{r}\sqrt{\frac{r\eps}{2}}$.
Thus, compared to bidding truthfully, the miner's expected revenue increases by $\frac{v_i + \Delta}{r}\sqrt{\frac{r\eps}{2}}$.
Consequently, the joint utility increases by
\begin{equation}\label{eq:proportional2}
    \frac{v_i}{r}\sqrt{\frac{r\eps}{2}} + \frac{\Delta}{r}\sqrt{\frac{r\eps}{2}} - \frac{\Delta^2}{2r}.
\end{equation}
Because the maximum of Eq.(\ref{eq:proportional}) is $\frac{\eps}{4}$, the maximum of Eq.(\ref{eq:proportional2}) when $v_i < \sqrt{2r\eps}$ is at most $\frac{5\eps}{4}$.
Thus, overbidding $b_i$ can only increase the joint utility by $\frac{5\eps}{4}$.
\end{itemize}
To sum up, among all cases, overbidding $b_i$ can only increase the joint utility by at most $\frac{5}{4}\eps$.
The theorem thus follows.

\paragraph{Proportional auction for the MPC-assisted model.}
In the MPC-assisted model, we want to ensure
incentive compatibility for any miner-user
coalition controlling at most $\rho$ fraction
of the miners and at most $c$ users --- recall
that the total miner revenue is split among the miners.
By contrast, in the plain model, effectively $\rho$ is always equal to 
$1$ since we always focus on the miner of the present block.
Therefore, to make the proportional auction work
in the MPC-assisted 
model, we make a small modification 
to the scheme and proof.
For the scheme, the only modification is that 
we now allow the miner revenue to scale up 
w.r.t. $\frac{1}{\rho}$  (up to the total user payment), 
such that the miner revenue
can be larger if we only want to be resilient
against coalitions controlling small fraction of the miners --- see the formal
description below.

\begin{mdframed}
    \begin{center}
    {\bf Proportional Auction for the MPC-assisted Model}
    \end{center}
    \paragraph{Parameters:} the approximate factor $\eps$, upper bound $\rho$ on the fraction of the colluding miners and the reserved price $r$ such that $r \geq 2\eps$.

    \paragraph{Input:} a bid vector $\bfb = (b_1,\dots, b_N)$.
    
    \paragraph{Mechanism:}
    \begin{itemize}[leftmargin=5mm,itemsep=1pt]
    \item 
    {\it Allocation rule.}
    For each bid $b$, if $b < r$, it is confirmed with the probability $b/r$; otherwise, if $b \geq r$, it is confirmed with probability $1$.
    \item 
    {\it Payment rule.}
    For each confirmed bid $b$, if $b < r$, it pays $b/2$; otherwise, if $b \geq r$, it pays $r/2$.
    \item 
    {\it Miner revenue rule.}
    For each confirmed bid $b$, let $p$ be the payment of $b$, and the miner is paid $\min\left(p, \frac{\sqrt{2r\eps}}{2\rho}\right)$.\footnote{The minimum guarantees that the miner revenue never exceed the payment.}
    \end{itemize}
\end{mdframed}

It is not hard to see that 
our proof of~\cref{thm:proportional} can be easily modified to work
for the MPC-assisted model. 
The only difference is that when the colluding user's true value $v_i$ is smaller than the threshold $\frac{\sqrt{2r\epsilon}}{\rho}$, overbidding to $v_i+\Delta < \frac{\sqrt{2r\epsilon}}{\rho}$ also increases the joint utility. 
There are two cases:
\begin{itemize}
    \item When a user with true value $v_i$ overbids to $v_i +\Delta < \frac{\sqrt{2r\epsilon}}{\rho}$, the coalition of the miner and this colluding user can gain at most $\eps$ more utility if $\Delta = \sqrt{2r\epsilon}$.
    \item When a user with true value $v_i$ overbids to $v_i +\Delta \geq  \frac{\sqrt{2r\epsilon}}{\rho}$, the coalition of the miner and this colluding user can gain at most $\frac{9}{4}\eps$ more utility when $\Delta = \sqrt{2r\epsilon}$ and $v_i$ is arbitrarily close to  $\frac{\sqrt{2r\epsilon}}{\rho}$.
\end{itemize}

\elaine{explain why}

\section{Feasibility: 
Approximate Incentive Compatibility for Finite Blocks}
\label{section:staircase}
In this section, we give a mechanism, called staircase mechanism, 
that is $\eps$-UIC, MIC, and $\eps$-SCP for $c = 1$ in the plain model.
The staircase mechanism can in the best case achieve $\Theta(k^2 \epsilon)$ social 
welfare. 
Recall that in \cref{thm:eps-impossible-approx-plain-finite}, we showed that any 
plain-model mechanism that works for finite block size 
suffers from poor scaling of the social welfare w.r.t.~the bid distribution.
In particular, we showed that the 
social welfare is upper bounded by 
$O(k^3 \epsilon \log(M/\epsilon))$
where $M$ is an upper bound on the social welfare.
Our staircase mechanism can achieve $\Theta(k^2 \epsilon)$
social welfare in the best case. Thus, we still have a gap
between the upper and lower bounds. Bridging this gap
is an interesting open problem.

\begin{mdframed}
    \begin{center}
    {\bf Staircase Mechanism}
    \end{center}
    \paragraph{Parameters:} the block size $k$, the upper bound $c$ of the colluding users, the upper bound $M$ of the true value, the approximate factor $\epsilon$.
    \vspace{-1em}
    \paragraph{Notations:} We define \[
        F_0 = 
        \begin{cases}
            M - k\eps, &\text{ if } \lfloor\frac{M}{\eps} \rfloor \geq k,\\
            M - \lfloor\frac{M}{\eps} \rfloor \eps, &\text{ if otherwise}. 
        \end{cases}
    \]
    For all $i = 1,\dots, k$, we define \[
    F_i = F_0 + i \cdot \eps.
    \]
    \vspace{-1em}
    \paragraph{Input:} a bid vector $\bfb = (b_1,\dots,b_N)$.
    \vspace{-1em}
    \paragraph{Mechanism:}
    \begin{enumerate}[leftmargin=5mm]
    \item 
    {\it Inclusion rule.}
    Given the bid vector $\bfb = (b_1,\dots,b_N)$, choose the top $k$ bids.
    \item 
    {\it Confirmation rule.}
    \hao{How do we break tie?}
    \begin{itemize}[leftmargin=5mm]
        \item 
        Let $\bfc = (c_1,\dots, c_{N'})$ denote the bid vector in the block, where $c_1 \geq c_2 \geq \cdots \geq c_{N'}$ and $N' \leq k$.
        \item 
        If $c_1 < F_1$, set $t = 0$.
        Otherwise, set $t = \max_i \{i: c_i \geq F_i\}$.
        \item 
        If $t = 0$, no one is confirmed.
        Otherwise, $c_1, \ldots, c_t$ are confirmed.
    \end{itemize}
    \item 
    {\it Payment rule.}
    For each confirmed bid, it pays $F_t$.
    \item 
    {\it Miner revenue rule.}
    Miner is paid $t \cdot \eps$.
    \end{enumerate}
\end{mdframed}

In the staircase mechanism, the more bids confirmed, the higher the price. 
For example, let $M = 10$ be the maximum possible bid, let $\epsilon = 1$,
and let the block size be $k = 5$.
Thus, 
if only one user is confirmed, then the price would be set to $6$;
if two users are confirmed, the price would be $7$; and so on. 
Now, if the bid vector is $10, 9, 5, 3, 1$, 
the mechanism would confirm the top two bids and they each pay $7$.
One can see that the mechanism achieves at least $\Theta(k^2 \epsilon)$ social welfare
in the best case:
suppose $\lfloor\frac{M}{\eps} \rfloor \geq k$ and $k/2$ users have true value $M$
while the remaining users have a value of $0$.
Then, 
all the $k/2$ bids at $M$ will be confirmed and each bid
pays $F_{k/2} = M - (k\eps/ 2)$. 
In this case, the mechanism achieves $\Theta(k^2 \epsilon)$ social welfare.
\ignore{
For example, suppose there are exactly $k/2$ users with true value $M$.
In this case, an honest miner will include all of them, and all users are confirmed.
Each user pays $F_{k/2} = M - (k\eps/ 2)$, so each user's utility is $k\eps / 2$.
The miner's revenue is also $k\eps / 2$.
Thus, the social welfare is $(k\eps / 2) \cdot k / 2 + k\eps / 2 = \Theta(k^2 \epsilon)$.
}

Notice that the miner's revenue grows linearly in $t$, the number of the confirmed bids in the block.
On the other hand, any confirmed user's payment also grows linearly in $t$, so each confirmed user's utility actually decreases linearly in $t$.
The miner's revenue and any user's utility as the functions of $t$ can be visualized by \cref{fig:staircase}, which explains why the mechanism is called ``staircase''.

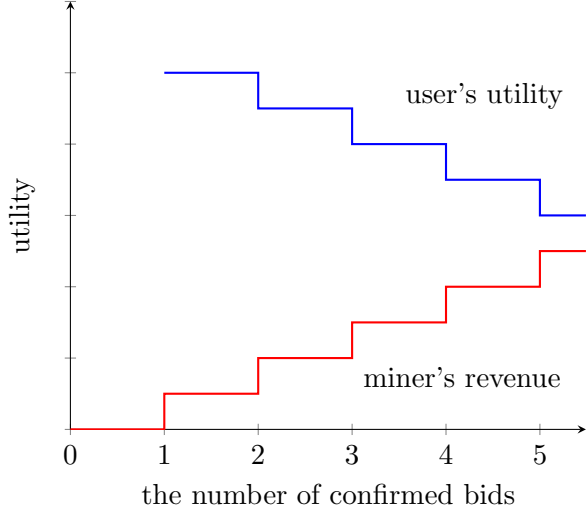
\begin{figure}[H]
    \centering
    \begin{tikzpicture}
        \begin{axis}[%
            ,xlabel=$\text{the number of confirmed bids}$
            ,ylabel=$\text{utility}$
            ,axis x line = bottom,axis y line = left
            ,yticklabels={,,}
            ,ymax=1.2 
            ]
        \addplot+[const plot, no marks, thick] coordinates {(1,1) (2,0.9) (3,0.8) (4,0.7) (5,0.6) (5.49, 0.6)} node[above=1cm,pos=.76,black] {$\text{user's utility}$};
        \addplot+[const plot, no marks, thick] coordinates {(0,0) (1,0.1) (2,0.2) (3,0.3) (4,0.4) (5,0.5) (5.49, 0.5)} node[below=1cm,pos=.76,black] {miner's revenue};
        \end{axis}
        \end{tikzpicture}
        \caption{The miner's revenue and any user's utility as the functions of the number of the confirmed bids in the block.}
        \label{fig:staircase}
\end{figure}

Intuitively, for any coalition consisting of the miner and a user, they do not have incentive to manipulate the number of confirmed bids, as the increase in miner revenue cancels out the decrease in the colluding user's utility.
The following example shows that a user or a miner-user coalition may have $\eps$ extra utility by deviation.
Suppose $M = k = 10$, and $\eps = 1$.
In this case, $F_0 = 0$.
Imagine that there are five users with the true values $8,7,6, 4.95, 4.9$, respectively.
If everyone bids truthfully, then $8,7,6, 4.95$ will be confirmed, since $F_4 = 4$ and $F_5 = 5$.
Notice that the fifth user (with the true value $4.9$) is unconfirmed, so its utility is zero.
However, if the fifth user bids $4.96$ instead, its bid will be confirmed, and $4.95$ will be unconfirmed.
The fifth user pays $F_4 = 4$, and gets the utility $4.9 - 4 = 0.9$.
Notice that the number of the confirmed bids does not change, so the miner is always paid $4\eps$.
Thus, if the miner colludes with the fifth user, their utility increases by $0.9$.
One can easily modify the true values so that the strategic gain is arbitrarily close to $\eps$.

The following theorem shows that 
a strategic user or miner-user coalition cannot
gain more than $\epsilon$. 

\begin{theorem}
The staircase mechanism above satisfies $\eps$-UIC, strict-MIC, and $\eps$-SCP when the miner colludes with at most $1$ user.
\end{theorem}
\begin{proof}
We prove the three incentive compatibility properties separately.

\paragraph{$\eps$-UIC.}
Let $v_i$ be user $i$'s true value.
Without loss of generality, we assume a strategic user $i$ always first injects some fake bids, and then changes its true bid (not the fake bids) from the true value to some other value.
We will show that user $i$'s utility does not increase in either step.
Consequently, user $i$'s utility can never increase even if it plays strategically.

First, we show that 
regardless of the current bid vector.
If user $i$ injects one more fake bid, its utility does not increase.
Suppose $\bfb = (b_1,\dots,b_N)$ is the current bid vector, where some bids might be fake bids injected by user $i$, and $b_i = v_i$ is user $i$'s true bid.
If $b_i$ is already confirmed, i.e.~$x_i(\bfb_{-i}, v_i) = 1$, injecting another fake bid can never lower $t$.
Thus, user $i$'s payment can never be lower after injecting another fake bid.
On the other hand, if $b_i$ is unconfirmed, i.e.~$x_i(\bfb_{-i}, v_i) = 0$, $b_i$ must still be unconfirmed after injecting another fake bid.
\hao{Explain more here?}
Thus, injecting fake bids does not increase user $i$'s utility.

Second, we show that no matter what the current bid vector is, if user $i$ changes its true bid from the true value to some other value, its utility does not increase.
Suppose $\bfb = (b_1,\dots,b_N)$ is the current bid vector, where some of the bids might be the fake bids injected by user $i$, and $b_i = v_i$ is user $i$'s true bid.
Let $t^* = \max_i \{i: b_i \geq F_i\}$.
There are two cases.
\begin{itemize}
\item 
{\bf Case 1: $b_i$ is confirmed under the bid vector $\bfb$.}
Notice that the payment never exceeds the bid, so user $i$'s utility is always non-negative when user $i$ bids truthfully.
Thus, if user $i$'s bid becomes unconfirmed after changing the bid, user $i$'s utility does not increase.
On the other hand, if user $i$'s bid is still confirmed after changing the bid, the number of confirmed bids in the block is still $t^*$ because changing the bid only permutes the order of the top $t^*$ bids.
Thus, user $i$'s payment is still $F_{t^*}$, so user $i$'s utility does not change.
\item 
{\bf Case 2: $b_i$ is unconfirmed when the bid vector is $\bfb$.}
If user $i$ underbids, its bid must still be unconfirmed.
If user $i$ overbids, the number of the confirmed bids must be at least $t^*$, so the payment for each confirmed bid is at least $F_{t^*}$.
Because $x_i(\bfb_{-i}, v_i) = 0$, it must be $v_i \leq F_{t^*+1} = F_{t^*} + \eps$.
Thus, if user $i$'s bid becomes confirmed because of overbidding, user $i$'s utility is at most $v_i - F_{t^*} \leq \eps$.
\end{itemize}


\paragraph{Strict-MIC.}

Without loss of generality, we assume the strategic miner prepares the block in the following order:
the miner chooses a subset of the bids $\bfc' = (c'_1,\ldots,c'_{\ell'})$ from the bid vector 
(not necessarily the top $k$) where $c'_1 \geq \cdots \geq c'_{\ell'}$;
then, the miner injects some fake bids into $\bfc'$.
We will show that the miner's utility does not increase in either step.

First, let $\bfc = (c_1,\ldots,c_k)$ denote the top $k$ bids in the current bid vector 
(if the number of bids is less than $k$, append zeros), 
where $c_1 \geq c_2 \geq \cdots \geq c_k$.
Because $c_i \geq c'_i$ for all $i$, the number of the confirmed bids in $\bfc'$ cannot be more than the number of the confirmed bids in $\bfc$.
Thus, not choosing the top $k$ bids into the block never increases miner's revenue.

Second, let $\bfd = (d_1,\ldots,d_r)$ denote the bid vector that the miner prepares, where $d_1 \geq d_2 \geq \cdots \geq d_r$ for some $r$.
Here, $\bfd$ may or may not contain fake bids injected by the miner.
We will show that if the miner injects one more fake bid $f$, its utility does not increase.
Let $t^* = \max_i \{i: d_i \geq F_i\}$.
In this case, it must be $d_{t^*+1} < F_{t^*+1}$.
To increase the miner's utility, the number of the confirmed bids after injecting $f$ must increase, so we assume it is the case.
Because $d_{t^*+1} < F_{t^*+1}$, if the number of the confirmed bids increases, it must be that $f$ is confirmed and $f \geq F_{t^*+1}$.
Moreover, because the miner only injects one more fake bid to $\bfd$, the number of the confirmed bids after injecting the fake bid is at most $t^* + 1$.
Thus, the revenue that the miner gets increases by at most $\eps$.
The extra cost for injecting $f$ is $F_{t^*+1} \geq \eps$ for any $t \geq 0$.
Therefore, the overall utility does not increase.


\paragraph{$\eps$-SCP.}
Let user $j$ be the colluding user.
Let $\bfd = (d_1,\ldots,d_k)$ denote the top $k$ bids that the miner includes if both the miner and user $j$ are honest, where $d_1\geq \cdots \geq d_k$.
Let $t = \max_i \{i: d_i \geq F_i\}$;
that is, if the miner is honest, $t$ bids will be confirmed.
Next, suppose the coalition strategically includes
the bids $\bfd' = (d'_1,\ldots,d'_r)$ for some $r$,
where $d'_1 \geq d'_2 \geq \ldots \geq d'_r$.
Let $t' = \max_i \{i: d'_i \geq F_i\}$.

First, to increase the joint utility of the coalition, user $j$'s bid must be confirmed 
when the block is $\bfd'$ --- if user $j$'s bid is not included 
under $\bfd'$, then by strict-MIC, 
the miner's utility cannot increase when it chooses $\bfd'$ to be the block, and
obviously user $j$'s utility 
cannot increase either if it is not confirmed under $\bfd'$.
Henceforth, we assume user $j$'s bid is confirmed when the block is $\bfd'$.
There are two possible cases.
\begin{itemize}
\item 
{\bf Case 1: User $j$'s bid is confirmed if the block is $\bfd$.}
In this case, user $j$'s bid is confirmed under both $\bfd$ and $\bfd'$, so the change of user $j$'s utility only depends on its payment.
The payment changes from $F_t$ to $F_{t'}$, so user $j$'s utility increases by $F_t - F_{t'} = (t - t')\eps$ --- 
if $t-t'$ is negative, user $j$'s utility actually decreases.
On the other hand, the miner's revenue decreases by $(t - t')\eps$.
Therefore, the increase in user $j$'s utility 
cancels out the decrease in the miner's revenue, and their joint utility does not change.
\item 
{\bf Case 2: User $j$'s bid is unconfirmed if the block is $\bfd$.}
In this case, user $j$'s true value $v_j$ must be smaller than $F_{t+1}$.
Since user $j$'s bid is unconfirmed when the block is $\bfd$, its utility is zero.
Since user $j$'s bid is confirmed when the block is $\bfd'$, its utility now becomes $v_j - F_{t'} < F_{t+1} - F_{t'}$.
Thus, user $j$'s utility increases by $v_j - F_{t'} < F_{t+1} - F_{t'} = (t + 1 - t')\eps$.
On the other hand, the miner's revenue decreases by $(t - t')\eps$.
Therefore, the joint utility increases by at most $\eps$.
\end{itemize}

\end{proof}

\section{Deferred Proofs of~\cref{sec:MPC-finite}}

\subsection{Strict Incentive Compatibility in MPC-Assisted Model: Necessity of Zero Miner Revenue}
\label{sec:strict-ic-0-miner-rev}
Chung and Shi~\cite{foundation-tfm} showed that the posted-price auction with burning gives strict incentive compatibility in the plain model, assuming infinite block size.
One may hope that with the Bayesian notion of incentive compatibility, we can achieve larger miner revenue. 
Unfortunately, in this section we show that zero-miner revenue is the best we can hope for strict incentive compatibility, even in the Bayesian setting. 

\paragraph{Prelimary: Myerson's lemma for the Bayesian setting.}
We first review the Bayesian version of Myerson's lemma.
Recall that $\bids_{-i}$ denotes all but user $i$'s bid, and $(\bids_{-i}, b_i) = \bids$.
We also let $\mcal{D}_{-i}$ to denote $\mcal{D}_1\times\dots\times\mcal{D}_{i-1}\times\mcal{D}_{i+1}\times\dots\times\mcal{D}_n$, which denotes the distribution of other users' true values. 

\begin{lemma}[Myerson's Lemma~\cite{myerson}]
\label{lem:myerson-bayesian}
Let $\mcal{D} = \mcal{D}_1\times\dots\times\mcal{D}_n$ be the joint distribution of users' true values.
Let $({\bf x}, {\bf p}, \mu)$ be a single-parameter TFM that is Bayesian UIC.
Then, it must be that  
\begin{enumerate}[leftmargin=5mm]
	\item 
	The allocation rule ${\bf x}$ is {\it monotonically non-decreasing}.
	Formally, for any user $i$, and any $b'_i > b_i$, 
	it must be that $\EXPT_{\bfb_{-i}\sim \otherdist} \left[x_i({\bf b}_{-i}, b'_i)\right]  \geq \EXPT_{\bfb_{-i}\sim \otherdist} \left[x_i({\bf b}_{-i}, b_i)\right] $.
	\item 
	The payment rule ${\bf p}$ is defined as follows.
	For any user $i$, and bid $b_i$ from user $i$,
	it must be
	\begin{equation}\label{eq:payment}
		\EXPT_{\bfb_{-i}\sim \otherdist} \left[p_i(\bfb_{-i}, b_i)\right] 
		= \EXPT_{\bfb_{-i}\sim \otherdist} \left[b_i \cdot x_i(\bfb_{-i}, b_i) - \int_0^{b_i} x_i(\bfb_{-i}, t) dt\right].
	\end{equation}
	\end{enumerate}
	\label{lem:myerson}
\end{lemma}

\begin{lemma}[Technical lemma implied by the proof of Myerson's Lemma~\cite{myerson,myerson-lecture-hartline}]
	Let $f(z)$ be a non-decreasing function. 
	Suppose that $ z \cdot (f(z')-f(z)) \leq g(z') - g(z) \leq z' \cdot (f(z')-f(z))  $ 
	for any $z' \geq z \geq 0$, and moreover, $g(0) = 0$.
	Then, it must be that 
	\[
	g(z) = z \cdot f(z) - \int_0^{z} f(t) dt.
	\]
	\label{lem:sandwich}
\end{lemma}

\paragraph{Necessity of zero miner revenue.}
Henceforth we use the following simplified notation.
\[\overline{x_i}(\cdot) = \underset{\bids_{-i}\sim\otherdist}{\E} [{\bf x}(\bids_{-i}, \cdot)], \quad 
    \overline{p_i}(\cdot) = \underset{\bids_{-i}\sim\otherdist}{\E} [{\bf p}(\bids_{-i}, \cdot)], \quad 
    \overline{\mu_i}(\cdot) =\underset{\bids_{-i\sim\otherdist}}{\E}[\mu(\bids_{-i}, \cdot)].\]
The following technical lemma was given in~\cite{foundation-tfm}.
\begin{lemma}[Lemma 4.8 in~\cite{foundation-tfm}]
\label{lemma:randomInequality}
    Let $({\bf x}, {\bf p}, \mu)$ be any (possibly randomized) TFM in the Bayesian setting.
	If $({\bf x}, {\bf p}, \mu)$ is Bayesian SCP against a $(\rho,1)$-sized coalition, then for any bid vector $\bfb$, user $i$, and $r, r'$ such that $r < r'$, it must be \[
		r \cdot \left(\overline{x_i}(r') - \overline{x_i}(r)\right)
		\leq \pi(r') - \pi(r)
		\leq r' \cdot \left(\overline{x_i}(r') - \overline{x_i}(r)\right),
	\]
where $\pi(r):= \overline{p_i}(r) - \rho\overline{\mu_i}(r)$.
\end{lemma}
The following result shows that if we allow the strategic players to inject fake bids, then the miner's revenue can only be $0$ if the mechanism is UIC, MIC, and $1$-SCP.
Actually, in the proof of the lower bound, we only need the deviation where the miners in the coalition injecting fake bids, and colluding users only bid untruthfully. 

We first show that if a TFM is Bayesian UIC and Bayesian SCP against $(\rho,1)$-sized coalition, then the miner revenue must be independent from each user's bid. Without loss of generality, we assume that $0$ is the minimum value in the support of $\mcal{D}_i$ for $i\in[n]$. 
\begin{lemma}
	\label{lemma:constantRevenue}
	Let $\mcal{D} = \mcal{D}_1\times\dots\times\mcal{D}_n$ be the joint distribution of users' true values.
	Let $({\bf x}, {\bf p}, \mu)$ be any (possibly randomized) TFM in the MPC model.
	If $({\bf x}, {\bf p}, \mu)$ is Bayesian UIC and Bayesian SCP against a $(\rho,1)$-sized miner-user coalition, then
	for any user $i$, any bid $b$, it must be 
	\begin{equation}
	\label{eqn:constant-revenue}
	    \overline{\mu_i}(b) = \overline{\mu_i}(0).
	\end{equation}

	In other words, the miner's revenue is a constant that is independent of user $i$'s bid $b$ when other bids $\bfb_{-i}$ are drawn from the distribution $\otherdist$.
\end{lemma}
\begin{proof}
	Define $\widetilde{\pi}(r)$ as \[
		\widetilde{\pi}(r) = \overline{p_i}(r) - \rho\overline{\mu_i}(r) - 
		(\overline{p_i}(0) - \rho\overline{\mu_i}(0)).
	\]
	By Lemma~\ref{lemma:randomInequality}, 
	and the fact that definition of $\widetilde{\pi}(r)$
	and ${\pi}(r)$ differs by only a fixed constant,
	it must be that 
	\begin{equation}
			r \cdot \left(\overline{x_i}(r') - \overline{x_i}(r)\right)
			\leq \widetilde{\pi}(r') - \widetilde{\pi}(r)
			\leq r' \cdot \left(\overline{x_i}(r') - \overline{x_i}(r)\right).
	\label{eqn:sandwich}
	\end{equation}
	Therefore, we have the following two inequalities:
	\begin{align*}
	    r\cdot [\overline{x_i}(r') - \overline{x_i}(r)]&\leq \widetilde{\pi}(r') - \widetilde{\pi}(r)\\
	    r\cdot [\overline{x_i}(r') - \overline{x_i}(r)]&\geq \widetilde{\pi}(r') - \widetilde{\pi}(r)
	\end{align*}
	Now, observe that the above expression strictly 
	agrees with the ``payment sandwich'' in 
	the proof of Myerson's Lemma~\cite{myerson,myerson-lecture-hartline}.
	Furthermore, we have that 
	$\widetilde{\pi}(0)= 0$
	by definition; and 
	${\bf x}$ must be monotone because the TFM is UIC and satisfies
	Myerson's Lemma. 
	Due to Lemma~\ref{lem:sandwich}, 
	it must be that 
	$\widetilde{\pi}(\cdot)$ obeys the unique payment rule
	specified by Myerson's Lemma; that is, 
	\[
		\widetilde{\pi}(r) = 
		\left[b_i \cdot \overline{x_i}(b_i) - \int_0^{b_i} \overline{x_i}(t) dt \right].
	\]
	On the other hand, since the TFM is UIC, its payment rule  
	itself must also satisfy the same expression (Eq.(\ref{eq:payment})), that is, 
	\[
		\overline{p_i}(b_i)
		= b_i \cdot \overline{x_i}(b_i) - \int_0^{b_i} \overline{x_i}(t) dt.
	\]
	We therefore have that 
	\[
		\widetilde{\pi}(r)
	= \overline{p_i}(b_i).
	\]
	In other words, $\rho\overline{\mu_i}(r) =  \rho\overline{\mu_i}(0) - \overline{p_i}(0)$.
	Because $\overline{p_i}(0) = 0$, we conclude $\overline{\mu_i}(r) =  \overline{\mu_i}(0)$.
\end{proof}

\noindent Note that the result in Lemma~\ref{lemma:constantRevenue} holds even if users do not inject any fake bids. This provides a stronger impossibility result.


Now we show that, if in addition the mechanism $\mecha$ is Bayesian MIC, then the total miner revenue can only be $0$.
\begin{theorem}
\label{thm:inject-constantRevenue}
Let $\{\mcal{D}^{(n)}\}_n$ be a sequence of distributions where  $\mcal{D}^{(n)} = \mcal{D}_1\times\dots\times\mcal{D}_{n}$ is the joint distribution of $n$ users' true values, where user $i$'s true value is drawn from $\mcal{D}_i$ independently.
Let $\mecha$ be any (possibly randomized) TFM in the MPC model. 
If $\mecha$ is Bayesian UIC, Bayesian MIC against $\rho$-sized miner coalition and Bayesian SCP against $(\rho,1)$-sized miner-user coalition, then 
\[\underset{\bids\sim\mcal{D}^{(n)}}{\E} [\mu(\bids)]=0.\]
\end{theorem}
\begin{proof}
For any $n\geq 2$, we have the following claim:
\begin{lemma}
\label{lem:inject-induction}
If $\mecha$ is Bayesian MIC against $(\rho,1)$-sized miner-user coalition, then $\underset{\bids\sim\mcal{D}^{(n)}}{\E} [\mu(\bids)] \leq \underset{\bids'\sim\mcal{D}^{(n-1)}}{\E} [\mu(\bids')]$.
\end{lemma}
For now assume Lemma~\ref{lem:inject-induction} holds and we explain why Theorem~\ref{thm:inject-constantRevenue} follows from it. The proof of Lemma~\ref{lem:inject-induction} appears right afterwards.
By induction on $n$, we have that 
\[\underset{\bids\sim\mcal{D}^{(n)}}{\E} [\mu(\bids)] \leq \underset{b\sim\mcal{D}_{1}}{\E} [\mu(b)].\]
By Lemma~\ref{lemma:constantRevenue}, for any $b\in\supp{\mcal{D}_1}$, it should be that $\mu(b) = \mu(0)$. Therefore, 
\[\underset{b\sim\mcal{D}_{1}}{\E} [\mu(b)]\leq \mu(0) = 0,\]
where the last equality comes from the requirement that the miner's revenue cannot exceeds the payment of the single identity, who will pay at most what it bids.
Theorem~\ref{thm:inject-constantRevenue} thus follows.
\end{proof}

\paragraph{Proof of Lemma~\ref{lem:inject-induction}}
Since $\mecha$ is Bayesian-SCP against $(\rho,1)$-sized coalition, it must be that for any user $i$,
\begin{equation}
    \underset{\bids\sim\mcal{D}^{(n-1)}}{\E} [\rho \mu(\bids, 0)]\leq \underset{\bids\sim\mcal{D}^{(n-1)}}{\E} [\rho \mu(\bids)].\label{eqn:inject-induction}
\end{equation}
Otherwise, the miners can collude with user $i$, ask user $i$ to bid, and inject $0$ and increase the coalition's miner revenue while it does not need to pay anything for injecting the $0$-bid. This violates the MIC condition.

By the law of total expectation, we have that
\begin{align*}
    \underset{\bids\sim\mcal{D}^{(n)}}{\E} [\mu(\bids)] &= \int_{0}^{+\infty} \underset{\bids'\sim\mcal{D}^{(n-1)}}{\E} [\mu(\bids', r)] f(r) dr & \\
    &=\int_{0}^{+\infty} \underset{\bids'\sim\mcal{D}^{(n-1)}}{\E} [\mu(\bids', 0)] f(r) dr &\text{By Lemma~\ref{lemma:constantRevenue}}\\
    &=\underset{\bids'\sim\mcal{D}^{(n-1)}}{\E} [\mu(\bids', 0)] \leq \underset{\bids'\sim\mcal{D}^{(n-1)}}{\E} [\mu(\bids')] &\text{By \eqref{eqn:inject-induction}}
\end{align*}
Lemma~\ref{lem:inject-induction} thus follows.





\subsection{Proof of ~\cref{lem:hartline}}
\label{section:missing-proof}
\begin{lemma}[Restatement of~\cref{lem:hartline}]
Let $\mecha$ be any (possibly random) mechanism that is Bayesian UIC and Bayesian SCP against $(\rho^*,2)$-sized coalition for some $\rho^*\in(0,1]$. 
Suppose each user's true value is drawn i.i.d. from a distribution $\mathcal{D}$.
Then for any user $i$ and $j$, for any bid $b_j$ and $b'_j$, it must be that for any $\ell\geq 1$,
\[\underset{v,\bids\sim\mathcal{D}^{\ell}}{\E}[{\sf util}^i(v,b_j,\bids)] = \underset{v,\bids\sim\mathcal{D}^{\ell}}{\E}[{\sf util}^i(v,b'_j,\bids)].\]
\end{lemma}

\begin{proof}

In this proof, we use the following notations for simplicity.
For any fixed $\ell\geq 1$, for any user $i$ and $j$, we define the following notations:
\[\overline{x_i}(\cdot,\cdot) = \underset{\bids_{-i,j}\sim\mathcal{D}^{\ell-1}}{\E} [x_i(\cdot,\cdot,\bids)], \quad 
    \overline{p_i}(\cdot,\cdot) = \underset{\bids_{-i,j}\sim\mathcal{D}^{\ell-1}}{\E} [p_i(\cdot,\cdot,\bids)], \quad 
   \overline{\mu}(\cdot,\cdot) = \underset{\bids_{-i,j}\sim\mathcal{D}^{\ell-1}}{\E} [\mu(\cdot,\cdot,\bids)].\]

Imagine that user $i$ has true value $v$ and user $j$ has true value $y$. 
Then for any feasible $\rho\leq \rho^*$, it must be that
\begin{align*}
    & \text{Honest utility} = [v\cdot \overline{x_i}(v,y) - \overline{p_i}(v,y)]
    +[y\cdot \overline{x_j}(v,y) - \overline{p_j}(v,y)]
    +\rho\overline{\mu}(v,y)\\
    \geq & \text{Overbid utility} = [v\cdot \overline{x_i}(v,z) - \overline{p_i}(v,z)]
    +[y\cdot \overline{x_j}(v,z) - \overline{p_j}(v,z)]
    +\rho\overline{\mu}(v,z).
\end{align*}
Otherwise, the miner can collude with user $i$ with true value $v$ and user $j$ with true value $y$ and ask user $j$ to overbid to some $z\geq y$. 
This will increase the joint utility of the coalition, which contradicts $(\rho^*,2)$-SCP.
For the same reason, if user $j$'s true value is $z$, then 
\begin{align*}
    & \text{Honest utility} =[v\cdot \overline{x_i}(v,z) - \overline{p_i}(v,z)]
    +[z\cdot \overline{x_j}(v,z) - \overline{p_j}(v,z)]
    +\rho\overline{\mu}(v,z)\\
    \geq & \text{Underbid utility} =[v\cdot \overline{x_i}(v,y) - \overline{p_i}(v,y)]
    +[z\cdot \overline{x_j}(v,y) - \overline{p_j}(v,y)]
    +\rho\overline{\mu}(v,y).
\end{align*}
Combining these two inequalities together, we get the following payment difference sandwich.
For any $z\geq y$, we have
\begin{align*}
    &v[\overline{x_i}(v,z) -  \overline{x_i}(v,y)] + 
z[\overline{x_j}(v,z) -  \overline{x_j}(v,y)]
+\rho[\overline{\mu}(v,z)-\overline{\mu}(v,y)]\\
\geq
&\overline{p_i}(v,z)-\overline{p_i}(v,y)+\overline{p_j}(v,z)-\overline{p_j}(v,y)\\
\geq & v[\overline{x_i}(v,z) -  \overline{x_i}(v,y)] + 
y[\overline{x_j}(v,z) -  \overline{x_j}(v,y)]
+\rho[\overline{\mu}(v,z)-\overline{\mu}(v,y)]
\end{align*}
Divide the inequality with $z-y$ and take limit $y \rightarrow z$, we get
\begin{equation}
\label{eqn:scp-diff}
    v\cdot \frac{\partial}{\partial z} \overline{x_i}(v,z) + z\cdot \frac{\partial }{\partial z} \overline{x_j}(v,z) + \rho\frac{\partial }{\partial z} \overline{\mu}(v,z) = \frac{\partial}{\partial z} \overline{p_i}(v,z) + \frac{\partial}{\partial z} \overline{p_j}(v,z).
\end{equation}
Note that \cref{eqn:scp-diff} should hold for at least two different values of $\rho \leq \rho^*$.
Hence, it must be that $\frac{\partial }{\partial z} \overline{\mu}(v,z) = 0$. 
\cref{eqn:scp-diff} thus becomes
\begin{align}
\label{eqn:2-uic}
    &v[\overline{x_i}(v,z) -  \overline{x_i}(v,y)] + 
z[\overline{x_j}(v,z) -  \overline{x_j}(v,y)]\nonumber\\
\geq
&\overline{p_i}(v,z)-\overline{p_i}(v,y)+\overline{p_j}(v,z)-\overline{p_j}(v,y)\nonumber\\
\geq & v[\overline{x_i}(v,z) -  \overline{x_i}(v,y)] + 
y[\overline{x_j}(v,z) -  \overline{x_j}(v,y)].
\end{align}
This is equivalent to say: when user $j$ changes its bid, the joint utility of user $i$ and user $j$ should not increase.
That means, for any $v$, if a user $j$ with true value $y$ changes its bid from $y$ to $z$, it must be that
\begin{align*}
    \text{i-gain}(v,y\rightarrow z)&:=\underset{\bids_{-i,j}\sim\mathcal{D}^{\ell-1}}{\E}{{\sf util}^i}(v,z,\bids_{-i,j}) - \underset{\bids_{-i,j}\sim\mathcal{D}^{\ell-1}}{\E}{{\sf util}^i}(v,y,\bids_{-i,j})\\ &\leq \underset{\bids_{-i,j}\sim\mathcal{D}^{\ell-1}}{\E}{{\sf util}^j}(v,y,\bids_{-i,j}) - \underset{\bids_{-i,j}\sim\mathcal{D}^{\ell-1}}{\E}{{\sf util}^j}(v,z,\bids_{-i,j}) := \text{j-loss}(v,y\rightarrow z)
\end{align*}
Since the mechanism is UIC, by the same proof as of \cite{goldberghartline}, we get:
\begin{align*}
    \underset{v\sim\mathcal{D}}{\E}[\text{i-gain}(v, y\rightarrow z)] \leq & \underset{v\sim\mathcal{D}}{\E}[\text{j-loss}(v, y\rightarrow z)]\\
    \leq & \underset{v\sim\mathcal{D}}{\E}[(z-y)(\overline{x_j}(v,z) - \overline{x_j}(v,y))].
\end{align*}
Now consider the situation where user $j$ changes its bid from $b_j$ to $b_j'$.
Without loss of generality, we assume that $b'_j\geq b_j$.
If we divide the interval between $[b_j, b_j']$ into $L$ equally sized segments $b_j^{(0)},\dots,b_j^{(L)}$, then the total gain for user $i$ can be bounded by 
\begin{align*}
    \underset{v\sim\mathcal{D}}{\E}[\text{i-gain}(v, b_j\rightarrow b_j')] &= \sum_{l=0}^{L-1} \underset{v\sim\mathcal{D}}{\E}[\text{i-gain}(v, b_j^{(l)}\rightarrow b_j^{(l+1)})]\\
    &\leq \sum_{l=0}^{L-1}(b_j^{(l+1)} - b_j^{(l)}) \underset{v\sim\mathcal{D}}{\E}\left[\overline{x_j}(v,b_j^{(l+1)}) - \overline{x_j}(v,b_j^{(l)})\right]\\
    &=\frac{b'_j-b_j}{L} \underset{v\sim\mathcal{D}}{\E}\left[\overline{x_j}(v,b_j') - \overline{x_j}(v,b_j)\right].
\end{align*}
This holds for any $L$. 
Taking limit for $L \rightarrow \infty$, we have that
\[\underset{v\sim\mathcal{D}}{\E}[\text{i-gain}(v, b_j\rightarrow b_j')]\leq 0.\]
Since $\underset{v\sim\mathcal{D}}{\E}[\text{i-gain}(v, b_j\rightarrow b_j')] = - \underset{v\sim\mathcal{D}}{\E}[\text{i-gain}(v, b_j'\rightarrow b_j)]$, we have that $\underset{v\sim\mathcal{D}}{\E}[\text{i-gain}(v, b_j\rightarrow b_j')] = 0$, for arbitrary $b_j$ and $b'_j$. 
The lemma thus follows.
\end{proof}

\subsection{Full Proof of \cref{lemma:sym-0-bid}}
\label{sec:sym-0-bid}
In this section, we provide a full proof of \cref{lemma:sym-0-bid} 
assuming the symmetry assumption in~\cref{sec:TFM-plain-model}.

	By \cref{lem:hartline}, 
for any $i, j, b_j, \ell\geq 1$, we have
	\begin{align*}
		\underset{v_i,\bids_{-i,j}\sim\mathcal{D}^{\ell}}{\E}[{\sf util}^i(v_i,b_j,\bids_{-i,j})] = \underset{v_i,\bids_{-i,j}\sim\mathcal{D}^{\ell}}{\E}[{\sf util}^i(v_i,0_j,\bids_{-i,j})].
	\end{align*}
Therefore, it suffices
to prove that 
for any $i, j, \ell$, 
$\underset{v_i,\bids_{-i,j}\sim\mathcal{D}^{\ell}}{\E}[{\sf util}^i(v_i,0_j,\bids_{-i,j})] \leq \underset{v_i,\bids_{-i,j}\sim\mathcal{D}^{\ell}}{\E}[{\sf util}^i(v_i,\bids_{-i,j})]$.


	
Suppose for the sake of contradiction, the above statement is not true,
that is, there exist some $i, j, \ell$, such that 
$\underset{v_i,\bids_{-i,j}\sim\mathcal{D}^{\ell}}{\E}[{\sf util}^i(v_i,0_j,\bids_{-i,j})] > \underset{v_i,\bids_{-i,j}\sim\mathcal{D}^{\ell}}{\E}[{\sf util}^i(v_i,\bids_{-i,j})]$.
Then there must exists a $v_i$, such that $\underset{\bids_{-i,j}\sim\mathcal{D}^{\ell-1}}{\E}[{\sf util}^i(v_i,0_j,\bids_{-i,j})] > \underset{\bids_{-i,j}\sim\mathcal{D}^{\ell-1}}{\E}[{\sf util}^i(v_i,\bids_{-i,j})]$.

Consider  an arbitrary fake identity $m$ registered by the miner.
There are two possible cases.

\paragraph{Good identity $m$: } 
$\underset{\bids_{-i,j}\sim\mathcal{D}^{\ell-1}}{\E}[{\sf util}^i(v_i,0_{m},\bids_{-i,j})] > \underset{\bids_{-i,j}\sim\mathcal{D}^{\ell-1}}{\E}[{\sf util}^i(v_i,\bids_{-i,j})]$.
\ignore{
		Imagine that the world is $(v_i,\bids_{-i,j})$.
		In the honest case, user $i$'s utility is $\underset{\bids_{-i,j}\sim\mathcal{D}^{\ell}}{\E}[{\sf util}^i(v_i,\bids_{-i,j})]$.
		However, the miner can collude with user $i$ and inject $0_{m}$.
		Then, user $i$'s utility becomes $\underset{\bids_{-i,j}\sim\mathcal{D}^{\ell}}{\E}[{\sf util}^i(v_i,0_{m},\bids_{-i,j})] > \underset{\bids_{-i,j}\sim\mathcal{D}^{\ell}}{\E}[{\sf util}^i(v_i,\bids_{-i,j})]$.
		By~\cref{thm:inject-constantRevenue}, the miner's utility is always zero.
		Thus, the joint utility of the miner and user $i$ increases, which contradicts Bayesian SCP.
}		

\paragraph{Bad identity $m$:}
$\underset{\bids_{-i,j}\sim\mathcal{D}^{\ell-1}}{\E}[{\sf util}^i(v_i,0_{m},\bids_{-i,j})] \leq \underset{\bids_{-i,j}\sim\mathcal{D}^{\ell-1}}{\E}[{\sf util}^i(v_i,\bids_{-i,j})] < 
\underset{\bids_{-i,j}\sim\mathcal{D}^{\ell-1}}{\E}[{\sf util}^i(v_i,0_{j},\bids_{-i,j})]$.

Now, suppose the miner samples a fake identity $m$.   
\elaine{TODO: state assumption earlier.}
Over the choice of $m$,   
either $\Pr(\text{Good identity $m$}) \geq \frac{1}{2}$ or $\Pr(\text{Bad identity $m$}) \geq \frac{1}{2}$.
If $\Pr(\text{Good identity $m$}) \geq \frac{1}{2}$, then 
suppose that the world consists of 
$\ell$ users not including $j$, 
and the miner forms a coalition with user $i$ whose true value is $v_i$.
The miner can sample a random identity $m$, and if 
it is a good identity, the miner can inject a fake bid $0_m$, and the coalition 
can strictly gain. This violates SCP when $c = 1$. 

Henceforth, we focus on the case when $\Pr(\text{Bad identity $m$}) \geq \frac{1}{2}$.  
In this case, there are 
two possibilities, either 
with probability at least $1/4$ 
over the choice of the identity $m$, 
for all $v'_i$, 
\begin{equation}
\underset{\bids_{-i,j}\sim\mathcal{D}^{\ell-1}}{\E}[{\sf util}^i(v'_i,0_{m},\bids_{-i,j})] 
\leq \underset{\bids_{-i,j}\sim\mathcal{D}^{\ell-1}}{\E}[{\sf util}^i(v'_i,0_{j},\bids_{-i,j})],
\label{eqn:case2-3}
\end{equation}
or with probability at least $1/4$ over the choice of $m$, there exists
some $v'_i$ such that 
$\underset{\bids_{-i,j}\sim\mathcal{D}^{\ell-1}}{\E}[{\sf util}^i(v'_i,0_{m},\bids_{-i,j})] 
>
\underset{\bids_{-i,j}\sim\mathcal{D}^{\ell-1}}{\E}[{\sf util}^i(v'_i,0_{j},\bids_{-i,j})]
$.
If it is the 
latter case, 
then, consider a scenario where the miner colludes with user $i$ whose true value is $v'_i$,
and user $j$ whose true value is $0$, and  
the rest of the world is a random variable $\bids_{-i, j}$. Now, the miner 
can sample a random fake identity $m$, and see if 
dropping $0_j$ and injecting $0_m$ can help its friend $i$. If so, it performs
this strategic behavior. This strategy can strictly help the coalition which
violates SCP for $c = 2$.

It suffices to rule out the former case, that is, 
with probability at least $1/4$ 
over the choice of the identity $m$, 
for all $v'_i$, 
\Cref{eqn:case2-3}
is satisfied.
Recall also, for $v_i$ specifically, we have strict inequality, that is,  
$\underset{\bids_{-i,j}\sim\mathcal{D}^{\ell-1}}{\E}[{\sf util}^i(v_i,0_{m},\bids_{-i,j})] 
< \underset{\bids_{-i,j}\sim\mathcal{D}^{\ell-1}}{\E}[{\sf util}^i(v_i,0_{j},\bids_{-i,j})]$. 
Thus, 
$\underset{\bids_{-j}\sim\mathcal{D}^{\ell}}{\E}[{\sf util}^i(0_{m},\bids_{-j})] 
< \underset{\bids_{-j}\sim\mathcal{D}^{\ell}}{\E}[{\sf util}^i(0_{j},\bids_{-j})]$.

For every bad identity $m$ that additionally satisfies 
\Cref{eqn:case2-3}, 
there must exist some $i' \neq i$ and $i\neq j$, and some $b_{i'} > 0$,
such that 
\begin{equation}
\underset{\bids_{-j, i'}\sim\mathcal{D}^{\ell-1}}{\E}[{\sf util}^{i'}(0_{m}, b_{i'}, \bids_{-j, i'})]
> \underset{\bids_{-j, i'}\sim\mathcal{D}^{\ell-1}}{\E}[{\sf util}^{i'}(0_{j}, b_{i'}, \bids_{-j, i'})]
\label{eqn:case2}
\end{equation}
We can prove the above claim by contradiction.
Suppose for the sake of contradiction that for all $i' \neq i$ and $i \neq j$, 
and for all $b_{i'}$, 
$\underset{\bids_{j, i'}\sim\mathcal{D}^{\ell-1}}{\E}[{\sf util}^{i'}(0_{m},b_{i'}, 
\bids_{-j, i'})]
\leq \underset{\bids_{j, i'}\sim\mathcal{D}^{\ell-1}}{\E}[{\sf util}^{i'}(0_{j},
b_{i'}, 
\bids_{-j, i'})]$.
Therefore, it must be that 
for any $i' \neq i$ and $i \neq j$, 
$\underset{\bids_{-j}\sim\mathcal{D}^{\ell}}{\E}[{\sf util}^{i'}(0_{m},\bids_{-j})]
\leq \underset{\bids_{-j}\sim\mathcal{D}^{\ell}}{\E}[{\sf util}^{i'}(0_{j},
\bids_{-j})]$.
\ignore{
By the assumption of Case 2,
$\underset{\bids_{-i,j}\sim\mathcal{D}^{\ell}}{\E}[{\sf util}^i(v_i,0_{m},\bids_{-i,j})] 
< 
\underset{\bids_{-i,j}\sim\mathcal{D}^{\ell}}{\E}[{\sf util}^i(v_i,0_{j},\bids_{-i,j})]$.
Additionally, the bid $0_j$ or $0_m$ always has utility $0$.
}
Therefore, we have that 
\begin{equation}
\underset{\bids_{-j}\sim\mathcal{D}^{\ell}}{\E}\left[{\sf USW}(0_{m},\bids_{-j})\right] 
 <  
\underset{\bids_{-j}\sim\mathcal{D}^{\ell}}{\E}\left[{\sf USW}(0_{j},\bids_{-j})\right]
\label{eqn:case2-2}
\end{equation}
where ${\sf USW}(\bids)$ denotes the social welfare 
for all users (i.e., sum of all user utilities) 
when the bid vector is $\bids$.
However, 
by our symmetry assumption in~\cref{sec:TFM-plain-model}, 
it must be that 
$\underset{\bids_{-j}\sim\mathcal{D}^{\ell}}{\E}\left[{\sf USW}(0_{m},\bids_{-j})\right] =
\underset{\bids_{-j}\sim\mathcal{D}^{\ell}}{\E}\left[{\sf USW}(0_{j},\bids_{-j})\right]
$, which contradicts \Cref{eqn:case2-2}.

Let $i'$ be a user such that \Cref{eqn:case2}
happens with probability at least 
$1/4(\ell + 1)$ over the choice of $m$ --- clearly, such a user must exist 
since we are assuming that 
with probability at least $1/4$ over the choice of $m$, where $m$ is a bad
identity satisfying \Cref{eqn:case2-3}.
Now, imagine that the world consists 
of $\ell+1$ users including both $i$ and $j$,  
and the miner forms a coalition with users $i'$ and $j$.
The miner samples a random fake identity $m$, and if the identity helps $i'$ 
in the sense that \Cref{eqn:case2} holds, then 
the coalition replaces $j$'s bid $0_j$ with $0_m$.
This strategy strictly increases the coalition's joint utility, and this
violates SCP when $c = 2$. 

\section{Multi-Party Computation Protocol Realizing $\fmec$}
\label{sec:mpc}
So far in the paper, we have assumed that the transaction fee mechanism
is implemented by a trusted ideal functionality $\fmec$. 
In this section, we show how to instantiate $\fmec$ in the real 
world with cryptography.
The protocol described in this section uses generic MPC. 
However, as mentioned
in \Cref{rem:efficientmpc}, the MPC-assisted mechanisms described
in this paper actually need not employ generic MPC to be instantiated in practice ---
we describe efficient instantiations for our specific
protocols in \Cref{sec:efficientMPC}.

\paragraph{Terminology and model.}
Imagine that there are $m$ {\it miners} 
and  a set of user {\it identities}.
Since each user can assume multiple identities,
henceforth, we often use the term {\it identities} to refer to 
the set of purported user identities.
We assume that the miners can communicate with each other through 
a pairwise private channel. 
Further, every user identity can communicate with 
every miner through a pairwise private channel.
Morever, there is a broadcast channel
among the miners and the user identities.
We assume that all channels are authenticated, i.e., every message
is marked with the true sender.
Further, we assume a synchronous model of communication, i.e.,
the protocol proceeds in rounds and messages sent by honest parties
will be received by  
honest recipients at the beginning of the next round.

We assume that at the beginning of the protocol, the miners have reached consensus on 
the set of user identities that will participate in the auction.
For example, the consensus can be achieved in the 
following manner: every user identity announces itself to all miners. 
Then, each of the $m$ miners broadcasts to all miners a candidate set 
consisting of the identities it has heard.
Any identity that appears in the majority 
of the miners' candidate sets will be permitted into the auction.
As long as the majority of the miners are honest,
then any honest user identity will be included in the final permitted list. 

The parties
now execute an interactive protocol at the end of which 
all parties, including the miners and user identities,
learn the outcome of the auction, including which identities' bids
are confirmed and how much each confirmed bid pays.  
In our actual protocol, the user identities need not communicate
with each other. Each user identity communicates only with the miners ---
either it sends a direct message 
to a miner over the pairwise private channel,  
or it broadcasts a message which can be seen by all miners.

During the protocol, if  a subset of parties (miners or 
user identities) form a coalition, we assume
that the coalition  
has the advantage of performing a so-called ``rushing attack''.
Specifically, in any round $r$, parties in the coalition can observe
honest parties' messages sent to coalition members
or posted on the broadcast channel, before deciding what messages
coalition members want to send in the same round $r$.

\ignore{
The miners will first reach a consensus on the set of identities that are going to participate.
We assume that each identity has access to the broadcast channel, and a private channel with each of the miners.
In the honest protocol, each identity only need to communicate with miners, no interaction between the identities is needed.
All communication channels are authenticated, i.e., messages carry the sender’s
identity. 
Moreover, the network is synchronous and the protocol proceeds in rounds.
The protocol execution is parametrized with a security parameter $\secu$. 
We assume that the coalition may perform a \emph{rushing} attack: 
In the $r$-th round, it waits for all honest players (those not in the coalition) to send messages in round $r$ and decides what messages the identities or miners in the coalition send in round $r$.
At the end, the protocol outputs the outcome of the mechanism. 
}

\subsection{Building Blocks}
We first introduce some building blocks used in the protocol.

\subsubsection{Commitment Scheme}
\label{sec:prelim-com}

A commitment scheme, parametrized by a security parameter $\lambda$, 
and a message space $\{0, 1\}^{\ell(\lambda)}$ where $\ell(\cdot)$ is a polynomial in $\lambda$,
has two phases: 
\begin{itemize}
    \item Commitment phase:
the committer who has a message $X \in \{0, 1\}^{\ell(\lambda)}$ 
samples some random 
coins $r \getr \{0, 1\}^\lambda$, and computes 
the commitment  
$\widehat{X} \leftarrow {\sf comm}(X, r)$.
It sends 
the commitment $\widehat{X}$ to the receiver.
    \item Open phase: The committer sends the pair $(X,r)$ to the receiver. 
The receiver outputs ``accept'' if ${\sf comm}(X, r) = \widehat{X}$;
otherwise, it outputs ``reject''.
\end{itemize}
In our protocol, we require that the commitment scheme must satisfy the following two properties.
\begin{itemize}
    \item {\bf Perfect binding}: for any $X\neq X'$, and for any $r, r'$, it must be that
    \[{\sf comm}(X,r)\neq {\sf comm}(X', r').\]
    \item {\bf Computationally hiding}: for any $X$ and $X'$, it must be that
    \[\{{\sf comm}(X,r), r\getr \{0, 1\}^\lambda \}
\equiv_c \{{\sf comm}(X',r), r\getr \{0, 1\}^\lambda\},\]
    where 
$\equiv_c$ denotes computational indistinguishability.
\end{itemize}

\subsubsection{Shamir Secret Sharing}
In our final protocol, each user identity will split its bid 
into $m$ shares, one for each miner, using 
a $t$-out-of-$m$ Shamir secret sharing scheme.
Henceforth let $\mathbb{F}$ denote some finite field. 
A $t$-out-of-$m$ Shamir secret sharing 
consists of two algorithms, ${\sf share}$ and ${\sf reconstruct}$. 
\begin{itemize}
    \item {\sf share} takes as an input a secret $s \in \mathbb{F}$, 
and outputs $m$ shares $(s_1,\dots,s_m) \in \mathbb{F}^m$ of the secret.
    \item {\sf reconstruct} takes as input a set $I \subseteq [m]$, 
and the corresponding shares $\{s_i\}_{i\in I}$, and outputs the corresponding secret if and only if $|I|\geq t$. Otherwise, the algorithm returns $\bot$.
\end{itemize}
A $t$-out-of-$m$ secret sharing satisfies the following two properties:
\begin{itemize}
    \item {\bf Correctness:} For any secret $s$ and any set $I\subseteq[m]$ such that $|I|\geq t$, it must be that
    \[\Pr[
(s_1,\dots,s_m)\leftarrow {\sf share}(s): {\sf reconstruct}(I,\{s_i\}_{i\in I}) = s 
]
=1.\]
    \item {\bf Security:} For any two secret $s$ and $s'$, and for all set $I\subseteq[m]$ such that $|I|\leq t-1$, it must be that
    \[\{\{s_i\}_{i\in I}:(s_1,\dots,s_m)\leftarrow {\sf share}(s)\}\equiv \{\{s_i\}_{i\in I}:(s_1,\dots,s_m)\leftarrow {\sf share}(s')\}.\]
where $\equiv$ denotes identically distributed. In addition, Shamir secret sharing also satisfies the following properties. 
For any set $I$ such that $|I| < t$,
\[\{\{s_i\}_{i\in I}:(s_1,\dots,s_m)\leftarrow {\sf share}(s)\} \equiv \{\{u_i\}_{i \in I}: u_i\text{ uniformly randomly chosen from } \F\}.\]
\end{itemize}

\subsubsection{Honest-Majority Multi-CRS NIZK}
\label{sec:nizk}
In our protocol, user identities will need to rely on zero-knowledge proofs to prove that they have correctly shared their bids.
We will use a non-interactive zero-knowledge proof (NIZK). 
Since we assume the majority of the miners are honest, 
NIZK can be instantiated without a common reference string (CRS),
using an honest-majority multi-CRS NIZK scheme~\cite{groth2014cryptography}.
Specifically, every miner $j \in [m]$ acts as a CRS contributor, and 
posts a CRS denoted $\crs_j$ to the broadcast channel.
For any miner $j'$ who did not post a CRS, we treat
its $\crs_{j'}$ as $0$.
Given the collection of all CRSes $\{\crs_j\}_{j \in [m]}$, 
a prover can prove an NP statement 
given a valid witness.
As long as a majority of the miners (i.e., CRS contributors) are honest, 
the NIZK scheme satisfies completeness, zero-knowledge, and simulation sound extractability, as defined below.

For an NP language $L$, let $\mcal{R}_L(\stmt,w)$ denote the NP relation corresponding to the language $L$, i.e., $\stmt\in L$ if and only if there exists a $w$ such that $\mcal{R}_L(\stmt,w)=1$.
An honest-majority multi-CRS NIZK with $m$ CRS contributors for an NP language $L$, parameterized with a security parameter $\secu$, consists of the following algorithms, where part of the definition is taken verbatim from Guo, Pass, and Shi~\cite{guo2019synchronous}.

\begin{itemize}
    \item $\crs \leftarrow {\sf K} (1^{\secu})$: each CRS contributor $j\in[m]$ runs ${\sf K}(1^{\secu})$ to generate a CRS $\crs_j$. \ke{now the user identities also has to be the observer of the broadcast channel}
    \item $\tau \leftarrow {\sf P}(\{\crs_j\}_{j\in[m]}, \stmt, w)$: given a statement $\stmt$ and a witness $w$ such that $\mcal{R}_L(\stmt,w) = 1$, and the set of all CRSes denoted $\{\crs_j\}_{j\in[m]}$, compute a proof denoted $\pi$.
    \item $\{0,1\}\leftarrow {\sf V} (\{\crs_j\}_{j\in[m]}, \stmt, \pi)$: given a statement $\stmt$, the set of all CRSes $\{\crs_j\}_{j\in[m]}$, and a proof $\pi$, 
    the verifier algorithm ${\sf V}$ outputs either $0$ or $1$ denoting either reject or accept.
    \item $(\widetilde{\crs}, \tau) \leftarrow\widetilde{\sf K}(1^{\secu})$: a simulated CRS generation algorithm that generates a simulated $\widetilde{\crs}$ and a trapdoor $\tau$.
    \item $\pi \leftarrow \widetilde{\sf P}(\stmt, \{\widetilde{\crs}_j\}_{j\in[m]}, \{\tau_j\}_{j\in H})$ where $H\subseteq [m]$ and $|H| \geq \lfloor \frac{m}{2}\rfloor + 1$: a simulated prover algorithm produces a proof for the statement $\stmt$ without any witness, and the simulated prover has to have access to at least $\lfloor \frac{m}{2}\rfloor + 1$ number of trapdoors.
\end{itemize}

Henceforth, we use $\mcal{A}^{\mcal{O}(\cdot)}(x)$ to 
mean that $\algA$ is given oracle access to 
the oracle $\mcal{O}(\cdot)$.
\ignore{
denote the output of $\A$ on input $x$ with the oracle access to ${\sf O}(\cdot)$.
We say that an adversary $\mcal{A}$ with the oracle access to the honest key generation oracle ${\sf K}$ is \emph{minority-constrained} if and only if it outputs a set  $\{\crs_j\}_{j\in[m]}$ such that at least $\lfloor\frac{m}{2}\rfloor + 1$ out of $m$ entries are generated by ${\sf K}$ during its interaction with $\mcal{A}$. 
}
Next, we give the security properties we want from the NIZK.

\paragraph{Completeness.} 
Completeness says that an honest prover can always produce a proof
that verifies, if it knows a valid witness to the statement.
Formally, 
completeness requires that for every $\secu$, 
for any set of CRSes $\{\crs_j\}_{j \in [m]}$
where every $\crs_j$ is in the support of ${\sf K}(1^{\secu})$, 
for every 
statement ${\sf stmt}$ and witness $w$
such that $\mcal{R}_L({\sf stmt}, w) = 1$, 
with probability $1$, the following holds:
let 
$\pi \leftarrow {\sf P}(\{\crs_j\}_{j \in [m]},  \stmt, w)$,
it must be that 
${\sf V}(\{\crs_j\}_{j \in [m]}, \stmt, \pi) = 1$.

\ignore{
Completeness requires that with all but negligible probability, the proof generated with a correct witness $w$ will be accepted by the verifier algorithm.
Formally, for any non-uniform probabilistic polynomial time (\ppt) minority-constrained adversary $\A$, there exists a negligible function $\negl(\cdot)$, such that
\begin{align*}
&\Pr\left[
\begin{array}{ll}
(\{\crs_j\}_{j \in [m]}, \stmt, w) \leftarrow \A^{\sf K}(1^\secu),  \  
\pi \leftarrow {\sf P}(\{\crs_j\}_{j \in [m]},  \stmt, w):
&
\begin{array}{ll}
\mcal{R}_L(\stmt,w) = 1
\text{ but } \\[2pt]
 {\sf V}(\{\crs_j\}_{j \in [m]}, \stmt, \pi) = 0
\end{array}
\end{array}
\right]\\
\leq &\negl(\secu).
\end{align*}
}


\paragraph{Zero-knowledge.}
An honest-majority multi-CRS NIZK system satisfies zero knowledge iff the following properties hold.
First, we require that simulated reference strings are indistinguishable from real ones,
i.e., 
for every non-uniform \ppt $\A$, there exists a negligible function $\negl(\cdot)$, such that
\[
\left|\Pr\left[\crs \leftarrow {\sf K}(1^\secu): \A(1^\secu, \crs) = 1\right] 
-
\Pr\left[(\widetilde{\crs}, \tau) \leftarrow 
\widetilde{\sf K}(1^\secu): \A(1^\secu, \widetilde{\crs}) = 1\right]\right| \leq \negl(\secu).
\]
Moreover, we require that 
as long as the majority of the CRSes are honestly generated, then 
any efficient adversary cannot distinguish
an interaction with a real prover
using real witnesses to prove statements  
and an interaction with a simulated prover who proves statements without using witnesses --- 
even if $\A$ obtains the trapdoors of the simulated CRSes.

More formally, let $\A^{\widetilde{\sf K}}$ denote
and adversary $\A$ who is allowed to call the simulated key generation 
algorithm $\widetilde{\sf K}(1^{\secu})$ multiple times.
We say that $\A$ is {\it minority-constrained}, 
if among the set of CRSes $\{\crs_j\}_{j \in [m]}$
output by $\A$, the majority of them are CRSes
returned to \A from ${\sf K}$.
We want that 
for any non-uniform \ppt {\it minority-constrained} adversary $\A$,
there exists a negligible function $\negl(\cdot)$ such that 
\begin{align*}
&\Big| \Pr\left[
(\{\crs_j\}_{j \in [m]}, \stmt, w) \leftarrow \A^{\widetilde{\sf K}}(1^\secu), 
\pi \leftarrow {\sf P}(\{{\crs}_j\}_{j \in [m]}, \stmt, w):
\A(\pi) = 1
\text{ and } \mcal{R}_L(\stmt,w) = 1
\right]\\
-
&
\Pr\left[
(\{{\crs}_j\}_{j \in [m]}, \stmt, w) 
\leftarrow \A^{\widetilde{\sf K}}(1^\secu), 
\pi \leftarrow \widetilde{\sf P}(\{{\crs}_j\}_{j \in [m]}, 
\overrightarrow{\tau}, \stmt):
\A(\pi) = 1 \text{ and } \mcal{R}_L(\stmt,w) = 1
\right]\Big| \\
\leq &\negl(\secu)
\end{align*}
where $\overrightarrow{\tau}$ 
is the following vector: 
for every CRS in the set  $\{{\crs}_j\}_{j \in [m]}$ 
that is output by the simulated key generation algorithm $\widetilde{{\sf K}}$, the vector $\overrightarrow{\tau}$ 
includes its corresponding trapdoor. 
Note that there are at least $\lfloor\frac{m}{2}\rfloor + 1$ entries in $\overrightarrow{\tau}$
since $\A$ is minority-constrained.

\paragraph{Simulation sound extractability.}
Intuitively, simulation sound extractability requires that even though an $\A$
may adaptively interact with a simulated prover
and obtain simulated proofs of false statements, if $\A$ ever produces
a fresh proof for some purposed statement $\stmt$, 
then except with negligible probability, some \ppt extractor must be able to extract a valid witness from the proof, using an extraction key
that is produced during a simulated setup procedure. 

More formally, 
an honest-majority multi-CRS NIZK system satisfies simulation sound extractability
iff there exist
\ppt algorithms $\widetilde{{\sf K}}_0$ and $\mcal{E}$
such that the following is satisfied:
\begin{itemize}[leftmargin=5mm]
\item 
$\widetilde{{\sf K}}_0(1^{\secu})$ outputs a triple denoted $(\widetilde{\crs}, \tau, \ek)$ 
where the first two terms have an output distribution
identical to that of $\widetilde{{\sf K}}(1^{\secu})$; and  
\item 
for any non-uniform \ppt {\it minority-constrained adversary} $\A$, 
there exists a negligible function $\negl(\cdot)$, 
such that the following holds:
\begin{align*}
&\Pr\left[
\begin{array}{ll}
\begin{array}{l}
\left(\{\crs_j\}_{j \in [m]}, \stmt, \pi\right)
\leftarrow \A^{\mcal{O}(1^\secu, \cdot)}: \\
w \leftarrow \mcal{E}(\{\crs_j\}_{j \in [m]}, \overrightarrow{\ek}, \stmt, \pi)
\end{array}
&
\begin{array}{ll}
 (\stmt, \pi)
\text{ not output from $\mcal{O}(1^\secu, \cdot)$, and} \\[2pt]
\text{ $\mcal{R}_L(\stmt, w) = 0$ but ${\sf V}(\{\crs_j\}_{j \in [m]}, \stmt, \pi) = 1$}
\end{array}
\end{array}
\right]\\
\leq &\negl(\secu),
\end{align*}
where $\mcal{O}(1^\secu, \cdot)$ is the following oracle:
\begin{enumerate}[leftmargin=5mm]
\item 
Upon receiving crs generation query ${\tt gen}$ from $\A$, it runs $(\crs, \tau, \ek) \leftarrow \widetilde{{\sf K}}_0(1^\secu)$;
it then records $\tau$ and returns $\crs$ and $\ek$ to $\A$.
\item 
Then, at some point, $\A$ outputs $\{\crs_j\}_{j \in [m]}$ --- this set
of CRSes must be consistent with the CRSes in $\A$'s final output.
$\A$ is required to be minority-constrained, 
meaning that at least $\lfloor\frac{m}{2}\rfloor + 1$ number of entries in $\{\crs_j\}_{j \in [m]}$ must be output from $\widetilde{\sf K}_0$.

At this moment, define the following notation:
\begin{itemize}
\item $\overrightarrow{\tau}$
is the following vector: for every CRS in the set  $\{{\crs}_j\}_{j \in [m]}$
that is output by $\widetilde{{\sf K}}_0$, the vector $\overrightarrow{\tau}$
includes its corresponding trapdoor.
Note that ${\overrightarrow{\tau}}$ must contain at least $\lfloor\frac{m}{2}\rfloor + 1$ such trapdoors
since $\A$ is minority-constrained.
\item 
Similarly, the notation $\overrightarrow{\ek}$
denotes
the following vector: 
for every CRS in the set  $\{{\crs}_j\}_{j \in [m]}$
that is output by $\widetilde{{\sf K}}_0$, the vector $\overrightarrow{\ek}$
includes its corresponding extraction key $\ek$ included in the triple.
\end{itemize}
\item 
At this moment, $\A$ is allowed to send $({\tt prove}, \stmt)$ 
to the oracle multiple times; and for each such invocation, the oracle would
call $\widetilde{\pi} \leftarrow \widetilde{{\sf P}}(\{\crs_j\}_{j \in [m]}, \overrightarrow{\tau}, \stmt)$
and return the resulting $\widetilde{\pi}$ to $\A$.
\end{enumerate}
\end{itemize}
Groth and Ostrovsky~\cite{groth2014cryptography} showed how to construct a multi-CRS NIZK from 
standard cryptographic assumptions, resulting in the following theorem.

\begin{theorem}[Multi-CRS NIZK~\cite{groth2014cryptography}]
Assume the existence of enhanced trapdoor permutations. 
Then, there exists 
a multi-CRS NIZK system that satisfies completeness, zero-knowledge, 
and simulation sound extractability.
\end{theorem}

\subsection{Protocol Description}
\ignore{
Before the miners start the multi-party computation protocol, they have to reach consensus on which user identities are participating. 
This is done by the following protocol.
\begin{mdframed}
\begin{center}
    {\bf Protocol $\Pi_{\rm agreement}$}
\end{center}
\medskip
\noindent {\bf Parameters:} Let $m$ be the number of miners.
\begin{itemize}
    \item Each user sends its (possibly multiple) identities to every miner $j\in[m]$.
    \item Let ${\sf ID}_j$ denote the set of identities miner $j$ received.
    Every miner $j$ broadcasts to all miners the set  ${\sf ID}_j$.
    \item Let ${\sf ID}$ be the set of identities that appear in a majority number of ${\sf ID}_j, j\in[m]$. Every miner outputs ${\sf ID}$.
\end{itemize}
\end{mdframed}
At the end of the agreement protocol $\Pi_{\rm agreement}$, the miners reach consensus on the set ${\sf ID}$ of user identities that participate the protocol.
Note that when the number of colluding miners is less than half, each honest user identity must appear in ${\sf ID}$.

Once the miners agree on the set ${\sf ID}$ of the user identities, they can proceed to the following protocol.
}

Below we give the final multi-party computation protocol $\pmec$.
Roughly speaking, the user identities first secret share
their bids among the miners and prove in zero-knowledge the correctness
of the sharings. 
Then, the miners run an MPC protocol 
using the shares they have received as inputs. 
The MPC protocol will securely compute the 
rules of the auction, and determine which bids are confirmed
and how much each confirmed bid pays.
We will use the honest-majority multi-CRS NIZK defined in~\cref{sec:nizk}.
Moreover, we will describe our protocol $\pmec$ assuming that players have access to an ideal functionality 
$\mcal{F}_{\rm TFM}$ which computes the rules of the auction ---
the formal description of $\mcal{F}_{\rm TFM}$
will be provided at the end of $\pmec$.
The ideal functionality 
$\mcal{F}_{\rm TFM}$ can be realized using 
standard techniques --- in particular, we can use
an MPC protocol that secures against minority corruptions 
providing fairness and guaranteed output~\cite{gmw87,rabin-benor}. 
Finally, our $\pmec$ protocol also 
makes use of a perfectly binding and computationally hiding  
commitment scheme denoted ${\sf comm}$.

\ignore{
Without loss of generality, we assume that each bid is encoded as an element of a finite field $\F$.
We will give the protocol $\pmec$ assuming that players have the access to the ideal functionalities $\zkb$ and $\mcal{F}_{\rm TFM}$ (described at the end of the protocol). 
The ideal functionalities can later be replaced with real-world protocols: $\zkb$ can be replaced by Pass's protocol (\cref{thm:zk-b}) and $\mcal{F}_{\rm TFM}$ can be instantiated using GMW~\cite{gmw87}.
}

During the protocol, miners will keep track of a set $\mcal{C}$ containing the set of user identities who have misbehaved. 
The bids of those in $\mcal{C}$ will be treated as $0$.
All miners have the same view of $\mcal{C}$ since
$\mcal{C}$ is determined using only messages sent on the broadcast channel.

\begin{mdframed}
\begin{center}
{\bf Protocol $\pmec$ instantiating $\fmec$}
\end{center}
\medskip
\noindent {\bf Parameters}: Let $\secu$ be the security parameter.
Let $m$ be the number of miners running the protocol. 
Let $t = \lceil\frac{m}{2}\rceil$ be the reconstruction threshold of secret sharing.
Let ${\sf ID}$ be the agreed-upon set of user identities that are participating in the protocol.
Let $\mcal{C}$ be an initially empty set.
\medskip\\
{\bf Building blocks}: 
\begin{itemize}
    \item Shamir secret sharing.
    \item A perfectly binding, computationally hiding commitment scheme ${\sf comm}$. 
    \item An honest-majority multi-CRS non-interactive zero-knowledge proof (NIZK) system denoted as $\NIZK:=({\sf K}, {\sf P}, {\sf V})$. 
\end{itemize}

\noindent {\bf Input}: Each user identity $i\in{\sf ID}$ has a bid $b_i\in\mathbb{F}$. Each miner has no input. 

\paragraph{Sharing phase} 
\begin{enumerate}
    \item Each miner $j$ runs ${\sf NIZK.K}(1^\secu)$ and obtains $\crs_j$. Each miner $j$ broadcasts $\crs_j$ to all user identities and miners. 
    If a miner $j$ fails to broadcast $\crs_j$, set $\crs_j = {\bf 0}$.
    Let $\CRS:= \{\crs_j\}_{j\in[m]}$.
    \item 
    \label{stp:commit}
    Each user identity $i$ splits $b_i$ 
into $m$ secret shares using a $t$-out-of-$m$ secret sharing scheme.
Let $X_{i,j}$ denote the $j$-th share of $b_i$. Let $\widehat{X}_{i,j} = {\sf comm}(X_{i,j}, r_{i,j})$ where the $r_{i,j}$s are fresh randomness. 
    Broadcast the commitments of shares $\{\widehat{X}_{i,j}\}_{j\in[m]}$ to the miners.
    
    If a user identity $i$ fails to broadcast all the commitments,
    each miner adds $i$ to $\mcal{C}$.
    \item 
    \label{stp:zkb}
Each user identity $i\in{\sf ID}$ calls $\pi_i\leftarrow {\sf NIZK.P}({\sf CRS, stmt}_i, w_i)$ 
with the statement $\stmt_{i} = (i,\{\widehat{X}_{i,j}\}_{j\in[m]})$ and the witness $w_i = (b_i, \allowbreak \{X_{i,j}, r_{i,j}\}_{j\in[m]})$ to prove that 
    \begin{itemize}
        \item For each $j\in[m]$, $(X_{i,j}, r_{i,j})$ is the correct opening of $\widehat{X}_{i,j}$;
        \item $\{X_{i,j}\}_{j\in[m]}$ forms a valid $t$-out-of-$m$ secret sharing of $b_i$.
    \end{itemize}
Each user identity $i$ broadcasts $\pi_i$.

\item For each user identity $i$, if it fails to broadcast $\pi_i$, or ${\sf NIZK.V}(\CRS, \stmt_i,\pi_i)$ outputs $0$, i.e., the verifier algorithm rejects the proof, each miner adds $i$ to $\mcal{C}$.
    
    \item 
    \label{stp:share}
    Each user identity $i\in{\sf ID}$ sends $(X_{i,j},r_{i,j})$ to miner $j$ for all $j\in[m]$.
    
    \item
    \label{stp:complain}
    Each miner $j$ does the following: for all $i\in{\sf ID}\setminus\mcal{C}$, if it receives a message $(X_{i,j},r_{i,j})$ that is a correct opening with respect to $\widehat{X}_{i,j}$, record $(X_{i,j},r_{i,j})$ and broadcast $({\sf ok}, i, j)$. Otherwise, broadcast $({\sf complain}, i, j)$ to complain about user identity $i$.
    
    \item 
    \label{stp:dispute}
    Each user identity $i\in {\sf ID}$ does the following: for all $j$ such that there is a complaint $({\sf complain}, i,j)$ from miner $j$ at Step~\ref{stp:complain}, user identity $i$ broadcasts the corresponding opening $(i,j,X_{i,j},r_{i,j})$. 
Every miner records every correct opening 
$(i,j,X_{i,j},r_{i,j})$ it hears.    

    \item 
If there exists a complaint $({\sf complain}, i,j)$ from miner $j$ in Step~\ref{stp:complain} such that user identity $i$ has not broadcast the correct opening $(i, j, X_{i, j}, r_{i, j})$, each miner adds $i$ to $\mcal{C}$.
    
    
    \end{enumerate}

\paragraph{Computation Phase} Miners invoke $\mcal{F}_{\rm TFM}$ parameterized with ${\sf ID}$, $\mcal{C}$, the commitments of shares $\{\widehat{X}_{i,j}\}_{i\in{\sf ID\setminus\mcal{C}}, j\in[m]}$, and the transaction fee mechanism.
Each miner outputs the output of $\mcal{F}_{\rm TFM}$. 
\medskip

\hrule
\vspace {5pt}

\noindent{\bf Ideal Functionality $\mcal{F}_{\rm TFM}$}
\paragraph{Parameters:} The sets ${\sf ID}$ and $\mcal{C}$, as well as commitments of shares $\{\widehat{X}_{i,j}\}_{i\in{\sf ID}\setminus\mcal{C}, j\in[m]}$ and the transaction fee mechanism.
\paragraph{Input:} Each miner $j$ has input $\{(X_{i,j}, r_{i,j})\}_{i\in {\sf ID}\setminus\mcal{C}}$, where $(X_{i,j}, r_{i,j})$ is a correct opening of $\widehat{X}_{i,j}$.
\paragraph{Functionality:} 
\begin{enumerate}
    \item Each miner sends its input $\{(X_{i,j}, r_{i,j})\}_{i\in {\sf ID}\setminus\mcal{C}}$ to $\mcal{F}_{\rm TFM}$. 
    \item For each $j\in[m]$, the functionality $\mcal{F}_{\rm TFM}$ checks if $(X_{i,j}, r_{i,j})$ is an correct opening of $\widehat{X}_{i,j}$ for all $i\in{\sf ID}\setminus\mcal{C}$. 
    \item For each $i\in{\sf ID}\setminus\mcal{C}$, the functionality reconstructs $b_i$ only using those correct openings.
    If the reconstruction fails, treat $b_i$ as $0$.
    For each $i\in\mcal{C}$, set $b_i = 0$. 
    \item Let $\bids = \{b_i\}_{i\in{\sf ID}}$ denote all the bids. The functionality then computes the output of the transaction fee mechanism on input $\bids$ and sends the output to every miner.
\end{enumerate}
\end{mdframed}

\begin{theorem}
\label{thm:mpc-instantiation}
If the commitment scheme ${\sf comm}$ is perfectly binding and computationally hiding, and the honest-majority multi-CRS \NIZK~satisfies completeness, zero-knowledge and simulation sound extractability,
then $\pmec$ securely realizes $\fmec$ (See~\cref{fig:Ftfm}) in the $\mcal{F}_{\rm TFM}$-hybrid model as long as the number of colluding miners is less than $\frac{m}{2}$. 
\end{theorem}

\subsection{Proof of~\cref{thm:mpc-instantiation}}

Below we use $\equiv$ to denote identically distributed and $\equiv_c$ to denote computationally indistinguishability.
Let ${\sf Exp}_{\mcal{A}}^{\sf Real}$
denote the joint distribution
of the honest parties and the adversary $\A$'s view in the real-world experiment, where 
the adversary $\A$ who controls a subset of the miners and users  
interact with honest parties 
running the real-world protocol $\pmec$.
Let ${\sf Exp}_{\mcal{S}}^{\sf Ideal}$
denote the joint distribution
of the honest parties and the ideal-world adversary $\mcal{S}$'s 
view in the ideal-world experiment, where 
$\mcal{S}$ controls 
the same subset of miners and users, 
and all parties interact with $\fmec$ to compute the outputs.
We want to show that 
${\sf Exp}_{\mcal{A}}^{\sf Real}$ and 
${\sf Exp}_{\mcal{S}}^{\sf Ideal}$
are computationally indistinguishable assuming $\A$ is \ppt.
\ignore{
We will show that for any non-uniform probabilistic polynomial time adversary $\mcal{A}$ interacting with $\pmec$, there exists an adversary $\mcal{S}$ interacting with $\fmec$, such that the joint distribution of honest players' outputs and $\mcal{A}$'s view during an execution of the protocol $\pmec$ (denoted as  ${\sf Exp}_{\mcal{A}}^{\sf Real}$) is computationally indistinguishable from the joint distribution of the honest players' outputs as computed from the functionality $\fmec$ and the view simulated by $\mcal{S}$ (denoted as ${\sf Exp}_{\mcal{S}}^{\sf Ideal}$).
The view of the adversary during an execution contains its private input, its internal randomness, and all the messages it received during the execution.
}
In the proof, we use $\mcal{H}_{\rm miner}$ and $\mcal{K}_{\rm miner}$ to denote the set of honest miners and corrupted miners, respectively. 
Formally, the simulator $\mcal{S}$ interacting with $\fmec$ behaves as follows.

\begin{mdframed}
\begin{center}
    {\bf Simulator $\mcal{S}$ interacting with $\fmec$}
\end{center}
\paragraph{Sharing Phase} 
\begin{enumerate}
    \item Let $\mcal{C}$ be an empty set.
    \item Emulate honest miner $h\in\mcal{H}_{\rm miner}$ as follows: run the simulated CRS generation algorithm $\widetilde{\sf K}_0$ of $\NIZK$ and get a triple $(\crs_h, \tau_h, \ek_h)$. 
    Send $\{\crs_h\}$ to $\mcal{A}$.
    
    At the end of this step, define the following notation: 
    Let $\overrightarrow{\tau}$ be the vector of $\{\tau_h\}_{h\in \mcal{H}_{\rm miner}}$, and $\overrightarrow{\ek}$ be the vector of $\{\ek_h\}_{h\in \mcal{H}_{\rm miner}}$.
    \item For each corrupted miner $k\in \mcal{K}_{\rm miner}$, wait for its $\crs_k$. 
    If a corrupted miner $k$ fails to send $\crs_k$, set $\crs_k = {\bf 0}$. 
    Let $\CRS = \{\crs_j\}_{j\in[m]}$ be the set of all CRSes generated by miners.
    \item Emulate honest user identity $i$ as follows: 
    For every corrupted miner $k\in\mcal{K}_{\rm miner}$, let the share $X_{i,k}$ be a uniformly random element in the finite field $\F$.
    For every honest miner $h\in\mcal{H}_{\rm miner}$, let the share $X_{i,h}=0$.
    \item Emulate honest user identity $i$ as follows: commit to the shares $\widehat{X}_{i,j} = {\sf comm}(X_{i,j}, r_{i,j})$ using fresh randomness $r_{i,j}$ for each miner $j\in[m]$. 
    Send the commitments $\{\widehat{X}_{i,j}\}_{j\in[m]}$ to $\mcal{A}$.
    \item For each corrupted user identity $\ell\in{\sf ID}$, wait for its commitments $\{\widehat{X}_{\ell,j}\}_{j\in[m]}$. If a corrupted user identity $\ell$ fails to send all the commitments, add $\ell$ to set $\mcal{C}$. 

    \item Emulate honest user identity $i$ as follows: call $\pi_i \leftarrow \NIZK.\widetilde{\sf P}(\CRS, \overrightarrow{\tau}, \stmt_i)$, where $\stmt_i := (i,\{\widehat{X}_{i,j}\}_{j\in[m]})$. 
    Send $\pi_i$ to $\A$.
    
    \item For each corrupted user identity $\ell$, wait for $\pi_{\ell}$.
    If a corrupted identity $\ell$ fails to send a proof $\pi_{\ell}$, or that ${\sf NIZK.V}(\CRS, \stmt_{\ell}, \pi_{\ell}) = 0$ for $\stmt_\ell := (\ell,\{\widehat{X}_{\ell,j}\}_{j\in[m]})$, 
    add $\ell$ to $\mcal{C}$.
    
    \item For each corrupted user identity $\ell\in{\sf ID}\setminus\mcal{C}$, the simulator $\mcal{S}$ calls the extraction algorithm $\mcal{E}$ of $\NIZK$ and gets $w_\ell \leftarrow \mcal{E}(\CRS, \overrightarrow{\ek}, \stmt_\ell, \pi_{\ell})$. 
    If there exists an $\ell$ such that $w_{\ell}$ is not a valid witness of $\stmt_{\ell}$, the simulator $\S$ aborts.

         
         \item Emulate each honest identity $i\in{\sf ID}$ to send the shares for each corrupted miners  $\{(X_{i,k},r_{i,k})\}_{k\in\mcal{K}_{\rm miner}}$ to $\A$. 
         \item Receive the shares $\{(X_{\ell,h},r_{\ell,h})\}_{h\in\mcal{H}_{\rm miner}}$ for honest miners  from each corrupted identities $\ell\in{\sf ID}$.
        \item Emulate honest miner $h$ as follows: for each corrupted user identity $\ell\in{\sf ID}$, it checks whether $(X_{\ell,h},r_{\ell,h})$ it received is a correct opening of $\widehat{X}_{\ell,h}$. If yes, send $({\sf ok},h,\ell)$ to $\A$. Otherwise, send $({\sf complain},h,\ell)$ to $\A$.
        Meanwhile, send $({\sf ok},h,i)$ for each honest user identity $i\in{\sf ID}$ to $\A$. 
        \item Emulate honest user identity $i$ as follows: If it received $({\sf complain},k,i)$ from a corrupted miner $k$, send $(i,k,X_{i,k},r_{i,k})$ to $\A$. 
        \item For each corrupted user identity $\ell\in{\sf ID}$, if there exists a complaint $({\sf complain}, h, \ell)$ from an honest miner $h$, wait for $\ell$'s opening $(\ell,h,X_{\ell,h}, r_{\ell,h})$.
        \item For each corrupted user identity $\ell\in{\sf ID}$:
        if there exists a miner $j$ that broadcast a complaint $({\sf complain}, \ell, j)$ but $\ell$ did not broadcast
        the correct opening $(\ell, j, X_{\ell, j}, r_{\ell, j})$, then  add $\ell$ to $\mcal{C}$. 
\end{enumerate}
\paragraph{Computation Phase}
Note that by this point, if the simulator did not abort, for each corrupted user identity $\ell\in {\sf ID}\setminus\mcal{C}$, the simulator $\mcal{S}$ has extracted a valid witness $w_{\ell} = (b_{\ell}$, $\{X_{{\ell},j}, r_{{\ell},j}\}_{j\in[m]})$. 
The simulator sets $b_\ell = 0$ for $\ell \in \mcal{C}$. 
It then sends $b_{\ell}$ for all corrupted user identities $\ell\in{\sf ID}$ to the ideal functionality $\fmec$. 

After the simulator $\mcal{S}$ receives the output from $\fmec$, it sends the output of the mechanism to $\A$ on behalf of $\mcal{F}_{\rm TFM}$.
\end{mdframed}

We construct the following sequence of hybrid experiments.

\medskip\noindent\underline{${\sf Hyb}_0$.} This experiment is identical to a real execution of $\pmec$, except that now the adversary $\A$ interacts with a fictitious simulator $\S'$ which internally emulates the execution of all honest players. Moreover, the simulator $\S'$ also emulates $\mcal{F}_{\rm TFM}$.
We use ${\sf Hyb}_0$ to denote the joint distribution of honest players' outputs and the adversary's view in this experiment.

By definition, ${\sf Exp}_{\mcal{A}}^{\sf Real}\equiv {\sf Hyb}_0$.

\medskip\noindent\underline{${\sf Hyb}_1$.} This experiment is almost identical to the experiment in ${\sf Hyb}_0$, except the following modifications:
\begin{itemize}
    \item Instead of calling $\NIZK.{\sf K}$ to generate the CRS, the simulator $\S'$ calls the simulated CRS generation algorithm $\widetilde{\sf K}_0$, such that for each honest miner $h\in\mcal{H}_{\rm miner}$, the simulator gets $(\widetilde{\crs}_h, \tau_h, \ek_h)$.
    The simulator uses $\widetilde{\crs}_h$ as miner $h$'s NIZK CRS, and keeps the trapdoor $\tau_h$ and extraction key $\ek_h$ to itself.
    \item Whenever the simulator $\S'$ needs to compute a proof on behalf of an honest user identity $i$, it calls the simulated prover algorithm $\widetilde{\sf P}$ supplying the trapdoor $\overrightarrow{\tau}:=\{\tau_h\}_{h\in\mcal{H}_{\rm miner}}$ to compute a simulated proof without using the witness.
\end{itemize}
We use ${\sf Hyb}_1$ to denote the joint distribution of honest players' outputs and the adversary's view in this experiment.

\begin{claim}
Assuming that $\NIZK$ satisfies zero-knowledge, then ${\sf Hyb}_0\equiv_c{\sf Hyb}_1$.
\end{claim}
\begin{proof}
The proof can be done via a sequence of hybrid experiments. First,
one by one for each honest miner, we 
replace the real generation algorithm ${\sf K}$ 
with the simulated generation algorithm $\widetilde{\sf K}$.
Next, one by one for each NIZK proof of an honest user identity, 
we replace the proof
with a simulated proof computed using
$\widetilde{P}$ without using the witness.
Since the number of corrupted miners is less than half, the adversary is minority-constrained (as defined in~\cref{sec:nizk}), the adjacent hybrids in each step are indistinguishable by a straightforward reduction to the zero-knowledge property of $\NIZK$.
\end{proof}

\medskip\noindent\underline{${\sf Hyb}_2$.} This experiment is almost identical to the experiment in ${\sf Hyb}_1$, except that whenever $\A$ supplies a correct \NIZK~ proof $\pi_{\ell}$ on behalf of a corrupted user identity $\ell$ for statement $\stmt_{\ell}$, the simulator $\S'$ calls the \NIZK's extraction algorithm $\mcal{E}(\CRS, \overrightarrow{\ek}, \stmt_{\ell}, \pi_{\ell})$ to extract the witness $w_{\ell}$. 
If $w_{\ell}$ is not a valid witness yet $\NIZK.{\sf V}(\CRS, \stmt_{\ell}, \pi_{\ell}) = 1$, the simulator $\S'$ aborts.
We use ${\sf Hyb}_2$ to denote the joint distribution of honest players' outputs and the adversary's view in this experiment.
\begin{claim}
Assuming that $\NIZK$ satisfies simulation sound extractability, then ${\sf Hyb}_1\equiv_c{\sf Hyb}_2$.
\end{claim}
\begin{proof}
Given that the simulator $\S'$ does not abort, the two experiments are identical.
Since the adversary controls less than half corrupted miners, by the simulation sound extractability property of \NIZK, the probability of $\S'$ aborting in ${\sf Hyb}_2$ is negligible. 
Specifically, for applying the simulation sound extractability, all NIZK statements in the protocol are tagged with the user identity (identity of the prover), thus no statement can be reused.
Therefore, ${\sf Hyb}_1\equiv_c{\sf Hyb}_2$.
\end{proof}

\medskip\noindent\underline{${\sf Hyb}_3$.} This experiment is almost identical to the experiment of ${\sf Hyb}_2$, except for the following difference: 
\begin{itemize}
    \item In the sharing phase, for each honest user identity $i$, instead of committing to the $m$ shares $\{X_{i,j}\}_{j\in[m]}$ of the $t$-out-of-$m$ secret sharing scheme, the simulator $\S'$ commits to $X_{i,k}$ for corrupted miner $k\in\mcal{K}_{\rm miner}$, and commits to $0$ for honest miner $h\in\mcal{H}_{\rm miner}$.
    \item $\S'$ uses the simulated prover algorithm $\widetilde{\sf P}$ of \NIZK~ to vouch for honest user identities.
    \item Upon receiving the openings, it sends $({\sf ok}, h, i)$ for all honest user identities $i\in{\sf ID}$ and all honest miners $h\in\mcal{H}_{\rm miner}$, without actually checking the openings of the commitments.
\end{itemize}
We use ${\sf Hyb}_3$ to denote the joint distribution of honest players' outputs and the adversary's view in this experiment.

\begin{claim}
Assuming that the commitment scheme ${\sf comm}$ is computationally hiding, then ${\sf Hyb}_2\equiv_c{\sf Hyb}_3$.
\end{claim}
\begin{proof}
The proof can be done via a sequence of hybrid experiments, where one by one for each honest user identity $i$, we replace the commitments $\{\widehat{X}_{i,h}\}_{h\in\mcal{H}_{\rm miner}}$ of the shares $X_{i,h}$ with commitments of $0$. 
The adjacent hybrids in each step are indistinguishable by a direct reduction to the computational hiding property of ${\sf comm}$.
\end{proof}

Recall that ${\sf Exp}_{\mcal{S}}^{\sf Ideal}$ denotes the honest players' outputs computed by $\fmec$ and the view simulated by $\mcal{S}$ which interacts with $\fmec$.
\begin{claim}
If the commitment scheme ${\sf comm}$ is perfect binding and that the $t$-out-of-$m$ secret sharing scheme is secure, then ${\sf Hyb}_3 \equiv {\sf Exp}_{\mcal{S}}^{\sf Ideal}$.
\end{claim}
\begin{proof}
The only differences in ${\sf Hyb}_3$ and 
${\sf Exp}_{\mcal{S}}^{\sf Ideal}$
are: 
\begin{enumerate}
\item In ${\sf Exp}_{\mcal{S}}^{\sf Ideal}$, the simulator is generating
honest-to-corrupt shares at random; whereas 
in ${\sf Hyb}_3$, the honest-to-corrupt shares are generated honestly.
By the security of Shamir secret sharing, the two approaches
result in the same distribution since the adversary controls fewer
than $m/2$ miners.
\item 
In ${\sf Hyb}_3$, if the experiment did not abort, 
then the simulator sends
the shares actually opened by corrupt user identities to $\mcal{F}_{\rm TFM}$.
By contrast, in ${\sf Exp}_{\mcal{S}}^{\sf Ideal}$, 
the simulator uses the shares output by the NIZK's extractor
$\mcal{E}$ instead.
Since the commitment is perfectly binding, the two approaches 
result in the same outcome as long as the simulator did not abort.
\end{enumerate}

Therefore, the two hybrids are identically distributed.
\end{proof}
By the hybrid lemma, we have that ${\sf Exp}_{\mcal{A}}^{\sf Real} \equiv_c {\sf Exp}_{\mcal{S}}^{\sf Ideal}$. 
Therefore, the protocol $\pmec$ securely realizes $\fmec$ in the $\mcal{F}_{\rm TFM}$-hybrid model as long as the adversary controls only a minority number of miners.

\subsection{MPC Protocol in the Presence of Majority-Miner Coalitions}
So far, we have focused on instantiating the MPC protocol
when the coalition controls only  
minority of the miners.
As we explained in \Cref{rem:majoritycorrupt}, our game-theoretic
analyses also naturally extend to the case
when the coalition may control
majority of the miners. 

In this case, we can modify our MPC protocol as follows to achieve  
security  
with abort under corrupt majority.
First, instead of threshold secret sharing, the user identities
may use additive secret sharing to share their bids among the miners.
As before, each user identity will broadcast commitments of all shares of its bid, 
and then it gives the corresponding opening to every miner.
There is no more need to prove that the committed values are internally consistent
secret shares. 
If a miner did not receive the correct opening from a user identity,
it can broadcast a complaint in which case the corresponding user identity 
must reveal the correct opening or it will get kicked out.
During the reconstruction phase, if any miner fails to open, then the protocol
just aborts and no output is produced, i.e., no block will be mined.
Finally, $\mcal{F}_{\rm TFM}$
should also be instantiated with a corrupt majority MPC protocol.

\section{Efficient Instantiations of our MPC-Assisted Mechanisms}
\label{sec:efficientMPC}

The MPC-assisted mechanisms proposed in our paper, including
posted price with random selection and the diluted posted price
mechanism, achieve incentive compatibility in the ex post setting.
This means that instantiating these mechanisms in practice
actually does not require the use of generic MPC.
We can use the following efficient protocols: 
\begin{itemize}[leftmargin=5mm]
\item 
Instead of having the   
user identities verifiably secret share their 
bids with the miners, they can simply 
post the bids in the clear over a broadcast channel. 
In practice, we can use any consensus mechanism to realize the 
broadcast channel, such that the miners agree
on the set of all bids posted.
In particular, we can use the underlying blockchain 
itself to reach this consensus --- importantly, if we do this, we stress
that the initial set of bids agreed upon need not be permanently stored
by the blockchain, 
i.e., here we are using the blockchain for (transient) 
consensus but not for storage.
\item 
Once the miners agree on the initial set of bids, they can then
run any coin toss protocol
to decide a randomness seed, which 
can be used to generate the random coins and 
perform the random selection needed
 by the mechanisms.
\end{itemize}

\end{document}